\documentclass[reqno,11pt,nosumlimits]{article}
\usepackage{amssymb,amsfonts}
\usepackage[mathscr]{eucal}
\usepackage{amsmath}
\usepackage{graphics}

\usepackage{graphicx}
\usepackage{color}
\usepackage{ulem}
\renewcommand{\em}{\it}

\allowdisplaybreaks[1]
\numberwithin{equation}{section}

\newtheorem{theorem}{Theorem}[section]

\newtheorem{proposition}[theorem]{Proposition} 

\newtheorem{definition}[theorem]{Definition}

\newtheorem{remark}[theorem]{Remark}
\newtheorem{example}[theorem]{Example}

\newtheorem{convention}{Convention}

\newtheorem{choice}{\rm \bf Choice}

\newenvironment{proof}{{\it Proof. }}{{ \vskip 0.1cm
 \hfill{$\square$}}  \vspace{0.0cm} \medskip}

 \DeclareMathOperator{\Map}{Map}

\DeclareMathOperator{\start}{\bullet}
\DeclareMathOperator{\startpoint}{start}
\DeclareMathOperator{\epoint}{end}

\DeclareMathOperator{\Image}{Image}
\DeclareMathOperator{\arc}{arc}
 \DeclareMathOperator{\Tr}{Tr}

 \DeclareMathOperator{\Mat}{Mat}
\DeclareMathOperator{\supp}{supp}
 \DeclareMathOperator{\Det}{Det}

\DeclareMathOperator{\codim}{codim}

\DeclareMathOperator{\Hom}{Hom}
\DeclareMathOperator{\Ad}{Ad}
\DeclareMathOperator{\ad}{ad}
\DeclareMathOperator{\Aut}{Aut}

 \DeclareMathOperator{\WLO}{WLO}

\DeclareMathOperator{\gleam}{gleam}

\DeclareMathOperator{\sing}{sg}
\DeclareMathOperator{\aff}{aff}
\DeclareMathOperator{\Aff}{Aff}

\DeclareMathOperator{\Hol}{Hol}

\DeclareMathOperator{\mean}{mean}

\newcommand{\N}{{\mathbf N}}

\newcommand{\be}{\begin{equation}}
\newcommand{\ee}{\end{equation}}
\newcommand{\vf}{\varphi}

\newcommand{\bN}{{\mathbb N}}
\newcommand{\bR}{{\mathbb R}}
\newcommand{\bC}{{\mathbb C}}
\newcommand{\bZ}{{\mathbb Z}}

\newcommand{\ct}{{\mathfrak t}}
\newcommand{\cG}{{\mathfrak g}}

\newcommand{\ck}{{\mathfrak k}}
\newcommand{\cg}{c_{\cG}}

\newcommand{\cA}{{\mathcal A}}
\newcommand{\cB}{{\mathcal B}}
\newcommand{\cC}{{\mathcal C}}

\newcommand{\cK}{{\mathcal K}}

\newcommand{\cP}{{\mathcal P}}

\newcommand{\cR}{{\mathcal R}}

\newcommand{\cW}{{\mathcal W}}

\newcommand{\G}{{\mathcal G}}

\newcommand{\face}{\mathfrak F}
\newcommand{\cell}{\mathrm Cell}

\newcommand{\cl}{{\mathrm h}}
\newcommand{\reg}{reg}
\newcommand{\orth}{\perp}

\newcommand{\eps}{\epsilon}

\newcommand{\id}{{ \rm id}}

\begin{document}

\title{From simplicial Chern-Simons theory to the shadow invariant I}

\maketitle

\begin{center} \large
Atle Hahn
\end{center}

\begin{center} \em \large  Grupo de F{\'i}sica Matem{\'a}tica da Universidade de Lisboa\\
Av. Prof. Gama Pinto, 2\\
PT-1649-003 Lisboa, Portugal\\
Email: atle.hahn@gmx.de
  \end{center}

\begin{abstract} This is the first of a series of  papers
in which we introduce and study a rigorous ``simplicial'' realization
of the non-Abelian Chern-Simons path integral for manifolds
$M$ of the form $M=\Sigma \times S^1$ and arbitrary simply-connected
compact structure groups $G$.
More precisely, we will introduce, for general links $L$ in $M$,
 a rigorous simplicial version $\WLO_{rig}(L)$  of  the corresponding
 Wilson loop observable $\WLO(L)$ in the so-called
 ``torus gauge''  by Blau and Thompson (Nucl. Phys. B408(2):345--390, 1993).
For a simple class of links $L$ we then evaluate $\WLO_{rig}(L)$
explicitly in a non-perturbative way,
 finding agreement with Turaev's shadow invariant $|L|$.
\end{abstract}

\medskip


\medskip

\section{Introduction}
\label{intro}

In a celebrated paper, cf. \cite{Wi},   Witten  succeeded  in
defining, on a physical level of rigor,
  a large class of new 3-manifold and link invariants
 by making use of arguments based on
the heuristic Chern-Simons path integral.
   Later, Reshetikhin and  Turaev found a rigorous definition of Witten's  invariants using the representation theory of quantum groups and suitable surgery operations on the base manifold, cf.
\cite{ReTu1,ReTu2} and part I of \cite{turaev}.
 A related approach is the so-called ``shadow world'' approach  by Turaev, cf.
\cite{Tu2} and part II of \cite{turaev}, which also works with quantum group representations
 but eliminates  the use of surgery operations. In fact, the
  shadow invariant of a link is simply given by a finite ``state sum'' (see Appendix B in \cite{Ha7b} for an ``internal'' reference concerning the special case relevant for us).\par

Until now it has not been possible  to find a  rigorous realization of
the path integral expressions appearing in Witten's approach.
 Moreover,  it is not yet clear if/how one can derive the combinatorial/algebraic expressions
in the approach of Reshetikhin and  Turaev directly from Witten's path integral expressions.
These two open problems (``Problem 1'' and ``Problem 2'') are important problems by themselves, cf. \cite{Kup,Saw99}.
Moreover, their solution  would
greatly simplify the task of relating (in a rigorous way)
the  Reshetikhin-Turaev invariants to classical topology and geometry.
The latter problem (``Problem 3'') is one of the central open problems in the field of  3-manifold quantum topology,
 cf.  \cite{Oht}. \par

The aim of the the present paper is to make some progress towards the solution
of Problems 1--3. Our strategy will be to restrict our attention  to the special case
where the base manifold is of the form $M=\Sigma \times S^1$ and to apply
the so-called ``torus gauge fixing'' procedure, which was
introduced by Blau and Thompson in \cite{BlTh1} for the study of
Chern-Simons models on such manifolds.\par
 The results in \cite{Ha4,Ha6,HaHa}, which were obtained by extending and
combining the work in \cite{BlTh1,BlTh2,BlTh3,Ha3b,Ha3c} and \cite{ASen,Ha2},
 suggest that for such manifolds it should  be
possible to find a rigorous realization of the heuristic path
integral expressions for the Wilson loop observables (WLOs) which
appear after torus gauge fixing has been applied, cf. Eq.
\eqref{eq2.48} below.\par
 Let us remark, however, that the ``continuum'' approach in
\cite{Ha4,Ha6,HaHa} is quite technical and its complete implementation will be rather laborious.
Moreover, it is not clear
how useful the continuum approach will be regarding our aforementioned goal of relating
 the Reshetikhin-Turaev invariants to the classical topological and geometric invariants (= ``Problem 3'').\par

 In the present paper we will therefore propose  an alternative approach,
 which  was inspired by Adams' ``simplicial''  framework (see \cite{Ad0,Ad1})
 for Abelian Chern-Simons models\footnote{or, rather, Abelian $BF_3$-models, cf. the beginning of Sec. \ref{sec4} below}.  The new approach  is  elementary and does not have the other drawbacks of the continuum approach mentioned above.

 \smallskip

\noindent The present paper is organized as follows: \par

In  Sec.~\ref{sec2} we  describe  the torus gauge fixing
 approach of \cite{BlTh1,BlTh2,BlTh3,Ha3b,Ha3c,Ha4}
to  Chern-Simons models on manifolds $M$ of the form $M = \Sigma \times S^1$.
Most of this section is a summary of the exposition in \cite{Ha3c,Ha4},
but Sec. \ref{subsec2.3} contains some new material (e.g the two heuristic formulas Eq. \eqref{eq2.41simpl} and  Eq. \eqref{eq2.48}). \par

  In Sec.~\ref{sec3}, which can be considered as a continuation of the introduction,
  we give some background information regarding the general simplicial program
  for CS-theory/$BF_3$-theory. The simplicial program is very closely
   related to the open problems mentioned above. In fact, a successful implementation
   of the general simplicial program would immediately resolve Problem 1 and Problem 2
   and it is plausible to expect that this would also lead to some progress regarding
    Problem 3, see Sec.~\ref{subsec3.2}. \par

  In   Sec.~\ref{sec4.0} and in Sec.~\ref{sec4} we then return to the special situation which is relevant in the
  present paper and describe our approach for the discretization of the  heuristic equation Eq. \eqref{eq2.48}.\par

Finally, in Sec.~\ref{sec6} we present our main result, Theorem \ref{main_theorem}
(which will be proven in \cite{Ha7b}), before
we conclude the main part of the paper with a short
outlook in Sec.~\ref{sec7}.

\smallskip

The present paper has an appendix consisting of four parts.
In part \ref{appB} we list the Lie theoretic notation used in the present paper
and we give some explicit formulas in the special case $G=SU(2)$.
In part \ref{appB'} we fill in some   details  which were omitted in Sec. \ref{sec2}.
In part \ref{appF} we give a formal treatment of  the notion
of a  ``polyhedral cell complex''.  In part \ref{appG} we make some more comments regarding
the aforementioned simplicial program for CS-theory/$BF_3$-theory.

\section{Chern-Simons theory on $M=\Sigma \times S^1$ in the torus gauge}
\label{sec2}

\subsection{Chern-Simons theory}
\label{subsec2.1}

 Let us fix a  simply-connected compact Lie group\footnote{cf. part \ref{appB} of the Appendix for concrete
 formulas in the special case $G=SU(2)$}  $G$
 with Lie algebra $\cG$. \par

For every smooth manifold $M$, every real vector space $V$
and every  $n \in \bN_0$ we  will denote by $\Omega^n(M,V)$
  the space of $V$-valued n-forms on $M$
  and  we set
\begin{equation} \label{eq2.1} \cA_{M,V}:=\Omega^1(M,V), \quad \quad \cA_{M}:=\cA_{M,\cG}
\end{equation}
 By $\G_M$ we will denote the ``gauge group'' $C^{\infty}(M,G)$.
    We will usually write $\cA$ instead of $\cA_M = \Omega^1(M,\cG)$ and $\G$ instead of $\G_M$.

  \smallskip

 In the following  we will restrict ourselves to the special case where $M$ is an oriented closed 3-manifold.
Moreover,  we will  consider only
 the special case where $G$ is  simple (cf. Remark \ref{rm2.1}
 below for the case of general simply-connected compact Lie groups).
  The  Chern-Simons  action function  $S_{CS}:= S_{CS}(M,G,k)$
      associated to $M$, $G$,  and the ``level'' $k \in \bZ \backslash \{0\}$ is  given by
  \begin{equation} \label{eq2.2'} S_{CS}(A) = - k \pi \int_M \langle A \wedge dA \rangle
   + \tfrac{1}{3} \langle A\wedge [A \wedge A]\rangle \in \bR, \end{equation}
for all $ A \in \cA$. Above
 $[\cdot \wedge \cdot]$  denotes the wedge  product associated to the
bilinear map $[\cdot,\cdot] : \cG \times \cG \to \cG$
and where   $\langle \cdot \wedge  \cdot \rangle$
  denotes the wedge product  associated to the
  suitably\footnote{More precisely, the normalization is chosen such that
 $\langle \Check{\alpha},\Check{\alpha}
\rangle = 2$ for every short real coroot $\Check{\alpha}$ w.r.t. any fixed Cartan subalgebra of $\cG$.
Observe that  after making  the identification $ \ct \cong \ct^*$ which is induced by $\langle
\cdot,\cdot \rangle$ we have $\langle \alpha,\alpha \rangle = 2$ for every {\em long} root $\alpha$.
  Thus the normalization here   coincides with the one in \cite{Roz}.
   This  normalization
 guarantees that the exponential $\exp(i S_{CS})$ is ``gauge invariant'', i.e. invariant
under the standard $\G$-operation on $\cA$} normalized
  Killing form $\langle \cdot , \cdot \rangle : \cG \times \cG \to \bR$.

\medskip

From the definition of $S_{CS}$ it is obvious that $S_{CS}$ is
invariant under (orientation-preserving) diffeomorphisms. Thus, at a
heuristic level, we can expect that the heuristic integral (the
``partition function'') $Z(M) := \int  \exp(i S_{CS}(A)) DA$ is a
topological invariant of the  3-manifold $M$. Here $DA$ denotes  the
informal ``Lebesgue measure'' on the space $\cA$.\par
A similar statement holds if we consider
 (oriented and ordered) links $L$ in $M$,
i.e. finite tuples  $L= (l_1, l_2, \ldots, l_m)$, $m \in \bN$,
where each $l_i$ is a knot\footnote{i.e.  a smooth embedding $S^1 \to M$} in $M$
such that $\arc(l_i) \cap  \arc(l_{i'}) = \emptyset$ holds whenever $i \neq i'$.
In the following we will identify each knot $l_i:S^1 \to M$
with the loop $[0,1] \ni t \mapsto l_i(i_{S^1}(t)) \in M$
where $i_{S^1}: [0,1]  \ni s \mapsto  \exp(2\pi i s) \in U(1) \cong S^1$.\par
If we fix a finite tuple
  $\rho = (\rho_1,\rho_2,\ldots,\rho_m)$ of
irreducible, finite-dimensional, complex representations (=``colors'') of $G$
then we can expect  at a heuristic level that  the mapping which maps every
 link $L= (l_1, l_2, \ldots, l_m)$ in $M$
 to the heuristic integral (the ``expectation value of the Wilson loop observable
  associated to $L$ and $\rho$'')
\begin{equation} \label{eq2.4}
\WLO(L,\rho)  :=  \int_{\cA} \prod_i
\Tr_{\rho_i}\bigl(\Hol_{l_i}(A)\bigr) \exp(i S_{CS}(A)) DA
\end{equation}
is a link invariant. Here   $\Tr_{\rho_i}$ denotes the trace in the  representation $\rho_i$
and $\Hol_{l_i}(A)$  denotes the   holonomy of $A$ around the loop $l_i$.
Among the  many different ways of writing  $\Hol_{l_i}(A)$ explicitly
 the following equation
  will be particularly convenient  for our purposes (cf.  Sec. \ref{subsec4.3} below):
\begin{equation} \label{eq2.5}
\Hol_{l_i}(A) = \lim_{n \to \infty} \prod_{j=1}^n \exp(\tfrac{1}{n}
A(l'_i(\tfrac{j}{n})))
\end{equation}
Here $\exp:\cG \to G$ is the exponential map of $G$.

\begin{remark} \rm
We will simply write  $\WLO(L)$  instead of  $\WLO(L,\rho)$ if no confusion
about the tuple $\rho = (\rho_1,\rho_2,\ldots,\rho_m)$ of ``colors'' can arise.

\smallskip

We remark that in the standard physics literature the notation $Z(M,L)$
is normally used instead of  $\WLO(L)$.
\end{remark}

For convenience we will assume  (without loss of generality) in the following
that the  Lie group $G$  fixed above
 is a Lie subgroup of $U(\N)$, $\N \in \bN$. The Lie algebra $\cG$ of $G$ can
then be identified with the obvious Lie subalgebra of
 the Lie algebra $u(\N)$ of $U(\N)$ and we have
  \begin{equation} \label{eq_Tr_scalar}
 \langle A,B \rangle  =   - \Tr(A \cdot B)
\quad \forall A,B \in \cG
\end{equation}
 where ``$\cdot$'' denotes the matrix
 multiplication in $\Mat(\N,\bC)$ and
where $\Tr := c \Tr_{\Mat(\N,\bC)}$ for suitably chosen\footnote{observe that if  $G$ is simple
then every $\Ad$-invariant scalar product on $\cG$ is
proportional to the Killing form} $c \in \bR$.
For example, in the special case $G=SU(\N)$ we have $c = \tfrac{1}{4 \pi^2}$.\par

The Chern-Simons  action function  $S_{CS}$
  can  then  be rewritten as
\begin{equation} \label{eq2.2} S_{CS}(A) = k\pi \int_M \Tr(A \wedge dA + \tfrac{2}{3} A\wedge A\wedge A), \quad A \in \cA
\end{equation}
where  ``$\wedge$'' is now the wedge product
 for $(\Mat(\N,\bC), \cdot)$-valued forms. Moreover, on  the RHS
  of Eq. \eqref{eq2.5}
 we can then reinterpret  $\prod \cdots$ as the
 matrix product  and $\exp$  as the exponential map of $\Mat(\N,\bC)$.

\begin{remark} \label{rm2.1}  \rm
 Observe that a simply-connected compact Lie group is automatically semi-simple
and can therefore be written  as a  product
 of the form $G = \prod_{i=1}^{r} G_i$, $r \in
\bN$, where each $G_i$ is a simple simply-connected compact Lie
group. We can generalize the definition of $S_{CS}$  to this general situation
by setting  $S_{CS}(M,G,k)(A) := \sum_i S_{CS}(M,G_i,k)(A_i)$ for all $A \in \cA$
where $(A_i)_i$ are the components of $A$ w.r.t. to the decomposition
 $\cG = \oplus_i \cG_i$ ($\cG_i$ being the Lie algebra of $G_i$).\par
In view of Sec. 7 in \cite{Ha7b}  let us generalize the definition of $S_{CS}$ even further
and introduce for every sequence $(k_i)_{i \le r}$ of non-zero integers
the function  $S_{CS}(M,G,(k_i)_i)$by setting
$S_{CS}(M,G,(k_i)_i)(A) := \sum_i S_{CS}(M,G_i,k_i)(A_i)$ for all $A \in \cA$.
In fact,  in the present paper and in \cite{Ha7b} only
two special cases will play a role, namely the case
 $r=1$ (i.e. $G$ simple) and the case $r=2$, $G_2 = G_1$ and $k_2
= - k_1$, cf. Sec. 7 in \cite{Ha7b}.
\end{remark}

\subsection{Torus gauge fixing}
\label{subsec2.2}

For the rest of this paper let us fix a  maximal torus $T$
of $G$. The Lie algebra  of $T$ will be denoted by  $\ct$.
\subsubsection{Motivation}
\label{subsubsec2.2.1}

In order to motivate the definition
of the torus gauge fixing procedure for the manifold $M$ of the form
$M=\Sigma \times S^1$ where $\Sigma$ is a connected surface let us first have a quick look at
the orbit space $\cA / \G$ and the canonical projection $\pi_{\G}: \cA \to \cA/ \G$
for the three manifolds $M=\bR$, $M=S^1$, and $M=\Sigma \times \bR$.\par

In the following
  $\frac{\partial}{\partial t}$ will denote the vector field
  on $\bR$ (resp. $S^1$) which is induced by the map $\id_{\bR}$ (resp. the map
  $i_{S^1}: [0,1]  \ni s \mapsto  \exp(2\pi i s) \in U(1) \cong S^1$)
  and $dt$ will denote the dual 1-form on $\bR$ (resp. $S^1$).
  The obvious ``lift''/pullback of $\frac{\partial}{\partial t}$ and $dt$
 to the product manifolds $\Sigma \times \bR$ and $\Sigma \times S^1$
will again be denoted by $\frac{\partial}{\partial t}$ and $dt$.

\begin{enumerate}
 \item $M=\bR$: Here every 1-form $A = A_0 dt \in \cA$
is gauge-equivalent to the trivial 1-form  $0 \, dt = 0$, so  $\cA /
\G$ has just one element.

\item $M= S^1$: Here
every $A \in \cA$ is gauge equivalent\footnote{this follows, e.g.,
by looking at the explicit form of the well-known injection
$\psi:\cA_{flat}/ \G \to \Hom(\pi_1(M),G)/G$ and taking into account
that in the special case $M=S^1$ we have $\cA/ \G  = \cA_{flat}/ \G$
and $ \Hom(\pi_1(M),G)/G \cong G/G$ } to a 1-form of the form $Bdt$
with constant  $B:S^1 \to \cG$. Moreover, according to the
fundamental theorem of maximal tori we can choose $B$ to be
$\ct$-valued, so the map $\pi_{\G}:\{B dt \mid B \in
C^{\infty}(S^1,\ct)
 \text{ is constant }\} \to \cA/ \G$ is surjective.

\item $M = \Sigma \times \bR$:
 Every $A \in \cA$ can be written uniquely in the form
 $A = A^{\orth} + A_0 dt $
with  $A^{\orth} \in \cA^{\orth} := \{A \in \cA \mid A(\tfrac{\partial}{\partial t}) =0\}$
and $A_0 \in C^{\infty}(M,\cG)$.
Using the argument for the case $M=\bR$
for each of the ``fibers'' $\{\sigma\} \times \bR \cong \bR$, $\sigma \in \Sigma$,
  we can easily conclude that
every 1-form $A$ can be  gauge-transformed
 into an element of $\cA^{\orth}$.
In other words, the map $\pi_{\G}: \cA^{\orth} \to \cA / \G$ is surjective.
\end{enumerate}
After these preparations let us now go back to the original
manifold.
\begin{enumerate}
\item[iv)]  $M = \Sigma \times S^1$: Again every  $A \in \cA$ can be written uniquely in the form
  $A = A^{\orth} + A_0 dt $ with $A^{\orth} \in \cA^{\orth}$ and $A_0 \in C^{\infty}(M,\cG)$
  where $\cA^{\orth}$ is defined again by
 \begin{equation} \label{eq2.6}
 \cA^{\orth} := \{A \in \cA \mid A(\tfrac{\partial}{\partial t}) =0 \}
 \end{equation}

Combining the results for the case $M=S^1$ and $M = \Sigma \times \bR$
and making the identification
$  \{ B \in C^{\infty}(\Sigma \times S^1,\ct) \mid
  B(\sigma,\cdot) \text{ constant for all $\sigma \in \Sigma$}
 \} \cong C^{\infty}(\Sigma,\ct)$
one is naturally led to the space
\begin{align} \label{eq2.7}
   \cA^{qax}(T) & := \cA^{\orth} \oplus  \{ B dt \mid B \in C^{\infty}(\Sigma,\ct)\}
 \end{align}
 and to the question whether the map
$\pi_{\G}:  \cA^{qax}(T) \to \cA / \G$
is surjective.
For  technical reasons let us also introduce the space
\begin{align} \label{eq2.8}
 \cA^{qax} & := \cA^{\orth} \oplus  \{ B dt \mid B \in C^{\infty}(\Sigma,\cG)\}
\end{align}
\end{enumerate}

In Sec. \ref{subsubsec2.2.3} and Sec. \ref{subsubsec2.2.4} below we will  study the map $\pi_{\G}:
\cA^{qax}(T) \to \cA / \G$ in the situation relevant for us. Before
we do this we will have to make a short digression where we
introduce the two heuristic concepts of a generalized  gauge fixing
and an abstract gauge fixing which will be useful in Sec.
\ref{subsubsec2.2.3} and Sec. \ref{subsubsec2.2.4} below.

\subsubsection{Two heuristic concepts: ``generalized'' and ``abstract'' gauge fixing}
\label{subsubsec2.2.2}

Let us call a (not necessarily linear)
subspace $V$ of $\cA$
 a {\em gauge fixing subspace}
 iff its elements form a complete and minimal
set of representatives of $\cA/\G$, or, equivalently, iff the map
$\Pi_{V}:V \to \cA/ \G$ which is obtained by restricting the
canonical projection $\pi_{\G}: \cA \to \cA/ \G$ onto $V$ is a
bijection. (Thus  $\Pi_{V}^{-1}: \cA/\G \to \cA$
 will be a {\em gauge fixing} in the usual sense).\par
Let $ d\nu_{\cA / \G}$ denote the image of the informal Lebesgue measure $DA$
under  $\pi_{\G}$, i.e.
\begin{equation} \label{eq2.9}
 d\nu_{\cA / \G}  := (\pi_{\G})_* DA
\end{equation}
If $V$ is a gauge fixing subspace then, setting $d\mu_V:= (\Pi_V)^{-1}_* (d\nu_{\cA / \G})$,
 we trivially have (at an informal level)
\begin{equation} \label{eq2.10}
(\Pi_V)_{*} d\mu_V = d\nu_{\cA / \G}
\end{equation}
and therefore\footnote{since
$\int_{\cA} \chi(A) DA =
\int_{\cA/ \G} \bar{\chi} \ d\nu_{\cA / \G} =
\int_{\cA/ \G} \bar{\chi} \ (\Pi_V)_* d\mu_V
= \int_V   (\bar{\chi} \circ \Pi_V) d\mu_V = \int_V \chi(A) d\mu_V(A)$
where $\bar{\chi} : \cA/\G \to \bC$ is uniquely given by
$\chi = \bar{\chi} \circ \pi_{\G}$}
also  (informally)
\begin{equation} \label{eq2.11}
\int \chi(A) DA =  \int_V \chi(A) d\mu_V(A)
\end{equation}
for every $\G$-invariant   function $\chi:\cA \to \bC$.

\begin{remark} \label{rm_FadPop} \rm If $V$ is a ``sufficiently nice'' subspace of $\cA$
the informal measure $d\mu_V$ will have an explicitly computable
 ``density'' w.r.t. to $DA_{| V}$.
This density is usually called the ``Faddeev-Popov determinant'' (of
the gauge fixing associated to $V$), cf., e.g., \cite{Pok} and Appendix C in \cite{Ha4}.
\end{remark}

We will call a subset
$V \subset \cA$ a {\em generalized gauge fixing subspace}
iff the map
$\Pi_{V}:V \to \cA/ \G$ given as above is ``essentially surjective''
in the sense that the complement of $\pi_{\G}(V)$
in  $\cA/ \G$ is a $ d\nu_{\cA / \G}$-zero subset.
Since for such a generalized gauge fixing space the map $\Pi_V:V \to \cA / \G$
 need not be  injective we can not hope to obtain
an informal measure $d\mu_V$ fulfilling Eq. \eqref{eq2.10}  above in
a canonical way as above. However, since by assumption   $\Pi_V$ is
essentially surjective there will be (non-canonical informal)
measures   $d\mu_V$  on $V$ such that Eq. \eqref{eq2.10} above is
fulfilled at an informal level. We will call any  such\footnote{i.e.
$V$ is a generalized gauge fixing space and  $d\mu_V$ a  measure on
$V$ fulfilling Eq. \eqref{eq2.10} } pair $(V,d\mu_V)$
 a  {\em generalized gauge fixing}.\par

An {\em abstract  gauge fixing} is a triple $(V,\Pi_V,d\mu_V)$ where
$V$ is an arbitrary (heuristic) measurable space, $\Pi_V$ is a
(heuristic) measurable map $V \to \cA/ \G$ and $d\mu_V$ a measure on
$V$ such that Eq. \eqref{eq2.10}  above  is fulfilled\footnote{in this case, one can conclude
at an informal level that $\Pi_V$ should then be essentially
surjective in the sense above since one can argue that the informal
measure
 $d\nu_{\cA / \G}$ on $\cA / \G$   should have ``full support''}.
 Of course, for an abstract gauge fixing
 to be useful the space $V$ should have
 enough structure so that one can perform concrete computations\footnote{this excludes, for example,
 the ``trivial''  abstract  gauge fixing
 $V=\cA / \G$, $\Pi_V = \id_{\cA / \G}$ and $d\mu_V = d\nu_{\cA / \G}$.
  The orbit space  $\cA / \G$ has so little structure
that it is not very useful for explicit computations, which is
exactly the reason why one usually  tries to apply a suitable gauge
fixing (in the usual or generalized sense)}. For the abstract gauge
fixing which will appear in Sec. \ref{subsubsec2.2.4}
 the space $V$ will be the  product of a linear space and a countable discrete space.

\subsubsection{Torus gauge fixing for non-compact $M=\Sigma \times S^1$}
\label{subsubsec2.2.3}

Recall from \cite{Br_tD} that an element $g$ of $G$ is called ``regular'' iff it is contained
in exactly one maximal torus of $G$. In the following we set
\begin{equation}G_{reg}:= \{g \in G \mid g \text{ is regular}\}
\end{equation}

\medskip

Let us now restrict to the special situation
which is relevant for CS theory, i.e. where  $G$ is simply-connected, cf. Sec. \ref{subsec2.1}.
It turns out  that if $\Sigma$ is non-compact we have (cf. point ii) in part \ref{appB'_1} of the Appendix)
 \begin{equation} \label{eq2.13} \cA_{reg} / \G \subset \cA^{qax}(T)/ \G = \pi_{\G}(\cA^{qax}(T))
\end{equation}
 where
\begin{equation} \label{eq2.14}\cA_{reg} := \{ A \in \cA \mid \Hol_{l_{\sigma}}(A) \in G_{reg} \text{ for every }
\sigma \in \Sigma \}
\end{equation}
and where $l_{\sigma}$ is the ``vertical'' loop ``above''  $\sigma$,
i.e. $l_{\sigma}: [0,1] \ni s \mapsto (\sigma,i_{S^1}(s)) \in M$
with $i_{S^1}$  as above. Moreover, since $\codim(G \backslash
G_{reg}) \ge 3$ (cf., e.g., Chap. V in \cite{Br_tD}) and
$\dim(\Sigma)=2$, a ``generic'' function $f:\Sigma \to G$ will
remain inside $G_{reg}$. Thus we can argue at a heuristic level that
the difference sets
$$ C^{\infty}(\Sigma,G) \backslash C^{\infty}(\Sigma,G_{reg}),
\quad \cA  \backslash \cA_{reg}, \quad (\cA / \G)  \backslash
(\cA_{reg}/\G)$$ are all  ``negligible'', i.e. ``zero-subsets''
w.r.t. $D\Omega$ (the heuristic Haar measure on $
C^{\infty}(\Sigma,G)$) resp. $DA$ resp.  $ d\nu_{\cA / \G}$.
 Thus,  if $\Sigma$ is non-compact, the space $V:= \cA^{qax}(T)$ should  indeed  be a
generalized gauge fixing space and there should be a measures $d\mu_V$ on  $V$ fulfilling
Eq. \eqref{eq2.10}.\par

In order to describe such a measure $d\mu_V$ explicitly,
let us make the identification
\begin{equation} \label{eq2.15} V = \cA^{qax}(T) \cong A^{\orth} \times \cB
\end{equation}
where we have set
\begin{equation} \label{eq2.16} \cB:=C^{\infty}(\Sigma,\ct)
\end{equation}
Below $DA^{\orth}$ will denote the (informal) ``Lebesgue measure'' on
$\cA^{\orth}$ and $DB$ the (informal) ``Lebesgue measure'' on $\cB$.\par

It can be shown at an informal level\footnote{by computing the  Faddeev-Popov determinant
of a closely related (proper) gauge fixing,
 cf.  Sec. 2.3 and Sec. 2.4 in \cite{Ha4} (and Remark \ref{rm_FadPop} above).
Observe that Eq. \eqref{eq2.18} below is just the analogue of Eq. (2.23) in
Sec. 2.4  in \cite{Ha4} where $P$ is replaced by $\ct$}
  that the measure $d\mu_V$ on  $V \cong
A^{\orth} \times \cB$ given by
\begin{equation} \label{eq2.17}
 d\mu_V :=       DA^{\orth}  \otimes \bigl(\det\bigl(1_{\ck}-\exp(\ad(B))_{| \ck}\bigr) DB\bigr)
\end{equation}
fulfills Eq. \eqref{eq2.10} up to a multiplicative constant, or, equivalently, that
\begin{equation} \label{eq2.18}
 \int_{\cA} \chi(A) DA
 \sim  \int_{\cB} \biggl[ \int_{\cA^{\orth}} \chi(A^{\orth} + B dt)
   DA^{\orth}  \biggr]   \det\bigl(1_{\ck}-\exp(\ad(B))_{| \ck}\bigr) DB
  \end{equation}
holds for every $\G$-invariant function $\chi:\cA \to \bC$.
Here  $\ck$ is the  $\langle \cdot,\cdot \rangle$-orthogonal
complement of $\ct$ in $\cG$,   $1_{\ck}$ denotes the identity operator on $C^{\infty}(\Sigma,\ck)$,
and $\exp(\ad(B))_{| \ck}$ is the well-defined\footnote{observe that
$\exp(\ad(B(\sigma))) = \Ad(\exp(B(\sigma)))$ and that $\ck$ is an $\Ad_{| T}$-invariant
subspace of $\cG$} linear operator on $C^{\infty}(\Sigma,\ck)$
given by $(\exp(\ad(B))_{| \ck} \cdot f)(\sigma) = \exp(\ad(B(\sigma)))_{| \ck} \cdot f(\sigma)$
for all $f \in C^{\infty}(\Sigma,\ck)$ and $\sigma \in \Sigma$.

  \begin{convention} \rm Above and in the following ``$\sim$''   denotes
  equality up to a multiplicative constant.
  Of course, this  ``constant'' can/should depend on $G$ and $M$ but
  it is  independent of  $\chi$.
 \end{convention}

\subsubsection{Torus gauge fixing for compact $M=\Sigma \times S^1$ }
\label{subsubsec2.2.4}

The case of compact  $\Sigma$, which is the case we are actually
interested in, requires more care since in this case we   have
$\cA_{reg} / \G \not\subset \pi_{\G}(\cA^{qax}(T))$  (cf. point iii) in part \ref{appB'_1} of the Appendix; cf. also Example 2 in Sec. 3 in \cite{BlTh3}).
 If one still wants to transform a general 1-form $A \in \cA_{reg}$
 into an element $A^{\orth}+ Bdt$
 of $\cA^{qax}(T)$ one can do so only
 if  one uses a certain (mildly)
 singular gauge transformation $\Omega$
 and also allows  $A^{\orth}$ to have a similar  singularity
 (cf. the maps $\Omega_{\cl}$ and the 1-forms $A_{\sing}(\cl)$ appearing in Step 2 and Step 3 below).
For functions $\chi$ fulfilling an additional property (cf. Step 2
below) this strategy indeed allows us to generalize Eq.
\eqref{eq2.18} so that also the case of compact surfaces $\Sigma$ is
included, cf. Eq \eqref{eq2.27} below.\par
 We will now sketch the derivation of Eq \eqref{eq2.27}.
 For more details we refer to \cite{Ha3c}
where a detailed derivation of a very similar
equation is given where $\cB = C^{\infty}(\Sigma,\ct)$
is replaced by $C^{\infty}(\Sigma,P)$, $P \subset \ct$ being a fixed Weyl alcove.

\medskip

\noindent {\bf Preparation for Step 1:} Let us
set\footnote{\label{ft11} the notation $[\Sigma,G/T]$ is motivated
by the fact that $C^{\infty}(\Sigma,G/T)/\G_{\Sigma}$ coincides with
the set of  homotopy classes of maps $\Sigma \to G/T$, cf.
Proposition 3.2 in \cite{Ha3c}}
\begin{equation} \label{eq2.19}
[\Sigma,G/T] := C^{\infty}(\Sigma,G/T)/\G_{\Sigma}
\end{equation}
where the expression on the RHS  denotes the orbit space of the
$\G_{\Sigma}$-operation on $C^{\infty}(\Sigma,G/T)$ given by
$\bar{g} \cdot \Omega := \Omega^{-1}\bar{g}$ for $\bar{g} \in
C^{\infty}(\Sigma,G/T)$ and $\Omega \in \G_{\Sigma}$.
 For each $\cl \in
[\Sigma,G/T]$ we  pick a representative $\bar{g}_{\cl} \in
C^{\infty}(\Sigma,G/T)$. We will keep each $\bar{g}_{\cl}$ fixed in
the following.\par For each  $\bar{g} \in G/T$ and $b \in \ct$  we
set $\bar{g} b \bar{g}^{-1} := g b g^{-1} \in \cG$ where $g$ is an
arbitrary\footnote{observe that for $b \in \ct$ the value $\bar{g} b
\bar{g}^{-1}$ will not depend on the special choice of $g$} element
of $G$ fulfilling $gT = \bar{g}$ (and where we use the somewhat sloppy notation
$g b g^{-1}$ instead of $\Ad(g) \cdot b$).

\medskip
 \noindent {\bf Step 1:} We now introduce  a suitable abstract
gauge fixing  $(\overline{V},
\Pi_{\overline{V}},d\mu_{\overline{V}})$ (``abstract torus gauge
fixing''). We take
\begin{equation} \label{eq2.20} \overline{V} := A^{\orth} \times
\cB \times [\Sigma,G/T]
\end{equation} and  define $\Pi_{\overline{V}} :
\overline{V} \to \cA/ \G$  by
\begin{equation} \label{eq2.21} \Pi_{\overline{V}} (A^{\orth},B,\cl)  =
\pi_{\G}( A^{\orth} + \bar{g}_{\cl}B \bar{g}_{\cl}^{-1} dt)
\end{equation}
for all $(A^{\orth},B,\cl) \in  \overline{V}$ where $\bar{g}_{\cl}B
\bar{g}_{\cl}^{-1} \in C^{\infty}(\Sigma,\cG)$ is given by
$(\bar{g}_{\cl}B \bar{g}_{\cl}^{-1})(\sigma):= \bar{g}_{\cl}(\sigma)
B(\sigma) \bar{g}_{\cl}^{-1}(\sigma)$ for all $\sigma \in \Sigma$.

\smallskip

We will motivate this choice of  $\overline{V}$ and of $\Pi_{\overline{V}}$
in part \ref{appB'_1} of the Appendix.
There we will also show (cf. Eq. \eqref{eq2.22_app} below) that
\begin{equation} \label{eq2.22} \cA_{reg}/ \G \subset
  \Image(\Pi_{\overline{V}}),
\end{equation}    so $\Pi_{\overline{V}}$ is essentially surjective in the sense above.
Thus we can  hope to be able to find a heuristic measure $d\mu_{\overline{V}}$ on $\overline{V}$
fulfilling Eq. \eqref{eq2.10} (with $V$ replaced by $\bar{V}$).\par

Recalling that  in the situation of non-compact $\Sigma$
the measure $\mu_V$ given by   Eq. \eqref{eq2.17} above has the right properties (ie fulfills Eq. \eqref{eq2.10})
 and taking into account that  $[\Sigma, G/T]$ is a countable set (cf. Remark \ref{rm2.4} below)
it is tempting to try the ansatz
\begin{equation} \label{eq2.23}
d\mu_{\overline{V}} := DA^{\orth}  \otimes \bigl(\det\bigl(1_{\ck}-\exp(\ad(B))_{| \ck}\bigr) DB\bigr)
\otimes \#
\end{equation}
where $\#$ is the counting measure on $[\Sigma, G/T]$.\par

It turns out   that with this choice
Eq. \eqref{eq2.10} is indeed fulfilled at a heuristic level
(with $V$ replaced by $\bar{V}$ and with ``$=$'' replaced by ``$\sim$''), or, equivalently, that we have
\begin{multline} \label{eq2.24}
\int_{\cA} \chi(A) DA \\
\sim \sum_{\cl \in [\Sigma,G/T]}
 \int_{\cB} \biggl[ \int_{\cA^{\orth}} \chi(A^{\orth} +
  \bar{g}_{\cl} B  \bar{g}_{\cl}^{-1} dt)  DA^{\orth} \biggr]
  \det\bigl(1_{\ck}-\exp(\ad(B))_{| \ck}\bigr) DB
\end{multline}
for every $\G$-invariant function $\chi:\cA \to \bC$, cf. Eq. (2.26) in \cite{Ha4}.

\smallskip

Clearly,  for $\cl \in [\Sigma,G/T]$, $B \in \cB$ the function $\bar{g}_{\cl}
 B  \bar{g}_{\cl}^{-1} \in C^{\infty}(\Sigma,\cG)$  will in general not be in  $C^{\infty}(\Sigma,\ct)$.
   This reduces the usefulness of Eq. \eqref{eq2.24}
considerably. Fortunately, for a certain class of functions
 $\chi$ one can derive  a more useful variant of Eq. \eqref{eq2.24}, cf. Eq. \eqref{eq2.27} below.

 \medskip

\noindent {\bf Preparation for Step 2:} Let us fix a point $\sigma_0
\in \Sigma$ and   pick, for each $\cl \in [\Sigma,G/T]$,  a lift
  $\Omega_{\cl} \in   \G_{\Sigma \backslash \{\sigma_0\}} = C^{\infty}(\Sigma \backslash \{\sigma_0\},G)$ of $(\bar{g}_{\cl})_{|\Sigma \backslash
  \{\sigma_0\}} \in C^{\infty}(\Sigma \backslash \{\sigma_0\},G/T)$,
  cf. Remark \ref{rm2.3} below. We will keep $\sigma_0$ and each  $\Omega_{\cl}$ fixed for the rest of
this paper.

\begin{remark} \label{rm2.3}  \rm
The existence of the ``lifts'' $\Omega_{\cl} \in   \G_{\Sigma \backslash
\{\sigma_0\}}$ of $(\bar{g}_{\cl})_{|\Sigma \backslash
  \{\sigma_0\}}$ picked above is guaranteed
  by the following three observations:
\begin{enumerate}
\item the surface $\Sigma \backslash \{\sigma_0\}$ (and in fact every
non-compact  surface) is homotopy equivalent to a 1-complex
\item  $G/T$ is simply-connected
\item $\pi_{G/T}: G \to G/T$ is a  fiber bundle
and therefore possesses   the homotopy lifting property
 (cf., e.g.,  \cite{Hu})
\end{enumerate}
The first two observations imply that
 each of the maps $(\bar{g}_{\cl})_{|\Sigma \backslash
  \{\sigma_0\}}$ is $0$-homotopic and the third  observation then
  implies the existence of $\Omega_{\cl}$ (cf. also Example \ref{rm_SU2_Top} in part \ref{appB'_1} of the Appendix below).
\end{remark}
 Clearly, the restriction mapping
 $\G_{\Sigma} \ni \Omega \mapsto
  \Omega_{| \Sigma \backslash \{\sigma_0 \}} \in \G_{ \Sigma \backslash \{\sigma_0 \}}$
is injective so
  we can consider $\G_{\Sigma}$
to be  a subgroup of  $\G_{ \Sigma \backslash \{\sigma_0 \}}$.
In a similar way, we will identify
 $C^{\infty}(\Sigma,\cG)$ with a subspace of $C^{\infty}(\Sigma \backslash \{\sigma_0\}, \cG)$,
$\cA^{\orth}$ with a subspace of
 $\cA^{\orth}_{ (\Sigma \backslash \{\sigma_0 \}) \times S^1}$ (defined
 in the obvious way), and
$\cA^{qax}= \cA^{\orth} \oplus C^{\infty}(\Sigma,\cG) dt$ with a subspace of
$\cA^{qax}_{(\Sigma \backslash \{\sigma_0 \}) \times S^1}:=
\cA^{\orth}_{(\Sigma \backslash \{\sigma_0 \}) \times S^1}
 \oplus C^{\infty}(\Sigma \backslash \{\sigma_0\},\cG) dt$.
Summarizing this  we have
\begin{align*}
 \G_{\Sigma} & \subset &  \overline{\G_{\Sigma}} & :=
  \G_{\Sigma \backslash \{\sigma_0\}} \\
 C^{\infty}(\Sigma,\cG) &  \subset & \overline{C^{\infty}}(\Sigma,\cG) &:=  C^{\infty}(\Sigma \backslash \{\sigma_0\},\cG) \\
\cA^{\orth}   &\subset&  \overline{\cA^{\orth}} & :=
   \cA^{\orth}_{(\Sigma\backslash \{\sigma_0\}) \times S^1} \\
\cA^{qax}  & \subset & \overline{\cA^{qax}} & :=  \cA^{qax}_{(\Sigma \backslash \{\sigma_0 \}) \times S^1}
\end{align*}

\bigskip

\noindent {\bf Step 2:}  Let  $\chi: \cA \to \bC$ be a
$\G$-invariant function with the extra property that the function
$\chi^{qax}: \cA^{qax} \to \bC$ given by $\chi^{qax} := \chi_{|
\cA^{qax}}$ is not only $\G_{\Sigma}$-invariant but can also be
extended to a $\overline{\G_{\Sigma}}$-invariant function
$\overline{\chi^{qax}}: \overline{\cA^{qax}}
 \to \bC$.
 If $\chi$ has this property  we  obtain
for the integrand in the inner integral on the right-hand side of
\eqref{eq2.24}
\begin{multline} \label{eq2.25}
\chi(A^{\orth} + (\bar{g}_{\cl}  B  \bar{g}_{\cl}^{-1}) dt) =
\chi^{qax}(A^{\orth} + (\bar{g}_{\cl}  B  \bar{g}_{\cl}^{-1}) dt)
  = \overline{\chi^{qax}}(A^{\orth} + (\Omega_{\cl}  B  \Omega_{\cl}^{-1}) dt) \\
 = \overline{\chi^{qax}}( (A^{\orth} \cdot \Omega_{\cl}) +  B   dt)
= \overline{\chi^{qax}}( \Omega^{-1}_{\cl} A^{\orth} \Omega_{\cl} + \Omega_{\cl}^{-1} d  \Omega_{\cl} +  B   dt)
\end{multline}
So for such a function $\chi$ we arrive at  the following
modification of Eq. \eqref{eq2.24}
\begin{multline} \label{eq2.26}
\int_{\cA} \chi(A) DA  \sim \sum_{\cl \in [\Sigma,G/T]}
 \int_{\cB} \biggl[ \int_{\cA^{\orth}}
  \overline{\chi^{qax}}(\Omega^{-1}_{\cl} A^{\orth} \Omega_{\cl} + \Omega_{\cl}^{-1} d  \Omega_{\cl}
   +  B   dt) DA^{\orth} \biggr]\\
\times   \det\bigl(1_{\ck}-\exp(\ad(B))_{| \ck}\bigr)   DB
\end{multline}

\noindent {\bf Step 3:}  As a final step let us  simplify Eq.
\eqref{eq2.26} by performing -- for each
 fixed $\cl \in [\Sigma,G/T]$ and each fixed $B \in \cB$
   -- the  change of variable   $ \Omega^{-1}_{\cl} A^{\orth} \Omega_{\cl} + \pi_{\ck}(\Omega_{\cl}^{-1}
  d  \Omega_{\cl})  \rightarrow  A^{\orth}$
  in the integral  $\bigl[ \int_{\cA^{\orth}}
  \overline{\chi^{qax}}(\Omega^{-1}_{\cl} A^{\orth} \Omega_{\cl} + \Omega_{\cl}^{-1} d  \Omega_{\cl}
   +  B   dt) DA^{\orth} \bigr]$ appearing  on the RHS of   Eq. \eqref{eq2.26}.
  Here   $\pi_{\ck}$ is the   $\langle \cdot,\cdot \rangle$-orthogonal projection $\cG \to \ck$.\par

  For the special  $\chi$ relevant for us (cf. Sec. \ref{subsubsec2.3.1} below),
  these changes of variable can indeed be justified\footnote{by contrast
 the more radical change of  variable  $ \Omega^{-1}_{\cl} A^{\orth} \Omega_{\cl} + \Omega_{\cl}^{-1} d  \Omega_{\cl}
   \rightarrow  A^{\orth}$ is not admissible (at least not in the present form),
    cf. Remark \ref{rm_var_trans} below}, cf. Sec. 4.2 in \cite{Ha3c}.
       After performing  them  we finally obtain
\begin{multline} \label{eq2.27}
 \int_{\cA} \chi(A) DA
  \sim \sum_{\cl \in [\Sigma,G/T]}
 \int_{\cB} \biggl[ \int_{\cA^{\orth}}  \overline{\chi^{qax}}(
  A^{\orth} + A_{\sing}(\cl) + Bdt)
 DA^{\orth} \biggr] \\
 \times  \det\bigl(1_{\ck}-\exp(\ad(B))_{| \ck}\bigr) DB
\end{multline}
where we have set
\begin{equation} \label{eq2.28}
A_{\sing}(\cl) := \pi_{\ct}(\Omega_{\cl}^{-1} d\Omega_{\cl})
\end{equation}
where  $\pi_{\ct}$ is the   $\langle \cdot,\cdot \rangle$-orthogonal projection $\cG \to
\ct$.

\medskip

Observe that we can equally well work with the abstract gauge fixing
space $\overline{V} := A^{\orth} \times C^{\infty}(\Sigma,\ct_{reg})
\times [\Sigma,G/T]$ instead of $\overline{V} := A^{\orth} \times
\cB \times [\Sigma,G/T]$ in Step 1 above
where
\begin{equation}\ct_{reg} := \exp^{-1}(T_{reg})
\end{equation}
with  $T_{reg}:= T \cap G_{reg}$ ($\ct_{reg}$ is just the union of all the Weyl alcoves).

\medskip

 If we do so we  arrive
at\footnote{alternatively, one could try to ``derive'' Eq.
\eqref{eq2.29} directly form Eq. \eqref{eq2.27} by argueing that,
since $\det\bigl(1_{\ck}-\exp(\ad(b))_{| \ck}\bigr) = 0$ for every
$b \in \ct \backslash \ct_{reg}$ the set $C^{\infty}(\Sigma,\ct)
\backslash C^{\infty}(\Sigma,\ct_{reg})$ should be a
$\det\bigl(1_{\ck}-\exp(\ad(B))_{| \ck}\bigr) DB$-zero subset of
$\cB= C^{\infty}(\Sigma,\ct)$}
 \begin{multline} \label{eq2.29}
 \int_{\cA} \chi(A) DA
  \sim \sum_{\cl \in [\Sigma,G/T]}
 \int_{\cB} 1_{C^{\infty}(\Sigma,\ct_{reg})}(B) \biggl[ \int_{\cA^{\orth}}  \overline{\chi^{qax}}(
  A^{\orth} + A_{\sing}(\cl) + Bdt)
 DA^{\orth} \biggr] \\
 \times  \det\bigl(1_{\ck}-\exp(\ad(B))_{| \ck}\bigr) DB
\end{multline}
which will be slightly more convenient in Sec. \ref{subsubsec2.3.1}
below.

 \begin{remark}  \label{rm2.4}  \rm
Let $n: [\Sigma,G/T] \to \ct$  be given  by
$$n(\cl) :=
 \int_{\Sigma \backslash \{\sigma_0\}} dA_{\sing}(\cl) := \lim_{\eps \to 0} \int_{\Sigma \backslash  B_{\eps}(\sigma_0)} dA_{\sing}(\cl)$$
 where $ B_{\eps}(\sigma_0)$, $\eps >0$, denotes the
  closed $\eps$-ball around $\sigma_0$ w.r.t.  to an arbitrary fixed
  Riemannian metric ${\mathbf g}$ on $\Sigma$. One can show that $n$ is well-defined and
    independent of the special choice of
 ${\mathbf g}$, $\sigma_0$, $\bar{g}_{\cl}$,  and $\Omega_{\cl}$ involved\footnote{the independence of $\sigma_0$
 is not mentioned explicitly in \cite{Ha3c} but is obvious from the proof of Proposition 3.4 and
  Proposition 3.5  in \cite{Ha3c} }(see Proposition 3.4 and Proposition 3.5 in \cite{Ha3c}).
 Moreover, one can show
that $n$ is  a bijection from $[\Sigma,G/T]$
onto the lattice
\begin{equation} \label{eq2.30} I := \ker(\exp_{| \ct})
\end{equation}
(see Proposition 3.6 and Remark 3.2 in \cite{Ha3c}; cf. also Sec. 5 in  \cite{BlTh3}  and Example \ref{rm_SU2_Top} in part \ref{appB'_1} of the Appendix below).
\end{remark}

\subsection{Chern-Simons theory in the torus gauge}
\label{subsec2.3}

\subsubsection{Application of Eq. \eqref{eq2.29} to Chern-Simons theory}
\label{subsubsec2.3.1}

 From now on we will assume that the (connected) surface $\Sigma$ fixed above
 is oriented and compact. Let  $L=
((l_1, l_2, \ldots, l_m),(\rho_1,\rho_2,\ldots,\rho_m))$ be a
colored link in $M= \Sigma \times S^1$. We would like to apply Eq.
\eqref{eq2.29} to  the $\G$-invariant function $\chi:\cA \to \bC$
given by
\begin{equation} \label{eq2.31} \chi(A) = \prod_i  \Tr_{\rho_i}\bigl(\Hol_{l_i}(A)\bigr)
 \exp(iS_{CS}(A))
 \end{equation}
 Before we can do this we need to extend
 the $\G_{\Sigma}$-invariant
 function $\chi^{qax}:= \chi_{|\cA^{qax}}$ on $\cA^{qax}$
 to a  $\overline{\G_{\Sigma}}$-invariant
 function $\overline{\chi^{qax}}$ on $\overline{\cA^{qax}}$.
 We do this by extending each of   the  $\G_{\Sigma}$-invariant functions
 \begin{subequations}
\begin{align} \label{eq2.32}
  \Tr_{\rho_i}(\Hol_{l_i}):\cA^{qax} \ni  A^q  & \mapsto
   \Tr_{\rho_i}\bigl(\Hol_{l_i}(A^q)\bigr) \in \bC, \quad \quad i \le m\\
S_{CS}^{qax}:\cA^{qax} \ni  A^q  & \mapsto S_{CS}( A^q) \in \bC
\end{align}
\end{subequations}
separately
 to  $\overline{\G_{\Sigma}}$-invariant
 functions $\overline{\Tr_{\rho_i}(\Hol_{l_i})}$ and $\overline{S_{CS}^{qax}}$
   on $\overline{\cA^{qax}}$ and then setting
 \begin{equation} \label{eq2.34}
 \overline{\chi^{qax}}(A^q) := \prod_i  \overline{\Tr_{\rho_i}(\Hol_{l_i})}(A^q)
 \exp(i \overline{S_{CS}^{qax}}(A^q))
 \end{equation}
 for $A^q \in \overline{\cA^{qax}}$.
In the following we will assume that $\sigma_0$ does not lie on
 any of the loops $l^i_{\Sigma}:= \pi_{\Sigma} \circ l_i$ where $\pi_{\Sigma}:\Sigma \times S^1 \to \Sigma$ is the canonical projection.
More precisely, we will assume that
\begin{equation} \label{eq_ass_sigma0} \sigma_0 \notin \bigcup_{i=1}^m \arc(  l^i_{\Sigma}) \subset \Sigma
\end{equation}
Clearly, under this assumption  $\Hol_{l_i}(A^q)$
makes sense for arbitrary    $A^q \in \overline{\cA^{qax}}$
   and setting  $\overline{\Tr_{\rho_i}(\Hol_{l_i})}(A^q) :=
\Tr_{\rho_i}(\Hol_{l_i}(A^q))$ for all $A^q \in \overline{\cA^{qax}}$
we obtain a  $\overline{\G_{\Sigma}}$-invariant function on $\overline{\cA^{qax}}$.\par
Before we write down an  explicit expression for
 $\overline{S_{CS}^{qax}}$  we will
  first give a convenient explicit formula
 for the function $S_{CS}^{qax}$ on $\cA^{qax} = \cA^{\orth} \oplus C^{\infty}(\Sigma,\cG) dt$.
 We have
\begin{multline}\label{eq2.35} S_{CS}^{qax}(A^{\orth} + B dt ) = S_{CS}(A^{\orth} + B dt ) \\
= k\pi  \int_M \bigl[ \Tr(A^{\orth} \wedge dA^{\orth} ) +
2 \Tr( A^{\orth} \wedge  B dt \wedge A^{\orth}) + 2  \Tr(A^{\orth}
\wedge dB \wedge dt) \bigr]
\end{multline}
for all $A^{\orth} \in \cA^{\orth}$, $B \in C^{\infty}(\Sigma,\cG)$.
As in \cite{Ha3b,Ha4} we can rewrite  the RHS  of Eq.
\eqref{eq2.35} using the identification
\begin{equation} \label{eq2.36}
\cA^{\orth} \cong C^{\infty}(S^1,\cA_{\Sigma})
\end{equation}
 where $C^{\infty}(S^1,\cA_{\Sigma})$
denotes the space of all functions $\alpha:S^1 \to \cA_{\Sigma}$
which are ``smooth'' in the sense that for every smooth vector field
$X$ on $\Sigma$ the function $ \Sigma \times S^1 \ni (\sigma,t)
\mapsto \alpha(t)(X_{\sigma})$ is smooth.\par

Using this identification we have (cf. Proposition 5.2 in
\cite{Ha3b})
\begin{multline} \label{eq2.37}
S_{CS}^{qax}(A^{\orth} + B dt ) = S_{CS}(A^{\orth} + B dt )\\
 =  - k\pi  \int_{S^1} dt \int_{\Sigma} \biggl[ \Tr\bigl( A^{\orth}(t) \wedge
\bigl(\tfrac{\partial}{\partial t} + \ad(B) \bigr) A^{\orth}(t)
\bigr) - 2  \Tr\bigl( A^{\orth}(t) \wedge dB \bigr) \biggr]
\end{multline}
where $\tfrac{\partial}{\partial t}:C^{\infty}(S^1,\cA_{\Sigma}) \to
C^{\infty}(S^1,\cA_{\Sigma})$ is the obvious differential
operator.\par

Let us now introduce  $\overline{S_{CS}^{qax}}:  \overline{\cA^{qax}}
 \to \bC$ by
\begin{multline} \label{eq2.38}
\overline{S_{CS}^{qax}}(A^{\orth} + B dt )\\
 :=  - k\pi \biggl\{ \lim_{\eps \to 0}
  \int_{S^1} dt \int_{\Sigma \backslash B_{\eps}(\sigma_0)}
\biggl[ \Tr\bigl( A^{\orth}(t) \wedge \bigl(\tfrac{\partial}{\partial t} +
\ad(B) \bigr) A^{\orth}(t) \bigr) - 2  \Tr\bigl( dA^{\orth}(t) \cdot
B \bigr) \biggr] \biggr\}
\end{multline}
(with  $B_{\eps}(\sigma_0)$  as in Remark \ref{rm2.4} above)
for all $A^{\orth} \in \overline{\cA^{\orth}}$,  $B \in \overline{C^{\infty}}(\Sigma,\cG)$
for which the $\eps \to 0$-limit exists and by
$$\overline{S_{CS}^{qax}}(A^{\orth} + B dt ) := 0$$
otherwise.
 Observe that $\overline{S_{CS}^{qax}}$ is indeed an extension of $S_{CS}^{qax}$
since in the special case
where $A \in \cA^{\orth} $
and $B \in  C^{\infty}(\Sigma,\cG)$
Stokes' Theorem implies that for every $t \in S^1$ we have\footnote{here and above $dA^{\orth}(t)$ is a short notation for $d(A^{\orth}(t))$,
 i.e. ``$d$'' is the differential of $\Sigma$}
\begin{equation} \label{eq2.39}
  \int_{\Sigma \backslash \{\sigma_0\}}
  \Tr\bigl( dA^{\orth}(t) \cdot B \bigr) =
 \int_{\Sigma }
  \Tr\bigl( dA^{\orth}(t) \cdot B \bigr)
=  \int_{\Sigma }
  \Tr\bigl( A^{\orth}(t) \wedge dB \bigr)
\end{equation}
Moreover,   $\overline{S_{CS}^{qax}}$ is indeed a
$\overline{\G_{\Sigma}}$-invariant function, cf. Proposition 4.1 in
\cite{Ha3c}.
 We can therefore apply Eq. \eqref{eq2.29}  to the functions $\chi$
and $\overline{\chi^{qax}}$ given by Eq. \eqref{eq2.31} and Eq.
\eqref{eq2.34} with the choice Eq. \eqref{eq2.38} above and obtain
\begin{multline}  \label{eq2.41_pre}
 \WLO(L)  \sim \sum_{\cl \in [\Sigma,G/T]}  \int_{\cB} 1_{C^{\infty}(\Sigma,\ct_{reg})}(B)
\biggl[ \int_{\cA^{\orth}} \prod_i  \overline{\Tr_{\rho_i}( \Hol_{l_i}\bigr)}\bigl(A^{\orth}+ A_{\sing}(\cl) + Bdt\bigr) \\ \times   \exp(i  \overline{S^{qax}_{CS}}( A^{\orth} + A_{\sing}(\cl) + B dt)) DA^{\orth} \biggr]  \det\bigl(1_{\ck}-\exp(\ad(B))_{|\ck}\bigr) DB
\end{multline}
Here and in the following $\sim$ denotes equality up to a
  multiplicative  constant independent of $L$.

  \begin{remark} \label{rm_var_trans} \rm
 It is tempting to try to simplify Eq. \eqref{eq2.41_pre}
 by performing   the informal change of variable $A^{\orth} + A_{\sing}(\cl) \to A^{\orth}$
 in the $\int_{\cA^{\orth}} \cdots DA^{\orth}$ integral.
  However, it turns out that if we perform this change of variable directly
 in Eq. \eqref{eq2.41_pre} things go wrong.
 More precisely, the explicit  evaluation of the expression that we obtain by replacing  each appearance of
 $A^{\orth} + A_{\sing}(\cl)$    on the RHS of Eq. \eqref{eq2.41_pre}
   by $A^{\orth}$ leads to incorrect values for $\WLO(L)$.
On the other hand, it turns out that if we perform an analogous change of variable at a later stage
(after having rewritten Eq. \eqref{eq2.41_pre} above in a suitable way,
cf. Eq. \eqref{eq2.41simpl} below)  we do get the correct values for $\WLO(L)$.
\end{remark}

In order to get to the  aforementioned
 formula \eqref{eq2.41simpl} we  will use  the following identity\footnote{step $(*)$ holds
since by Stokes' theorem we have
 $\int_{\Sigma \backslash \{\sigma_0\}} d (A_{\sing}(\cl) \cdot B)
 = \lim_{\eps \to 0} \int_{\Sigma \backslash  B_{\eps}(\sigma_0)}
  d (A_{\sing}(\cl) \cdot B) =  \lim_{\eps \to 0} \int_{\partial B_{\eps}(\sigma_0)}
 (A_{\sing}(\cl) \cdot B) \overset{(+)}{=}
 \bigl(\lim_{\eps \to 0} \int_{\partial B_{\eps}(\sigma_0)}   A_{\sing}(\cl)\bigr) \cdot B(\sigma_0)
 =\bigl(\lim_{\eps \to 0} \int_{\Sigma \backslash  B_{\eps}(\sigma_0)} dA_{\sing}(\cl) \bigr) \cdot B(\sigma_0)
 =  n(\cl) \cdot B(\sigma_0)$.
 Here step $(+)$ holds for all $B$ provided that  $A_{\sing}(\cl)$ was chosen to be sufficiently regular.
Moreover, if $B$ is locally constant around $\sigma_0$ (which is the only case relevant for us, cf. part \ref{appB'_3}
of the appendix) then step $(+)$ holds for an arbitrary choice of $A_{\sing}(\cl)$}
 \begin{align} \label{eq_identity_ncl} \int_{\Sigma \backslash \{\sigma_0\}} \Tr\bigl(dA_{\sing}(\cl) \cdot B\bigr)
  & =\int_{\Sigma \backslash \{\sigma_0\}} \Tr\bigl(d (A_{\sing}(\cl) \cdot B)\bigr)
   + \int_{\Sigma \backslash \{\sigma_0\}} \Tr (A_{\sing}(\cl) \wedge dB) \nonumber \\
& \overset{(*)}{=} \Tr\bigl( n(\cl) \cdot B(\sigma_0)) + \int_{\Sigma \backslash \{\sigma_0\}} \Tr (A_{\sing}(\cl) \wedge dB)\bigr)
 \end{align}
where $ n(\cl) = \int_{\Sigma \backslash \{\sigma_0\}} dA_{\sing}(\cl)
       =  \lim_{\eps \to 0} \int_{\Sigma \backslash  B_{\eps}(\sigma_0)} dA_{\sing}(\cl)$
  (cf. Remark \ref{rm2.4} above).\par

Using \eqref{eq_identity_ncl}
and  taking into account that
\begin{equation} \label{eq2.40} \overline{S_{CS}^{qax}}(A^{\orth}+ A_{\sing}(\cl)+ Bdt) =
S_{CS}^{qax}(A^{\orth}+Bdt) + 2 \pi k \int_{\Sigma \backslash
\{\sigma_0\}} \Tr\bigl(dA_{\sing}(\cl) \cdot B\bigr)
\end{equation}
 we see that we  can rewrite Eq. \eqref{eq2.41_pre} above as
    \begin{multline}  \label{eq2.41'} \WLO(L)
 \sim \sum_{\cl \in [\Sigma,G/T]}  \int_{\cB} 1_{C^{\infty}(\Sigma,\ct_{reg})}(B) \times \\
 \times \biggl[ \int_{\cA^{\orth}} \prod_i
  \overline{\Tr_{\rho_i}( \Hol_{l_i}\bigr)}\bigl(A^{\orth}+ A_{\sing}(\cl) + Bdt\bigr)
  \exp(i  \bar{S}^{qax}_{CS}( A^{\orth} + A_{\sing}(\cl) + B dt)) DA^{\orth} \biggr] \\
 \times \exp\bigl( 2\pi i k  \Tr(n(\cl) \cdot B(\sigma_0) ) \bigr)
  \det\bigl(1_{\ck}-\exp(\ad(B))_{|\ck}\bigr) DB
\end{multline}
 where  $\bar{S}^{qax}_{CS}(A^{\orth} + A_{\sing}(\cl) + B dt)$
 is given in an analogous way\footnote{ie on the RHS of \eqref{eq2.37} we replace $A^{\orth}(t)$ by $A^{\orth}(t) + A_{\sing}(\cl)$  and we replace the integral $\int_{\Sigma} \cdots$ by
  $\lim_{\eps \to 0}  \int_{\Sigma \backslash B_{\eps}(\sigma_0)} \cdots$}
  as in  Eq. \eqref{eq2.37} above.
 If we now perform in Eq. \eqref{eq2.41'} above
   the informal change of variable $A^{\orth} + A_{\sing}(\cl) \to A^{\orth}$
  in the $\int_{\cA^{\orth}} \cdots DA^{\orth}$ integral we obtain
\begin{multline}  \label{eq2.41simpl} \WLO(L)
 \sim \sum_{\cl \in [\Sigma,G/T]}  \int_{\cB} 1_{C^{\infty}(\Sigma,\ct_{reg})}(B) \times \\
 \times \biggl[ \int_{\cA^{\orth}} \prod_i  \Tr_{\rho_i}\bigl(
 \Hol_{l_i}(A^{\orth} + Bdt)\bigr)
  \exp(i  S_{CS}( A^{\orth} + B dt)) DA^{\orth} \biggr] \\
 \times \exp\bigl( 2\pi i k  \Tr(n(\cl) \cdot B(\sigma_0) ) \bigr)
  \det\bigl(1_{\ck}-\exp(\ad(B))_{|\ck}\bigr) DB
\end{multline}
(Here we use again the notation $S_{CS}( A^{\orth} + B dt)$ instead of $\bar{S}^{qax}_{CS}( A^{\orth} + B dt)$
and $\Tr_{\rho_i}(\Hol_{l_i}(\cdot))$ instead of $\overline{\Tr_{\rho_i}(\Hol_{l_i})}$).

\smallskip

 Let us emphasize that the derivation of Eq. \eqref{eq2.41simpl} above was
  rather sloppy. In particular, we did not specify explicitly the domain of the function $\bar{S}^{qax}_{CS}(\cdot)$.
  Moreover, we did not give any explanation why now the informal change of variable $A^{\orth} + A_{\sing}(\cl) \to A^{\orth}$ will ``work'' while above (cf. Remark \ref{rm_var_trans}) it didn't.
 We will clarify this issue in part \ref{appB'_2} of the Appendix below where we will
  give a  careful derivation of Eq. \eqref{eq2.41simpl}.

 \begin{remark} \label{rm_loc_constant} \rm
  In view of the discussion  in Sec. \ref{subsec4.10}  below let us mention
 that  from  part \ref{appB'_3} of the Appendix below
 it follows  (on a heuristic level) that the integral
 $$\int_{\cA^{\orth}} \prod_i  \Tr_{\rho_i}\bigl( \Hol_{l_i}(A^{\orth} + Bdt)\bigr)
  \exp(i  S_{CS}( A^{\orth} + B dt)) DA^{\orth}$$
  vanishes for every  $B \in \cB$ which is not locally constant around $\sigma_0$.
  Accordingly, we can replace the space $\cB$  appearing in the outer integral $\int_{\cB} \cdots DB$
  in Eq. \eqref{eq2.41simpl} above and in Eq. \eqref{eq2.48} below
   by the space
 $$\cB^{loc}_{\sigma_0}:= \{ B \in \cB \mid B \text{ is locally constant around $\sigma_0$} \}$$
  \end{remark}

\subsubsection{Rewriting $S_{CS}(A^{\orth} + B dt)$}
\label{subsubsec2.3.2}

Let us rewrite $S_{CS}(A^{\orth} + B dt)$ yet another time. In order
to do so we fix an auxiliary Riemannian metric ${\mathbf g}$ on
$\Sigma$.  By $\star$ we will denote both the Hodge star operator on
$\cA_{\Sigma}$ induced by  ${\mathbf g}$
and the linear isomorphism on $\cA^{\orth} \cong
C^{\infty}(S^1,\cA_{\Sigma})$ given by $(\star A^{\orth})(t) = \star
(A^{\orth}(t))$ for all $t \in S^1$.
By $\ll \cdot, \cdot
\gg_{\cA^{\orth}}$ we will denote the scalar product on $\cA^{\orth} \cong
C^{\infty}(S^1,\cA_{\Sigma})$ given by
\begin{equation} \label{eq2.43} \ll A^{\orth}_1,  A^{\orth}_2 \gg_{\cA^{\orth}} =  \int_{S^1} dt
  \int_{\Sigma} ( A^{\orth}_1(t),  A^{\orth}_2(t))_{\cA_{\Sigma}} d\mu_{\mathbf g}
\end{equation}
where $d\mu_{\mathbf g}$ is the volume measure  on $\Sigma$ associated to $\mathbf g$ and  $(\cdot,\cdot)_{\cA_{\Sigma}}: \cA_{\Sigma}
\times  \cA_{\Sigma} \to C^{\infty}(\Sigma,\bR)$ is the obvious
bilinear map induced by $\mathbf g$
and $\langle \cdot, \cdot \rangle_{\cG}$.\par

  Eq. \eqref{eq2.37} can then be rewritten as  (cf.\footnote{we remark that in  \cite{Ha4} we use a different sign convention for the $\star$ operator} Sec. 3.3 in \cite{Ha4})
   \begin{align} \label{eq2.42}
S_{CS}(A^{\orth} + B dt ) & =  k \pi  \bigl[ \ll A^{\orth},
 \star  \bigl(\tfrac{\partial}{\partial t} + \ad(B) \bigr) A^{\orth} \gg_{\cA^{\orth}}
 + 2 \ll  \star A^{\orth}, dB \gg_{\cA^{\orth}} \bigr],
\end{align}
We remark that even though  Eq. \eqref{eq2.42}  is  less natural
than Eq. \eqref{eq2.37} above (since it depends on the
auxiliary  Riemannian metric  ${\mathbf g}$)
 it will be more convenient for our purposes,
 cf. the paragraph after Eqs \eqref{eq2.49}  below and
 also Sec. \ref{subsec4.2} below where a ``simplicial'' analogue of
  $S_{CS}(A^{\orth} + B dt)$ is  introduced.

\subsubsection{The final heuristic formula}
\label{subsubsec2.3.3}

Clearly,  the operator $ \tfrac{\partial}{\partial t} +
\ad(B):C^{\infty}(S^1,\cA_{\Sigma}) \to
C^{\infty}(S^1,\cA_{\Sigma})$ is not injective. For $B \in
C^{\infty}(\Sigma,\ct_{reg})$  the kernel of this operator is given
by
\begin{equation} \label{eq2.44} \cA_{c}^{\orth}   :=
 \{ A^{\orth} \in  C^{\infty}(S^1,\cA_{\Sigma}) \mid A^{\orth}
   \text{ is constant and }\cA_{\Sigma,\ct}\text{-valued}\} \cong  \cA_{\Sigma,\ct}
\end{equation}
 For making rigorous sense of the RHS  of Eq. \eqref{eq2.41simpl}
 it is
 useful to work with a  decomposition $\cA^{\orth} = \cC \oplus
\cA^{\orth}_c$ for a suitably chosen linear subspace $\cC$ of
$\cA^{\orth}$ and to incorporate this decomposition into Eq.
\eqref{eq2.41simpl}. For technical reasons
  we worked in \cite{Ha3b,Ha3c,Ha4,Ha6} with the choice
 $\cC:= \hat{\cA}^{\orth}$ where
\begin{equation}\label{eq2.45} \hat{\cA}^{\orth}  := \{ A^{\orth} \in C^{\infty}(S^1,\cA_{\Sigma})
 \mid A^{\orth}(t_0) \in \cA_{\Sigma,\ck}\}
 \end{equation}
where $t_0$ is an arbitrary but fixed point in $S^1$.
 It would also
have been possible to work with  $\cC:= \Check{\cA}^{\orth}$ where
\begin{equation} \label{eq2.46}
\Check{\cA}^{\orth}  := \{ A^{\orth} \in
C^{\infty}(S^1,\cA_{\Sigma})
 \mid  \int_{S^1} A^{\orth}(t) dt \in \cA_{\Sigma,\ck}\}
\end{equation}
In fact, the latter choice is in some sense more natural and  has
the important advantage of possessing a ``good'' simplicial analogue, cf. the second of the two equations in
\eqref{eq4.31} below. In the present paper we will therefore work
with the latter choice. Taking into account that
\begin{equation} \label{eq2.47} S_{CS}(\Check{A}^{\orth} + A^{\orth}_c + B dt)
= S_{CS}(\Check{A}^{\orth}  + B dt) + S_{CS}(A^{\orth}_c + B dt)
\end{equation}
for $\Check{A}^{\orth} \in \Check{\cA}^{\orth}$, $A^{\orth}_c \in \cA^{\orth}_c $, and $B \in \cB$
and the fact that the map $n:[\Sigma,G/T] \to I = \ker(\exp_{| \ct})$
appearing in Remark \ref{rm2.4} above is a bijection
we arrive, after incorporating the decomposition
 $\cA^{\orth} = \Check{\cA}^{\orth}  \oplus \cA^{\orth}_c$ into Eq. \eqref{eq2.41simpl}, at
\begin{multline}  \label{eq2.48} \WLO(L)
 \sim \sum_{y \in I}  \int_{\cA^{\orth}_c \times \cB} \biggl\{
 1_{C^{\infty}(\Sigma,\ct_{reg})}(B)  \Det_{FP}(B)\\
 \times   \biggl[ \int_{\Check{\cA}^{\orth}} \prod_i  \Tr_{\rho_i}\bigl(
 \Hol_{l_i}(\Check{A}^{\orth} + A^{\orth}_c ,B)\bigr)
  \exp(i  S_{CS}( \Check{A}^{\orth}, B)) D\Check{A}^{\orth} \biggr] \\
 \times \exp\bigl( - 2\pi i k  \langle y,  B(\sigma_0)\rangle \bigr) \biggr\}
 \exp(i S_{CS}(A^{\orth}_c, B)) (DA^{\orth}_c \otimes DB)
\end{multline}
where we have used Eq. \eqref{eq_Tr_scalar}
and where,  as a preparation for Sec. \ref{sec4}, we have introduced the
short notation
\begin{subequations} \label{eq2.49}
\begin{align}
S_{CS}(A^{\orth},B) & := S_{CS}(A^{\orth} + B dt ) \\
\Det_{FP}(B) &:=    \det\bigl(1_{\ck}-\exp(\ad(B))_{|
\ck}\bigr)\\
\label{eq2.49c} \Hol_{l_i}(A^{\orth},   B) & := \Hol_{l_i}(A^{\orth} +  B dt)
\end{align}
\end{subequations}

 From the explicit formula Eq. \eqref{eq2.42} for $S_{CS}(A^{\orth} + B dt )$ it
follows immediately that both (heuristic) complex measures $\exp(i
S_{CS}( \Check{A}^{\orth}, B)) D\Check{A}^{\orth}$ and $\exp(i
S_{CS}(A^{\orth}_c, B)) (DA^{\orth}_c \otimes DB)$ appearing above
are of ``Gaussian type''. This considerably increases the chances of
making rigorous sense of the RHS  of Eq. \eqref{eq2.48}.\par
 In fact, in \cite{Ha3b,Ha4,Ha6} we have  already sketched how this works
in the framework of white noise analysis.
In \cite{Ha3b,Ha4,Ha6,HaHa} we also  demonstrated that in the special
case where the link $L$ has no double points and  ``horizontal''
framing is used (cf. Sec. 5.2 in \cite{Ha4}) the rigorous
realization of $\WLO(L)$
 can be evaluated explicitly and that
 \begin{equation} \label{eq2.50}
\WLO(L) \sim |L|
\end{equation}
holds where $|\cdot|$ is  Turaev's shadow invariant associated to   $\cG$ and $k$
 (cf.  Remark \ref{rm3.4} below and part B of  the Appendix in \cite{Ha7b}) and where  $\sim$ denotes equality up to a  multiplicative constant independent of $L$. \par

 In the present paper  we will develop an alternative
  approach for making rigorous sense of Eq. \eqref{eq2.50},
   which is based on a suitable ``discretization'' of the RHS  of
Eq. \eqref{eq2.48}.

\subsubsection{Some Remarks}
\label{subsubsec2.3.4}

\begin{remark}   \rm
  By generalizing the heuristic arguments above in an
 obvious way  Eq. \eqref{eq2.48}  can also be ``derived'' for a general
simply-connected compact $G= \prod_{i=1}^r G_i $ (and a fixed r-tuple
$k=(k_i)_{i \le r}$),  cf. Remark \ref{rm2.1} above.
In this case  the first equation in
\eqref{eq2.49} has to be replaced by
$S_{CS}(A^{\orth},B)  := S_{CS}(M,G,(k_i)_i)(A^{\orth},B)$.
\end{remark}

\begin{remark} \label{rm2.5pre}  \rm We could rewrite
  formula \eqref{eq2.48}  by applying  (at an informal level) the Poisson summation formula
 but for stylistic reasons we prefer to use the Poisson summation formula  only  at a
 later stage (at a rigorous level) in the proof of Theorem \ref{main_theorem} below, see \cite{Ha7b}.
\end{remark}

\begin{remark} \label{rm2.5}  \rm
 In the special case where $L$ consists of $m$ ``vertical''\footnote{cf. Sec. \ref{subsubsec2.2.3} above and see also
Remark \ref{rm3.3} below; we remark also that  the case $m=0$ corresponds to the situation
$L= \emptyset$, ie $L$ is the empty link, cf. Sec. \ref{subsec4.9} below}
 loops $l_i = l_{\sigma_i}$  ``above'' the fixed points $\sigma_i$ in $\Sigma$
  we have $\Hol_{l_i}(A^{\orth} +  B dt) = B(\sigma_i)$. In particular, the inner integral in
  Eq. \eqref{eq2.48} is then trivial
  and we can perform this trivial integration right away at a heuristic level.
  By doing so and using the notation $Z(\Sigma \times S^1,(\sigma_i)_i,(\rho_i)_i)$ instead of $\WLO(L) $
  we then obtain
 \begin{multline} \label{eq2.48_simpl}
Z(\Sigma \times S^1,(\sigma_i)_i,(\rho_i)_i) \sim \sum_{y \in I}  \int_{\cA^{\orth}_c \times \cB} \biggl\{
 1_{C^{\infty}(\Sigma,\ct_{reg})}(B)  \Det_{FP}(B) Z(B)  \bigl( \prod_{i=1}^m  \Tr_{\rho_i}\bigl(\exp(B(\sigma_i))\bigr) \bigr) \\
 \times   \exp\bigl( - 2\pi i k  \langle y,  B(\sigma_0)\rangle \bigr) \biggr\}
 \exp\bigl(i S_{CS}(A^{\orth}_c, B)\bigr) (DA^{\orth}_c \otimes DB)
 \end{multline}
 where $Z(B) := \int  \exp(i  S_{CS}( \Check{A}^{\orth}, B)) D\Check{A}^{\orth}$.
 (Here and throughout the present remark ``$\sim$'' denotes equality up to a multiplicative constant
 which depends only on $G$, $\Sigma$, and $k$ but not on  $(\sigma_i)_i$ or $(\rho_i)_i$.)\par

 As explained in Sec. 3 in \cite{BlTh1} in the special case where $B$ is constant and taking values
 in $\ct_{reg}$  (which is the only case relevant in the present situation)
 the expression $\Det_{FP}(B) Z(B)$ can be interpreted as the Ray-Singer torsion associated
 to  any fixed Riemannian metric on $\Sigma \times S^1$
 and  the  1-form $B dt$ on $M = \Sigma \times S^1$ (considered\footnote{observe
 that $B dt$ induces a (flat) connection 1-form on the trivial principle fiber bundle $P=M \times T$
 which in turn induces a flat connection on the associated vector bundle $P \times_{\rho} \ck \cong M \times \ck = E$
where $\rho:T \to \Aut(\ck)$ is the restriction of $\Ad_{|T}$ to  $\ck$.
  Recall that $\ck$ is the $\langle \cdot,\cdot \rangle$-orthogonal complement of $\ct$ in $\cG$,
  which is a $\Ad_{|T}$-invariant subspace of $\cG$} as a flat connection in the trivial vector bundle
 $E=M \times \ck$). This  Ray-Singer torsion can be evaluated explicitly and we then obtain
  \begin{equation} \label{eq_RaySinger} \Det_{FP}(B) Z(B) = \det(1_{\ck}-\exp(\ad(b))_{| \ck})^{\chi(\Sigma)/2}
 \end{equation}
where $b \in \ct_{reg}$ is the unique value of $B$ and
 where $\chi(\Sigma)$ is the Euler characteristic of $\Sigma$.
  In view of the last formula
  and the relation  $S_{CS}(A^{\orth}_c, B) = 2 \pi k  \int_{\Sigma }   \Tr\bigl(B \cdot dA_c^{\orth}  \bigr)$
  it is clear  that  Eq. \eqref{eq2.48_simpl} is closely related\footnote{observe that even though both formulas
  look quite similar and give rise to the same values of $\WLO(L)$ (cf. Eq. \eqref{eq_Gen_Verlinde} below)
there are some differences.
  For example, in Eq. \eqref{eq2.48_simpl} we have a sum $\sum_{y \in I}$,
   a  factor $ \exp(- 2\pi i k  \langle y, B(\sigma_0)\rangle)$, and the integration
   $\int \cdots DA^{\orth}_c$
       while in  (the generalization of) formula (7.9) in  \cite{BlTh1}
 we have  a sum $\sum_{\lambda \in \Lambda}$ over the weight lattice $\Lambda$,
 a factor ``$ \exp\bigl( -  \int_{\Sigma} tr(\lambda \cdot F)\bigr)$'',
 and the integration    $\int \cdots DF$ where $F$  runs over the space of all  2-forms
 on $\Sigma$. We remark also that in contrast to Eq. \eqref{eq2.48_simpl}, which
 is a special case of Eq. \eqref{eq2.48} for general links $L$,
   formula (7.9) in  \cite{BlTh1} does not seem to have a natural generalization
         to the situation of general links $L$} to the  formula (7.9) in \cite{BlTh1},  or rather, the obvious generalization/modification of (7.9) in \cite{BlTh1}
which one obtains after including the analogue of the factor $\prod_{i=1}^m  \Tr_{\rho_i}\bigl(\exp(B(\sigma_i))\bigr)$ appearing above, replacing\footnote{here $h$ is the notation in \cite{BlTh1} for the
 the dual Coxeter number of $\cG$ (which we denote by $\cg$). We refer to Remark \ref{rm_shift_in_k} below for a comment on the  replacement $k+h \to k$}
 ``$k+h$'' by $k$
and replacing the group $G=SU(n)$ by a general simple simply-connected Lie group.

\smallskip

 Formula (7.9) in \cite{BlTh1} was evaluated at a heuristic level leading to
 Eq. \eqref{eq_Gen_Verlinde} below.
 We can obtain Eq. \eqref{eq_Gen_Verlinde} below also from  Eq. \eqref{eq2.48_simpl}
 using essentially the same heuristic arguments as in \cite{BlTh1} (but in a slightly different order):\par
 First we integrate out the variable $A^{\orth}_c$. Since the only term in Eq. \eqref{eq2.48_simpl}
 depending on $A^{\orth}_c$ is the factor
 $ \exp\bigl(i S_{CS}(A^{\orth}_c, B)\bigr) =  \exp\bigl( 2 \pi i k  \ll  \star A^{\orth}_c, dB \gg_{\cA^{\orth}}
 \bigr) $
 we obtain, informally, a delta function $\delta(dB)$.
 In view of this delta-function the $\int \cdots DB$-integral  can be replaced
 by an integral over the subspace $\cB_{c} := \{ B \in \cB \mid B \text{ is constant }\} \cong \ct$.
From this and Eq. \eqref{eq_RaySinger} above
 we therefore obtain at a heuristic level
  \begin{multline}
Z(\Sigma \times S^1,(\sigma_i)_i,(\rho_i)_i) \sim \sum_{y \in I}  \int_{\ct} \bigl\{ 1_{\ct_{reg}}(b)  \det(1_{\ck}-\exp(\ad(b))_{| \ck})^{\chi(\Sigma)/2}   \\ \times \bigl( \prod_{i=1}^m  \Tr_{\rho_i}\bigl(\exp(b)\bigr) \bigr)  \exp\bigl( - 2\pi i k  \langle y, b \rangle \bigr) \bigr\} db
 \end{multline}
Using the Poisson summation formula (at an informal level) we arrive at
 \begin{multline}
Z(\Sigma \times S^1,(\sigma_i)_i,(\rho_i)_i) \\
\sim \sum_{\lambda \in \Lambda} \bigl\{ 1_{\ct_{reg}}(\lambda/k)  \det(1_{\ck}-\exp(\ad(\lambda/k)_{| \ck})^{\chi(\Sigma)/2}   \bigl( \prod_{i=1}^m  \Tr_{\rho_i}\bigl(\exp(\lambda/k)\bigr) \bigr)  \bigr\}
 \end{multline}
 where $\Lambda$ is the weight lattice of $(\cG,\ct)$, i.e. the lattice dual to $I$.
 Using Weyl's character formula\footnote{which implies that  $\Tr_{\rho_i}(e^{(\lambda + \rho)/k}) = \tfrac{S_{\lambda \mu_i}}{S_{\lambda 0}}$ for $\lambda \in  \Lambda_+^k$}, the relation
 $\det(1_{\ck}-\exp(\ad((\lambda+\rho)/k))_{| \ck}) \sim (S_{\lambda 0})^2$
where $\rho \in \Lambda$ is the    ``Weyl vector''\footnote{ ie. the half sum of positive roots of  $\cG$
   w.r.t.  $\ct$ and the fixed Weyl chamber mentioned after Eq. \eqref{eq_Gen_Verlinde} below},
 and a suitable invariance argument based on the affine Weyl group we eventually obtain
 \begin{equation} \label{eq_Gen_Verlinde}
Z(\Sigma \times S^1,(\sigma_i)_i,(\rho_i)_i) \sim \sum_{\lambda \in \Lambda_+^k} \bigl(\prod_{i=1}^m \tfrac{ S_{\lambda \mu_i}}{S_{\lambda0}}\bigr) (S_{\lambda0})^{\chi(\Sigma)}
\end{equation}
 where  $\Lambda_+^k \subset \Lambda$  is the set of ``dominant weights of $\cG$
   (w.r.t. $\ct$ and a fixed Weyl chamber) which are integrable at level $ k - \cg$''
    where $\cg$ is the dual Coxeter number of $\cG$.
  Moreover,  $(S_{\mu \nu})_{\mu, \nu \in \Lambda_+^k}$
   is the  ``$S$-matrix'' in our situation\footnote{more precisely, the
   $S$-matrix of the WZW model associated to $G$    and the level $k - \cg$
     or, equivalently, the $S$-matrix of the modular category associated
   to  $U_q(\cG_{\bC})$    with $q:= \exp( \tfrac{2 \pi i}{k})$, cf. Sec. 1.4 in Chap. II in \cite{turaev}}.
  Finally,  $\mu_i$ is the highest weight of $\rho_i$.

   \smallskip

In the special case where $\Sigma \cong S^2$ and $m=3$ we obtain
$$Z(S^2 \times S^1,(\sigma_1,\sigma_2,\sigma_3),(\rho_1,\rho_2,\rho_3)) \sim \sum_{\lambda \in \Lambda_+^k} \frac{ S_{\lambda \mu_1}  S_{\lambda \mu_2}  S_{\lambda \mu_3} }{S_{\lambda0}}  = N_{\mu_1 \mu_2 \mu_3}$$
 where $N_{\mu_1 \mu_2 \mu_3}$ is the so-called
``Verlinde number'' associated to $(\mu_1,\mu_2,\mu_3)$.

\smallskip

As we will see later, using the approach of the present paper we can obtain  formula \eqref{eq_Gen_Verlinde}
in a rigorous (and elementary\footnote{in particular, there is no need to involve any arguments
based on the Ray-Singer torsion})
way  from the rigorous version of Eq. \eqref{eq2.48} which we will introduce below,
 cf. Sec. \ref{subsec4.9} below (cf. also Sec. \ref{subsec4.10} and Remark \ref{rm_more_general_links} below).
In fact we will evaluate the rigorous version of Eq. \eqref{eq2.48}   for a
considerably larger class of (ribbon) links, cf. Theorem \ref{main_theorem} and Remark \ref{rm_more_general_links} below.

\end{remark}

\section{The general simplicial program for CS theory \& $BF_3$-theory}
\label{sec3}

\subsection{Overview}
\label{subsec3.1}

  The goal of  what we will call\footnote{following the terminology introduced in \cite{Mn2} in the context of BF-theory} the ``simplicial program'' for   Chern-Simons theory
  is to find a  rigorous ``simplicial''  realization
of the original (or gauge fixed) Chern-Simons path integral  for the WLOs associated
to arbitrary (colored) framed (or ribbon) links in arbitrary  oriented closed  3-manifolds $M$
such that the  values of these simplicial realizations   coincide with
 the values of the corresponding Reshetikhin-Turaev invariants (see  \cite{Ati}). \par

 More concretely, if $G$ is a (simple) simply-connected compact Lie group,
  $k \in \bN$,  $M$ is a general oriented closed  3-manifold,   $L$ a colored framed (or ribbon) link in $M$, and
 if $RT_{q}(M,L)$ denotes the   Reshetikhin-Turaev invariant of $L$
  associated to  $U_q(\cG_{\bC})$  with    $q=\exp(2\pi i/k)$ or\footnote{cf. Remark \ref{rm_shift_in_k} below} $q=\exp(2\pi i/(k+\cg)$)
 then we want that the simplicial  version $\WLO_{rig}(L)$ of $\WLO(L)$
coincides with  $RT_{q}(M,L)$.  The definition of $\WLO_{rig}(L)$
is supposed to involve only  a finite triangulation (or, more generally,
a finite polyhedral cell decomposition) of $M$.\par

In the present paper we will mainly consider a weaker version of the simplicial program
where instead of demanding $RT_{q}(M,L) = \WLO_{rig}(L)$
we  demand only that there is a constant $C \in \bC$ depending on $G$, $M$, and $k$ but not on the framed link $L$
such that we have
\begin{equation} \label{eq_RT=WLO} RT_{q}(M,L) = C \cdot \WLO_{rig}(L)
\end{equation}

\begin{remark}  \label{ex_sec3.2} \rm
The main result of the present paper is a partial result in this direction\footnote{other partial results
are given in \cite{Ad0,Ad1} for Abelian groups $G$ and
in \cite{FrLo1,FrLo2,BaNai}  for the  non-Abelian non-compact  group $G= Euc(3)^c$, ie the (universal cover of the) Euclidean group  in $3$ dimensions where the Lie algebra $\cG$ of $G$ is equipped with a suitable
non-degenerate bilinear form. We remark that in  the  case  $G= Euc(3)^c$ the CS path integral is
formally equivalent to the $BF_3$-path integral with structure group $SU(2)$  and {\em vanishing}
cosmological constant, see the third paragraph after Remark  \ref{ex_sec3.2}}.
This result is concerned with the  special case $M=\Sigma \times S^1$ (where $\Sigma$ is a closed oriented surface)
and  we consider the torus  gauge-fixed version of the CS path integral
derived in Sec. \ref{subsec2.3} above, cf. Eq. \eqref{eq2.48} above.
 In Sec. \ref{sec4} below we construct a  simplicial version of the RHS of Eq. \eqref{eq2.48},
 ie we construct a simplicial version
$\WLO_{rig}(L)$ of $\WLO(L)$.
Theorem \ref{main_theorem} below shows that for the simple class of framed links considered in Sec. \ref{sec6} we have
$\WLO_{rig}(L) =  |L|/|\emptyset|$ where $|\cdot|$ is the shadow invariant for
  $M =\Sigma \times S^1$,  $\cG$ and $k$.
We remark that   $|L| = C' \cdot RT_q(\Sigma \times S^1,L)$ with $q=\exp(2\pi i/k)$
where $C' \in \bC$ is a constant which depends only on $G$, $\Sigma$, and $k$ but not on $L$
so in this special situation we have indeed  Eq. \eqref{eq_RT=WLO}
with $C:= |\emptyset|/C'$.
\end{remark}

 Instead of working with a simple simply-connected compact Lie group $G$
we can also work with general simply-connected compact Lie group $G$, cf.  Remark \ref{rm2.1} above.
In fact, as we will see later (cf. Sec. 6 and Sec. 7 in \cite{Ha7b}),
if we want to have a realistic chance of being able to generalize the aforementioned
Theorem \ref{main_theorem} to the case of general links
we will probably have to restrict ourselves to the special case where
$G$ is of the form  $G= \prod_{i=1}^r G_i$  where $r$ is even
and, moreover (up to reordering) we have $G_{2j-1} \cong G_{2j}$
and $k_{2j-1} = - k_{2j}$ for all $j \le  r/2$.
(Here $(k_i)_{i \le  r}$ is as in Remark \ref{rm2.1}).
We will consider this situation (for $r=2$) in Sec. 7 in \cite{Ha7b}. \par

In the aforementioned situation $r=2$ and $\tilde{G}:= G_1 = G_2$ and $k:= k_1 = -k_2$
the corresponding Reshetikhin-Tureav invariant coincides with the Turaev-Viro invariant associated to $U_{q}(\tilde{\cG})$.
Moreover, in this case
the (heuristic) CS path integral  is equivalent to the (heuristic)  $BF_3$ path integral with group $\tilde{G}$
and positive cosmological constant $\Lambda$ given by $\Lambda = 1/k^2$, cf. Appendix C in \cite{Ha7b}
and see also \cite{CCFM}.
Thus the  simplicial program for CS-theory (for this restricted set of groups $G$)
 can be considered as a sub program of the simplicial program for $BF_3$-theory. \par

We mention here that for Abelian structure groups the
 simplicial program for $BF_3$-theory has been essentially completed
 in  \cite{Ad0,Ad1}
 and for  $BF_3$-theory with structure group $SU(2)$  and {\em vanishing}
cosmological constant in \cite{FrLo1,FrLo2,BaNai} (cf. also \cite{Mn2}).
The explicit expression in \cite{FrLo1,FrLo2,BaNai} for the simplicial partition functions
are given in terms of the  Reidemeister torsion of $M$
and the WLOs in terms of linking numbers and the   Alexander polynomial.
We remark that the approaches in \cite{Ad0,Ad1} and \cite{FrLo1,FrLo2,BaNai}
work for a rather large (but not totally general\footnote{in \cite{Ad0}
$M$ must be a homology sphere, in \cite{FrLo2,BaNai} $M$ must fulfill a certain condition
involving the  twisted second cohomology group  })
 class of 3-manifolds $M$.

 \smallskip

The case of non-Abelian $BF_3$-theory with strictly positive (or negative) cosmological constant
is  open.  We emphasize that it is only the case of non-vanishing  cosmological constant
 which is related to the theory of  3-manifold quantum invariants.

\begin{remark} \label{rm_shift_in_k} \rm
  It was predicted in \cite{Wi} that  the value of  $\WLO(L)$ for a (colored framed)
link $L$ in $M$  (and $G$ and  $k$ as above)
is given by\footnote{we remark  that even though the rigorous version of the Reshetikhin-Turaev invariants appeared only in \cite{ReTu2,ReTu1}, Witten had already given in \cite{Wi} an algorithm for computing these invariants using arguments from conformal field theory}
 $RT_q(M,L)$ where $q= \exp(\tfrac{ 2 \pi
i}{k + \cg})$ rather than  $q= \exp(\tfrac{ 2 \pi
i}{k})$. Here $\cg$ is the dual Coxeter number  of $\cG$.
   This ``shift'' $k \to k + \cg$ was later confirmed by some authors using several different methods for the evaluation of the path integral.   However, many other authors using other gauges/regularization procedures
   did not find such a shift, see  \cite{GuMaMi,BlCo} and the references there.
   Apparently this finding caused some confusion at the time
   but if I understood it correctly it is now generally accepted
   that the occurrence (and magnitude) of such a  shift in $k$
 will depend on the regularization procedure/renormalization prescription, as suggested
 already in \cite{GuMaMi} (cf.  in particular p. 599 in Sec. 5 in \cite{GuMaMi}).

\end{remark}

\subsection{Potential Applications}
\label{subsec3.2}

 If the simplicial program for non-Abelian CS theory with simply-connected compact group $G$
 (or the related simplicial program for   non-Abelian $BF_3$-theory
 and simply-connected compact structure group  and  positive cosmological constant)
  can be carried out successfully this would probably lead to some progress
  regarding  Problem 3 of Sec. \ref{intro},   as we will now explain.

  \smallskip

Let us assume that for every closed 3-manifold $M$ and every colored, framed (or ribbon) link $L$ in $M$
we can indeed find  a natural simplicial version of the corresponding WLO
whose value coincides with the Reshetikhin-Turaev invariant.
Of course, at present this is only a hypothesis,
so most of the rest of this section is highly speculative.

\smallskip

We begin by considering the simplest situation $L=\emptyset$.
In this case  $\WLO(L)= \WLO(\emptyset)$ is  just the  partition function $Z(M)$ of $M$.
For simplicity (and in order to keep the formulas short)
we will consider first the idealistic\footnote{cf. the comment for Project 3 in  part \ref{appG} of the Appendix below} scenario based on the assumption that one
can find a simplicial realization of the  original (= non-gauge fixed) path integral expression
for $Z(M)$. In other words, we  assume  that for every
sufficiently fine  triangulation $\cC$ of the base manifold $M$
we can  find a simplicial analogue $\cA(\cC)$ of the space $\cA$
and a simplicial version  $S_{CS}^{\cC}(A)$ of the continuum action function $S_{CS}(A)$
such that\footnote{this is Eq. \eqref{eq_RT=WLO} above in the special case $L = \emptyset$ but written in a different  notation:  we write  $RT_q(M)$ instead of $RT_q(M,\emptyset)$ and $Z^{\cC}(M)$ instead of $\WLO_{rig}(\emptyset)$}

 \begin{equation} \label{eq0'} RT_q(M) =   Z^{\cC}(M)
     \end{equation}
   where
  \begin{equation}  \label{eq0}
 Z^{\cC}(M) := \int_{\cA(\cC)} \exp(i S_{CS}^{\cC}(A)) DA
     \end{equation}
is defined in a suitable way (for example as an improper integral, cf. Definition \ref{def3.2} below).

\smallskip

In fact, instead of working with a fixed triangulation $\cC$ of the base manifold $M$
  we can  consider a sequence  $(\cC_n)_{n \in \bN}$  of  triangulations
     such that each $\cC_n$ is a subdivison of $\cC_{n-1}$.
  From    Eq. \eqref{eq0'} and Eq. \eqref{eq0} above we then obtain
  \begin{equation} \label{eq_pre_doublelimit}
   RT_q(M) = Z^{\cC_n}(M) = \int_{\cA(\cC_n)} \exp(i S_{CS}^{\cC_n}(A)) DA
   \end{equation}
     Observe that  $RT_q(M)$ does not depend on $n$, so by applying
     a $n \to \infty$-limit on the LHS and RHS of Eq. \eqref{eq_pre_doublelimit} we simply obtain
  \begin{equation} \label{eq_doublelimit}
   RT_q(M) = \lim_{n \to \infty} \int_{\cA(\cC_n)} \exp(i S_{CS}^{\cC_n}(A)) DA
   \end{equation}

In a similar way one can write down analogues of Eq.
  \eqref{eq_doublelimit} for the gauge-fixed CS path integral and
  analogues where non-empty links are included.
  Moreover,  we can also study suitable perturbative analogues of Eq.
  \eqref{eq_doublelimit} and doing so might  lead to some progress regarding Problem 3 of Sec. \ref{intro},
  cf. the discussion after Example \ref{ex_sec3.3} below.

\begin{remark} \label{rm_sec3.3} \rm
 Before studying in more detail a concrete perturbative analogue of Eq.   \eqref{eq_doublelimit}
 (namely, Eq. \eqref{eq_doublelimit_pert} below)  let us briefly consider the following instructive,  non-perturbative, heuristic argument first: \par

 Observe that   if the simplicial space $\cA(\cC_n)$ was chosen
naturally, there will be a canonical projection
$R_n:\cA \to \cA(\cC_n)$ (the ``de Rham map'') and on an informal level the normalized Lebesgue measure $DA$ on $\cA(\cC_n)$
appearing on the RHS of Eq. \eqref{eq_doublelimit} will be equal to the $R_n$-image of the (heuristic)
Lebesgue measure $DA$ on $\cA$  (up to a multiplicative constant).
So, after applying on an informal level the transformation theorem of measure theory
we obtain
$$\int_{\cA(\cC_n)} \exp(i S_{CS}^{\cC_n}(A)) DA =
      \int_{\cA} \exp(i S_{CS}^{\cC_n}(R_n(A))) DA$$
If the simplicial action function was defined naturally
we should have
$$\lim_{n \to \infty} S_{CS}^{\cC_n}(R_n(A)) =  S_{CS}(A)$$
 so by  interchanging the two limit procedures $\lim_{n \to \infty} \cdots$
 and $\int \cdots DA$  we obtain, informally,
 \begin{multline} \label{eq_doublelimit_b}
   \lim_{n \to \infty} \int_{\cA(\cC_n)} \exp(i S_{CS}^{\cC_n}(A)) DA
   =   \lim_{n \to \infty} \int_{\cA} \exp(i S_{CS}^{\cC_n}(R_n(A))) DA\\
    =    \int_{\cA}  \lim_{n \to \infty} \exp(i S_{CS}^{\cC_n}(R_n(A))) DA =
    \int_{\cA} \exp(i S_{CS}(A)) DA
   \end{multline}
 Clearly, the combination of  equation Eq. \eqref{eq_doublelimit}
 and Eq. \eqref{eq_doublelimit_b}
 provides a direct link between $RT_q(M)$
 and Witten's  path integral expression $ \int_{\cA} \exp(i S_{CS}(A)) DA$. \par

The last two expressions in  Eq. \eqref{eq_doublelimit_b} are only heuristic
but  it is possible that already  Eq. \eqref{eq_doublelimit}  and Eq. \eqref{eq_doublelimit_b}
could be exploited to obtain a better understanding of the heuristic
CS path integral and to obtain new mathematical conjectures and results.
Let us illustrate this point by considering
 the following example
which is, in fact, very closely related\footnote{as observed in Sec. 2 in \cite{Wi}
 the Ray-Singer torsion appears naturally when studying  the
  semi-classical limit of the heuristic CS path integral}  to Chern-Simons theory.
 \end{remark}

\begin{example} \label{ex_sec3.3}
Let $X$ be a closed, connected, oriented, smooth manifold of odd dimension and
let  $\rho:\pi_1(X) \to O(N)$, $N \in \bN$,   be a group homomorphism such that
the Reidemeister torsion  $T(X,\rho)$ associated to $X$ and $\rho$ is well-defined\footnote{a sufficient condition
for this being the case is that the de Rham complex associated to the $A_{\rho}$-twisted exterior derivative
is acyclic (with  $A_{\rho}$ given below)}.
  Let $(\cC_n)_n$ be a sequence of smooth triangulations of  $X$  such that each $\cC_n$ is a subdivison of $\cC_{n-1}$.    For every $n \in \bN$ we have (cf. Proposition 1.7 in \cite{RaSi})
  $$T(X,\rho) = \prod\nolimits_{k=0}^{\dim(X)} \det(-\triangle_k(\cC_n,\rho))^{(-1)^{k+1} k/2} $$
       where each $\triangle_k(\cC_n,\rho)$, $k=0,\ldots,\dim(X)$, is the
     ``$\rho$-twisted'' combinatorial Laplacian on the space of $k$-cochains
         associated to $\cC_n$.
   Since the LHS of the last equation
    does not depend on $n$ we trivially have
  \begin{equation} \label{eq3}
  T(X,\rho) = \lim_{n \to \infty} \prod\nolimits_{k} \det(-\triangle_k(\cC_n,\rho))^{(-1)^{k+1} k/2}
  \end{equation}
  As in Eq. \eqref{eq_doublelimit_b} above one can now try to informally
interchange the  limit $n \to \infty$  in Eq. \eqref{eq3} with each of the determinants\footnote{\label{ft_interchange} Clearly, since  the determinants can be  rewritten
 in terms of a (Gaussian) integral the situation in the present example
  is indeed very similar to the situation in Eq. \eqref{eq_doublelimit_b} above} $\det(\cdot)$.
 By doing so one is led quite naturally to
  Ray and Singer's famous conjecture (=  the Cheeger-Mueller Theorem), cf. Eq. \eqref{eq_RaySInger} below.
Indeed, if we fix an  auxiliary Riemannian metric $g$ on $X$
and denote by  $E$ the  flat vector bundle on $X$ associated to $\rho$, and
by  $A_{\rho}$ the  flat connection on $E$ associated to $\rho$
  then it is natural to conjecture that
   for a suitable notion of limit
    the    ${n \to \infty}$-limit of   $\triangle_k(\cC_n,\rho)$
   will be  the ``$A_{\rho}$-twisted'' Hodge Laplacian  $\triangle_k(g,A_{\rho}):\Omega^k(E) \to \Omega^k(E)$
      where   $\Omega^k(E)$ is  the space of $k$-forms with values in $E$.
  Of course, in order to make sense of the determinants $\det(-\triangle_k(g,A_{\rho}))$
  we will now have to work a bit  harder, and, e.g.,
perform a suitable zeta-function regularization. After doing so one arrives at
 Ray and Singer's ``analytic torsion'' $\tau(X,\rho)$ and their conjecture
 \begin{equation} \label{eq_RaySInger} T(X,\rho) = \tau(X,\rho)\end{equation}
  which was later proven independently in \cite{Che} and in \cite{Mue}.
The approximation idea sketched above could also be used as a strategy for proving Eq. \eqref{eq_RaySinger}
and, apparently, the proof in \cite{Mue} (which is partly based on \cite{Dod,Pat})
 makes use of this strategy.
\end{example}

So far we have restricted ourselves to
the original CS path integral  and to the situation of empty links
but -- as mentioned above -- we can easily generalize the arguments above
to the gauge-fixed CS path integral and to the situation
of non-empty links. If the simplicial program can be carried out successfully it is plausible
to expect that at least some progress towards the solution of Problem 3 of Sec. \ref{intro}
can be made (cf. also the list of problems in Sec. 7 in \cite{Oht}). \par

Let us be a bit more explicit and consider, for simplicity\footnote{cf. footnote \ref{ft_ASE} below},
the special case  $M=S^3$.
Assume that $L$ is a fixed (ribbon) link in $S^3$. We now choose  a sequence  $(\cC_n)_{n \in \bN}$ of
finite polyhedral cell decompositions of $S^3$ such that $\cC_{n}$ is a subdivision of $\cC_{n-1}$ for each $n \ge 2$
and such that for each $n \in \bN$ there is a ``simplicial ribbon link'' $L_n$ in $\cC_n$
(cf. Sec. \ref{subsec4.0.3b} below)
so that the sequence $(L_n)_{n \in \bN}$ approximates $L$ in a suitable way.
Finally, we assume that  for each $n \in \bN$ we can find  a rigorous simplicial realization
  $\WLO^{\cC_n}_{rig}(L_n)$ of the (original or gauge-fixed)   heuristic path integral $\WLO(L)$ such that
\begin{equation} \label{eq_pre_doublelimit2} RT_{q}(S^3,L) = \WLO^{\cC_n}_{rig}(L_n)
\end{equation}
and, in particular,  $RT_{q}(S^3) = Z^{\cC_n}(S^3):= \WLO^{\cC_n}_{rig}(\emptyset)$.
    If this is possible, then  after normalizing\footnote{cf. again footnote \ref{ft_ASE} below}
    we obtain
\begin{equation} \label{eq_RT=WLO_S3} \frac{RT_{q}(S^3,L)}{RT_{q}(S^3)} =
\frac{ \WLO^{\cC_n}_{rig}(L_n) }{Z^{\cC_n}(S^3)}
\end{equation}
Let us now  expand\footnote{\label{ft_ASE} If we do not normalize (cf. the preceding footnote)
 or if we consider the case $M \neq S^3$ then the situation becomes more complicated and  a simple asymptotic
 expansion in terms of (non-negative) integer powers of $1/k$  will in general not be possible.
 Instead one will then have to rewrite the argument/discussion above using a more complicated form of asymptotic expansion, cf., e.g., the RHS of Conjecture 7.7 in  \cite{And} which deals with (general $M$ and)
  the special case  $L= \emptyset$.}
   the LHS and the RHS of Eq. \eqref{eq_RT=WLO_S3}
 as an (asymptotic) series of powers of $1/k$.
  Let $c_m$, $m \in \bN_0$, be the coefficient appearing in front of the power $(1/k)^m$
in the asymptotic expansion (as $k \to \infty$) of $RT_q(S^3,L)/RT_q(S^3)$. Similarly, let $d_m^{(n)}$ be the coefficient appearing in front of the power $(1/k)^m$ in the asymptotic expansion (as $k \to \infty$)
of $\WLO^{\cC_n}_{rig}(L_n)/Z^{\cC_n}(S^3)$.
Clearly, for each fixed $m \in \bN_0$ we have $c_m = d_m^{(n)}$ for every $n \in \bN$
so by applying the (trivial) limit ${n \to \infty}$ we obtain
\begin{equation} \label{eq_doublelimit_pert}
c_m = \lim_{n \to \infty} d_m^{(n)}
\end{equation}
Using standard techniques from perturbative QFT we can hope to be able
to evaluate $ d_m^{(n)}$
explicitly\footnote{and even rigorously, since  $\WLO^{\cC_n}_{rig}(L_n)$ is a
 finite dimensional (oscillatory) integral} for each fixed $n \in \bN$
 and to find  a suitable topological/geometric interpretation of
the ``continuum expression'' $d_m := \lim_{n \to \infty} d_m^{(n)}$.  \par
This might have applications to the theory of Vassiliev invariants.
Recall that one  of the central results in knot theory is Kontsevich's discovery (cf. \cite{Kon}) of the
so-called universal Vassiliev invariant (for $M=S^3$ or $M=\bR^3$). There are several different versions of
the universal Vassiliev invariant, see \cite{BarStoi} where these versions are called the ``combinatorial'' version, the ``analytic'' version (involving the so-called Kontsevich integral, cf. \cite{Kon}),
the ``geometric''  version (involving so-called ``configuration space integrals'', see \cite{AlFr,Thu}), and the ``algebraic'' version (cf. \cite{Car,LeMu,Piu}).\par

As is explained, e.g. in  \cite{CCFM,Lab}, the perturbative evaluation of the heuristic CS path integral\footnote{
or $BF_3$-theory with simply-connected compact structure group and positive cosmological constant}
 with  $M=S^3$ or $M=\bR^3$ in different gauges is very closely related
  to the last three of these four versions:
\begin{itemize}
\item the geometric version arises naturally after applying Lorentz gauge fixing to the heuristic CS path integral
for the WLOs (with $M=S^3)$, see
 \cite{GMM,BN,AxSi1,AxSi2,BoTa,AlFr,Thu}.

\item the analytic version arises essentially\footnote{some subtleties are
not yet clear from the QFT point of view, e.g. it is not clear how  the so-called ``Kontsevich factors''  can be obtained from the CS path integral}  after applying light-cone gauge fixing (in a complexified setting) to the heuristic CS path integral
for the WLOs (with $M=\bR^3)$, see \cite{FK}.

\item the algebraic version was  conjectured in \cite{Lab}
to arise after applying temporal/axial  gauge fixing (in a non-complexified setting) to the heuristic CS path integral for the WLOs (with $M=\bR^3)$.
It seems more likely, though,  that a successful derivation of the
  algebraic version of the universal Vassiliev invariant
from the CS path integral will involve
 torus gauge fixing in some way\footnote{either applied in the situation $M=S^2 \times S^1$ and followed
 by a surgery argument, or alternatively, by applying  torus gauge fixing directly to the manifold $M=S^3$
 considered as a $S^1$-bundle, cf. \cite{BlTh4} and Sec. \ref{sec7} below}.
\end{itemize}

 Having a rigorous simplicial  version of the Chern-Simons path integral
  should increase the chances for finding a  unified treatment
   for the last three versions of the universal Vassiliev invariant.

\section{A ``simplicial'' differential geometric framework}
\label{sec4.0}

As a preparation for Sec. \ref{sec4} below we will  now introduce
 a suitable  ``simplicial'' differential geometric framework
 which was inspired by the work of Adams on Abelian $BF_3$-theory (cf. \cite{Ad0,Ad1}).
 Two important differences in comparison to \cite{Ad0,Ad1} are:
 \begin{itemize}
 \item the introduction of  the cell complex $q\cK$, cf. Sec. \ref{subsec4.0.4c} below
 \item the introduction of the  notion of a ``simplicial (closed) ribbon'', cf. Sec. \ref{subsec4.0.3b} below.
 \end{itemize}

\subsection{Chains and cochains}
\label{subsec4.0.1}

Let $d \in \bN$ and let  $\cK$ be an  oriented $d$-dimensional
simplicial complex. For $p \in \{0,1, \ldots, d  \}$ we will denote by   $\face_p(\cK)$ the set of  $p$-faces in $\cK$.
For every   $p
\in \{0,1, \ldots, d \}$ we will denote by $C_p(\cK)$  the space
of ``$p$-chains of $\cK$ with coefficients in $\bR$'' and by
$C^p(\cK,\bR)$ the space of ``$p$-cochains with values in $\bR$''\footnote{so  $C^p(\cK,\bR)$ is simply the space
 $\bR^{\face_p(\cK)}$ of all maps $\face_p(\cK) \to \bR$}.
As usual we will denote by
$\partial_{\cK}: C_p(\cK) \to C_{p-1}(\cK)$,  $p \in \{1,2,
\ldots, d \}$, the corresponding
``boundary operator'' and  by
$ d_{\cK} :C^p(\cK,\bR) \to C^{p+1}(\cK,\bR)$,   $p \in \{0,2,
\ldots, d-1 \}$ the corresponding ``coboundary operator''.

\smallskip

From now on we will assume that
  $\cK$ is finite and  we will make the identification
$C_p(\cK) \cong C^p(\cK,\bR) $ for each $p \in \{0,1,2,\ldots,d\}$
in the obvious way. \par

 By $\langle \cdot, \cdot \rangle_p := \langle \cdot, \cdot \rangle_{C_p(\cK)}$ we will denote the standard scalar product  on $C_p(\cK) \cong C^p(\cK,\bR) =  \bR^{\face_p(\cK)}$. Observe that the linear maps
$ d_{\cK} :C^p(\cK,\bR) \to C^{p+1}(\cK,\bR)$ and $\partial_{\cK}:C_{p+1}(\cK) \to C_{p}(\cK)$
 are dual to each other w.r.t. the scalar products
$\langle \cdot, \cdot \rangle_p$ and $\langle \cdot, \cdot \rangle_{p+1}$
i.e. we have
\begin{equation} \label{eq_d_vs_partial} \langle d_{\cK} \alpha, \beta \rangle_{p+1} =  \langle
\alpha, \partial_{\cK} \beta \rangle_p
\end{equation}
 for all $\alpha \in C_p(\cK) \cong C^p(\cK,\bR)$ and $\beta \in C_{p+1}(\cK) \cong
C^{p+1}(\cK,\bR)$.

\smallskip

Now let $V$ be a finite-dimensional real vector space and
let   $C^p(\cK,V)$, for $p \in \{0,1, \ldots, d  \}$,
denote the space of ``$p$-cochains with values in $V$''\footnote{so  $C^p(\cK,V)$ is the space
$V^{\face_p(\cK)}$ of all maps $\face_p(\cK) \to V$}.
Then we can make the identification
$$C^p(\cK,V) \cong C^p(\cK,\bR) \otimes_{\bR} V \cong C_p(\cK) \otimes_{\bR} V$$
We will denote the linear operators
$ d^V_{\cK} :C^p(\cK,V) \to C^{p+1}(\cK,V)$ and $\partial^V_{\cK}:C_{p+1}(\cK,V) \to C_{p}(\cK,V)$
given by $ d^V_{\cK} := d_{\cK} \otimes \id_V$ and $ \partial^V_{\cK} := \partial_{\cK} \otimes \id_V$
 simply by $d_{\cK}$ and  $\partial_{\cK}$ in the following.
 Observe that if we equip $V$
with a  scalar product $\langle \cdot, \cdot \rangle$
then this induces  scalar products $\langle \cdot, \cdot \rangle_p$ on $C^p(\cK,V)$, $p \in \{0,1, \ldots, d  \}$,
and  Eq. \eqref{eq_d_vs_partial} above will then  hold
 for all $\alpha \in  C^p(\cK,V)$ and $\beta \in C^{p+1}(\cK,V)$.

\begin{convention} \rm \label{conv_alpha=deltaalpha}
 For each $p \in \{0,1, \ldots, d  \}$ we identify $\face_p(\cK)$
  with its image in $  C_p(\cK) \cong C^{p}(\cK,\bR)$
  under the injection $\face_p(\cK) \ni \alpha \mapsto \delta_{\alpha} \in C^{p}(\cK,\bR)$
  where  $\delta_{\alpha} \in C^{p}(\cK,\bR)$
is given by $\delta_{\alpha}(\alpha') = \delta_{\alpha,\alpha'}$ for all  $\alpha' \in \face_p(\cK)$
\end{convention}

\smallskip

 The constructions and definitions above can be generalized in a straightforward way to
the situation where instead of a (finite) oriented simplicial complex $\cK$ we work
with a (finite) oriented
 ``polyhedral cell complex''  $\cP$, i.e.  a cell complex
which is obtained by glueing together bounded ``convex
polytopes''    in an analogous way as simplicial complexes arise from glueing
together simplices,  cf. part \ref{appF} of the Appendix for the formal definitions.\par
Polyhedral cell complexes arise naturally, e.g.  as the duals and
products of simplicial  complexes (cf.  part \ref{appF} of the Appendix).

\subsection{Simplicial curves, loops, and links}
\label{subsec4.0.3}

Let $\cP$ be a fixed oriented polyhedral cell complex\footnote{the only seven cases which will be relevant
later are $\cP \in \{\bZ_N, \cK,\cK',q\cK,\cK \times \bZ_N,\cK' \times \bZ_N,q\cK \times \bZ_N\}$
with $\bZ_N$ and $\cK$ as in Sec. \ref{subsec4.0.4}}.
We will call the elements of $\face_0(\cP)$ (resp. $\face_1(\cP)$
resp.  $\face_1(\cP) \cup \{0\} \cup (- \face_1(\cP))\subset
C_1(\cP)$) the ``vertices'' (resp. ``edges'' resp. ``generalized
edges'') in $\cP$. The element $0 \in C_1(\cP)$ will be called ``the
empty edge''.\par

Using the orientations on each of the edges
we can define the starting point $\startpoint(e)$ and the endpoint $\epoint(e)$ of an edge
$e \in \face_1(\cP)$ in the obvious way. For $e \in - \face_1(\cP)$ we set $\startpoint(e) := \epoint(-e)$
and  $\epoint(e) := \startpoint(-e)$.

\medskip

A  ``simplicial curve''\footnote{by abuse of language we use the term ``simplicial curve'' also
when  $\cP$ is  not  a simplicial complex but a general polyhedral cell complex}  in $\cP$ is a finite sequence $x=
(x^{(k)})_{k \le n}$, $n \in \bN$, of  vertices in $\cP$
 such that for every $k \le n-1$
the two vertices
 $x^{(k)}$ and  $x^{(k+1)}$
   either coincide or are the two endpoints
 of an edge $e \in \face_1(\cP)$.
 If $x^{(n)} = x^{(1)}$ we will call
$x= (x^{(k)})_{k \le n}$ a ``simplicial loop'' in $\cP$.

\medskip

Every simplicial curve  $x= (x^{(k)})_{k \le n}$ with $n > 1$ induces a
sequence $(e^{(k)})_{k \le n-1}$ of generalized edges given by
$$e^{(k)} =
\begin{cases}
e & \text{ if } x^{(k)} = \startpoint(e) \\
0 & \text{ if } x^{(k)} = x^{(k+1)} \\
-e & \text{ if } x^{(k)} = \epoint(e)  \\
\end{cases}
$$
where in the first and in the last case $e \in \face_1(\cP)$ is the unique edge with
$\{\startpoint(e), \epoint(e)\} = \{x^{(k)},x^{(k+1)} \}$. In the
following we will mostly\footnote{note that these two points of view
are almost equivalent the exception being the case where $x=
(x^{(k)})_{k \le n}$ is constant. In this (degenerate) case $x=
(x^{(k)})_{k \le n}$ can, of course,  not be recovered from $e=
(e^{(k)})_{k \le n-1}$.}
 take the ``generalized edge point of
view'', i.e. when we mention simplicial curves
we consider them as sequences $e= (e^{(k)})_{k}$
of generalized edges and write
 $\start e^{(k)}$ for the
corresponding  vertex  $x^{(k)}$.

\smallskip

Observe that every simplicial loop $l= (l^{(k)})_{k \le n}$  in
a polyhedral cell complex $\cP = (X,\cC)$ induces, in an obvious way,
a continuous loop $[0,1] \to X$, which will also be denoted by $l$.

\smallskip

Finally, let us specialize to the situation where the topological space $X$
appearing in  $\cP = (X,\cC)$ is a 3-dimensional topological manifold.
A ``simplicial link'' in $\cP$
is then defined to be a finite tuple $L=(l_1,l_2, \ldots, l_m)$, $m \in \bN$,
of simplicial loops in $\cP$ such that
the corresponding tuple of continuous loops in $X$ is a link in $X$\footnote{in particular, each $l_i$ neither intersects itself nor any of the other $l_j$, $j \neq i$}.

\subsection{Simplicial ribbons and ribbon links}
\label{subsec4.0.3b}

In version (v3) of the present paper and version (v2) of \cite{Ha7b}  (ie in arXiv:1206.0439v3 and in arXiv:1206.0441v2)   we worked  with simplicial loops $l_i$, $i \le m$, in $ K_1 \times \bZ_N$  which came with suitable ``framings'' $l'_i$.  This point of view, to which we will refer below
 as the ``loop pair point of view'',
  was motivated by \cite{Ad0,Ad1}. It is fairly general and works very well in the case of Abelian CS-theory/$BF_3$-theory.\par  However, Remark 5.3 in Sec. 5.2  of arXiv:1206.0441v2 suggests that
 for non-Abelian CS-theory/$BF_3$-theory
   it is more natural to work within a more restricted setting where for each $i \le m$ the two  loops  $l_i$ and $l'_i$ are the boundary loops of a ``(closed) simplicial ribbon''
   in the sense of Definition \ref{def_simpl.ribbon} below.

\smallskip

In order to motivate the notion of a ``(closed) simplicial ribbon''
let us first consider its continuum analogue:

\begin{definition} \label{def_cont.ribbon} Let $M$ be a topological space.
A closed ribbon $R$ in $M$ is an embedding
$R: S^1 \times [0,1] \to M$
\end{definition}

We remark  that closed ribbons (often called ``annuli'') are frequently used
in knot theory  cf., eg, Sec. 2.2. in Chap. I in \cite{turaev}.
In fact, the notion of a  closed ribbon in a 3-dimensional manifold $M$ is essentially equivalent to
the notion of a ``framed knot''\footnote{\label{ft_framing} a framed knot in $M$ can be defined as a pair $(l,X)$
where $l:S^1 \to M$ is a smooth embedding and $X: \arc(l) \to TM$
is a smooth vector field such that for each $p \in \arc(l)$ the vector $X_p$ is not tangent to $\arc(l)$.
Alternatively, we could define a framed knot in $M$ to be a pair $(l,l')$
where $l:S^1 \to M$ and $l':S^1 \to M$ are two non-intersecting embeddings which are homotopic to each other.
The second definition is rarely used. We mention it as a motivation for the
``loop pair point of view'' mentioned above}
in $M$.

\begin{definition} \label{def_simpl.ribbon} Let $\cP$ be a polyhedral cell complex.
A closed simplicial  ribbon $R$ in $\cP$ is a finite sequence $R = (F_i)_{i \le n}$ of 2-faces of $\cP$
such that
\begin{enumerate}

\item[(SR1)] Every $F_i$ is a tetragon

\item[(SR2)] $F_i \cap  F_{j} = \emptyset$ unless $i = j$ or $j= i \pm 1$ (mod n).
In the latter case  $F_i$ and $F_{j}$ intersect in  a (full) edge.
\end{enumerate}

In the following we will often  say just ``simplicial  ribbon'' instead of ``closed simplicial  ribbon''.
\end{definition}

\begin{figure}[h]
  \centering
  \begin{minipage}[b]{4.5 cm}
  \includegraphics[height=4cm,width=4cm, angle=0]{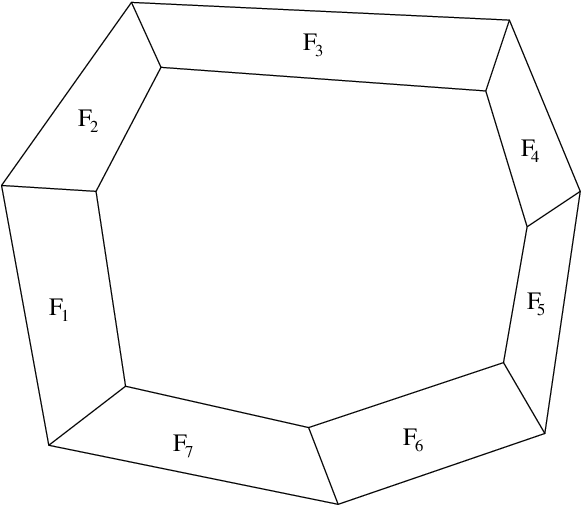}
    \caption{}     \label{fig_D1}
   \end{minipage}
  \begin{minipage}[b]{0.5 cm}
   \ \
    \end{minipage}
   \begin{minipage}[b]{4.5 cm}
   \includegraphics[height=4cm,width=4cm, angle=0]{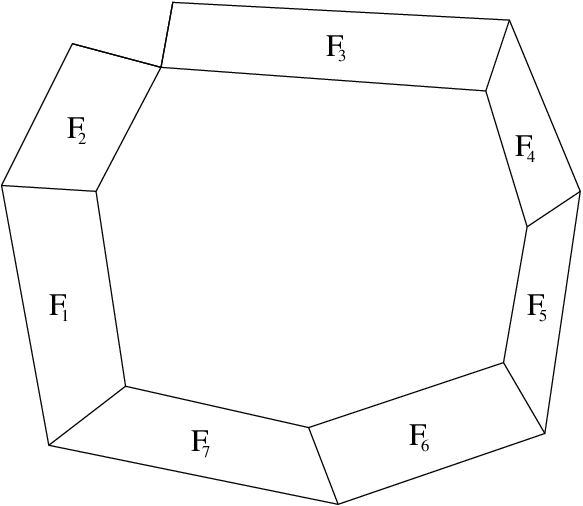}
    \caption{}     \label{fig_D3}
  \end{minipage}
  \begin{minipage}[b]{0.5 cm}
   \ \
    \end{minipage}
 \begin{minipage}[b]{4.5 cm}
  \includegraphics[height=4cm,width=4cm, angle=0]{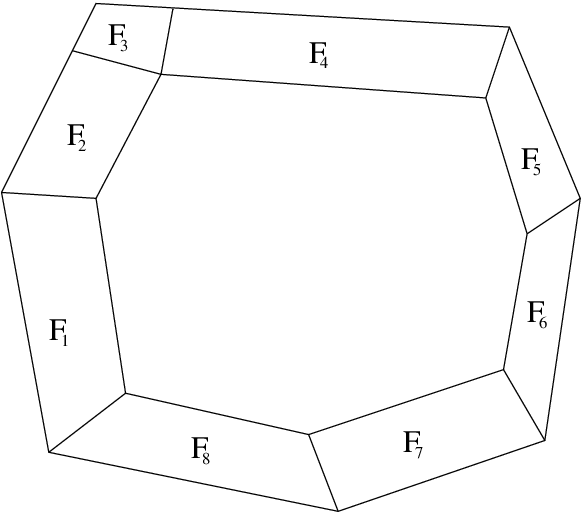}
    \caption{}     \label{fig_D2}
   \end{minipage}
\end{figure}

Fig. \ref{fig_D1} shows an example for a closed simplicial ribbon in 2 dimensions.
On the other hand Fig. \ref{fig_D3} and Fig. \ref{fig_D2}
  show  non-examples (where condition (SR2) is not fulfilled).
We remark that  condition (SR2) implies that the intersection of $F_i$ and  $F_{i+1}$
 is the edge that lies ``opposite'' to the edge which is the intersection of $F_{i-1}$ and  $F_{i}$.
Obviously, in Fig. \ref{fig_D2} this is not the case for the index $i=3$.

\begin{remark} \label{rm_appK_0} \rm
\begin{enumerate}

\item Observe that if $M$ is a topological space and $\cC$ a polyhedrical cell decomposition of $M$ then
every closed simplicial  ribbon $R$ in $\cP := (M,\cC)$ can be considered in a natural way
as a closed ribbon $S^1 \times [0,1]  \to M$. (This is  analogous to the fact that we can consider
every simplicial loop $l$ in $\cP$ in a natural way as  a ``continuous loop'', i.e. as
a continuous map $S^1 \to M$).

\item We observe also that every simplicial ribbon $R$ in $\cP$ induces a pair of simplicial loops $(l,l')$ in $\cP$ in the obvious way ($l$ and $l'$ are the loops which lie on the boundary of $R$).

\item One could try to generalize  Definition  \ref{def_simpl.ribbon}
by weakening one or both of the conditions (SR1) and (SR2),
for example by dropping the condition implicit in (SR2)
that the  two edges  $e := F_i \cap F_{i+1}$ and $e' := F_{i-1} \cap F_{i}$
must not touch each other (so that the example in Fig. \ref{fig_D2} is no longer excluded).
However, the reader will probably agree that such a generalized definition is  less natural
than the original one\footnote{In particular,  in the situation of part i) of the present remark
each such generalized simplicial ribbon would again induce  a
closed ribbon $S^1 \times [0,1]  \to M$ but this time the explicit formula would
be rather ugly}. Even more importantly, it turns out that if we weaken condition (SR2)
in the aforementioned way Theorem \ref{main_theorem} below will no longer
be valid.
\end{enumerate}
\end{remark}

\begin{definition} \label{def_appK_RL} \rm Let $M$ be a 3-dimensional topological manifold and $\cC$ a polyhedrical cell decomposition
of $M$.  A ``simplicial ribbon link'' L  in $\cP := (M,\cC)$
is a tuple $L=(R_1, \ldots, R_m)$, $m \in \bN$, of closed simplicial ribbons $R_i$
which do not intersect each other.
 (Here we consider each $R_i$ as a map $[0,1] \times S^1 \to M$, cf part i) of Remark \ref{rm_appK_0} above).
 \end{definition}

\begin{remark} \label{rm_appK_1} \rm
\begin{enumerate}

\item  Obviously, only for special polyhedral cell complexes $\cP$ simplicial ribbons will exist.
For example, if $\cP$ is a simplicial complex no 2-face will be a tetragon so
condition (SR1) above will never be fulfilled.
For certain polyhedral cell complexes $\cP$ like, eg,  $\cP = q\cK$ or $\cP = q\cK \times \bZ_N$
appearing in Sec. \ref{subsec4.0.4} below,
 every 2-face is a tetragon, so in this special situation condition (SR1) in Definition \ref{def_simpl.ribbon} above  is automatically fulfilled.
 However, this does not mean that
  the existence question  of simplicial ribbons in such polyhedral cell complexes is a trivial issue.
  In fact, as a result of condition (SR2) above   certain complications can arise.
 For example, for a simplicial loop $l$ in $q\cK$  we can  in general not find a simplicial ribbon $R$ in $q\cK$
such that $l$ is one of the two boundary loops of $R$.

\smallskip

\item On the other hand for every  smooth  link $L = (l_1, l_2, \ldots l_m)$, $m \in \bN$, in $\Sigma \times S^1$   and every $\eps > 0$ we can always find
 a $N \in \bN$ and (smooth) polyhedral cell complex $\cK$ on $\Sigma$ such that in $q\cK \times \bZ_N$ there exists a simplicial ribbon link $L^{disc}$ which is an ``$\eps$-approximation'' of $L$ (w.r.t to a fixed auxiliary Riemannian metric ${\mathbf g}$ on $\Sigma$)  in a suitable sense.

 Moreover, if $\eps$ was chosen sufficiently small\footnote{and if the notion of ``$\eps$-approximation''  mentioned above is sufficiently strong} then $L^{disc}$, considered\footnote{recall that each of the simplicial ribbons appearing in $L$  can be considered as a closed ribbon in  $\Sigma \times S^1$ (cf. part i) of Remark \ref{rm_appK_0} above) or, equivalently,  as a framed knot in  $\Sigma \times S^1$, cf. the paragraph after Definition \ref{def_cont.ribbon} above} as a framed (piecewise smooth) link in $\Sigma \times S^1$,  will be equivalent to the framed link which we obtain by equipping each of the knots $l_i$, $i \le m$, appearing in  $L$  with  a ``horizontal framing''\footnote{ie for each $l_i$ the corresponding vector field $X:\arc(l_i) \to TM$ where $M= \Sigma \times S^1$   mentioned in footnote \ref{ft_framing}  has  the following property: for each $p \in \arc(l_i)$ the vector  $X_p$ is not tangent to $\arc(l_i)$ and, secondly,  $X_p$ is  ``parallel to $\Sigma$'', ie   the projection $(\pi_{S^1})_*(X_p) \in T_{\pi_{S^1}(p)} S^1$ is zero}.
  \end{enumerate}
\end{remark}

\subsection{Some special (polyhedral) cell complexes}
\label{subsec4.0.4}

\subsubsection{The cell complex $\bZ_N$}
\label{subsec4.0.4a}

For the rest of this paper we fix a natural number $N$ and
we will denote by $\bZ_N$  the cyclic group of order $N$.
We will identify $\bZ_N$ with the  subgroup
$ \{ e^{\frac{2 \pi i}{N} k} \mid 1 \le k \le N\}$  of the Lie group $S^1 \cong \{z \in \bC \mid |z|=1\} $. \par

The points in $\bZ_N$
 induce\footnote{more precisely, the 1-cells of $\bZ_N$ are the connected components of $S^1 \backslash \bZ_N$}
  a (polyhedral) cell decomposition  of $S^1$
 and the corresponding (1-dimensional polyhedral) cell complex will also be denoted
  by $\bZ_N$ in the following.
  We will equip $S^1$ with the standard  orientation (= the 1-form $dt$).
  Moreover, we choose for each 1-cell of  the (1-dimensional polyhedral) cell complex  $\bZ_N$
  the orientation which is induced by the orientation on $S^1$.

\subsubsection{The cell complexes $\cK=K_1$ and $\cK'=K_2$}
\label{subsec4.0.4b}

Recall that above we fixed an oriented compact surface $\Sigma$.
For the rest of this paper we will now
fix a finite  polyhedral cell decomposition
 $\cC$ of $\Sigma$. By $\cC'$ we will denote  the  canonical dual of $\cC$,
 cf.  part \ref{appF} of the Appendix.
 We turn $\cC$ and $\cC'$ into oriented polyhedral cell decompositions by picking\footnote{if $\cC$ is a smooth
 cell decomposition it is natural to assume that the orientation on the 2-cells of $\cC$
 comes from the orientation $\nu_{\Sigma}$ of $\Sigma$ and that the orientation of all the cells of $\cC'$
 are the ones induced by the cell orientations of $\cC$ and $\nu_{\Sigma}$} an orientation for each cell of
 $\cC$ and of $\cC'$.

   \smallskip

 Let  $\cK$ and $\cK'$  denote  the oriented polyhedral
 cell complexes corresponding to $\cC$ and $\cC'$
  (i.e.  $\cK = (\Sigma,\cC)$ and $\cK' = (\Sigma,\cC')$).
 Below we will often write
$K_1$ instead of $\cK$ and $K_2$ instead of $\cK'$.

\subsubsection{The cell complex $q\cK$}
\label{subsec4.0.4c}

Since   $\cC'$ was chosen to be the canonical dual of $\cC$
the barycentric subdivision $b\cC$ of $\cC$ coincides with the one of  $\cC'$.
We will, however, not work with the cell complex $b\cK:= (\Sigma,b\cC)$
 but  with the ``coarser'' polyhedral cell complex $q\cK:= (\Sigma,q\cC)$
which is determined by
\begin{itemize}
\item $\face_0(q\cK) = \face_0(b\cK)$
\item $\face_1(q\cK) = \face_1(b\cK) \backslash \{e \in \face_1(b\cK) \mid \text{ both endpoints of $e$
lie in $\face_0(K_1) \sqcup  \face_0(K_2)$} \}$,
\end{itemize}
The set $\face_2(q\cK)$ is uniquely determined by $\face_0(q\cK)$ and $\face_1(q\cK)$.
Observe that each $F \in \face_2(q\cK)$ is a tetragon
 and that each $F$ is the union of exactly two faces of $\face_2(b\cK)$.\par

 We turn $q\cK$  into an oriented polyhedral cell complex by picking
 an orientation for each of the cells of $q\cK$.
 For simplicity    we assume that the orientation on each edge $e \in \face_1(q\cK)$
 is the one that comes from the orientation on the unique  edge
 $e' \in \face_1(K_1) \cup \face_1(K_2)$ in which $e$ is contained.

 \begin{figure}[h]
\begin{center}
\begin{minipage}[b]{6.5 cm}
\begin{example}
Fig. \ref{fig_sec3} gives an example.
In this example the 2-faces of the cell complex $q\cK$ are the 16 smaller faces that appear there.
The original polyhedral cell complex $\cK$ from which $q\cK$ is derived
has five 2-faces, namely one square, three equilateral triangles
and another triangle. These five faces are bounded by fat lines.
\end{example}
\end{minipage}
\begin{minipage}[b]{1.5 cm}
\ \
\end{minipage}
\begin{minipage}[t]{6.5 cm}
  \includegraphics[height=4 cm,width=6 cm]{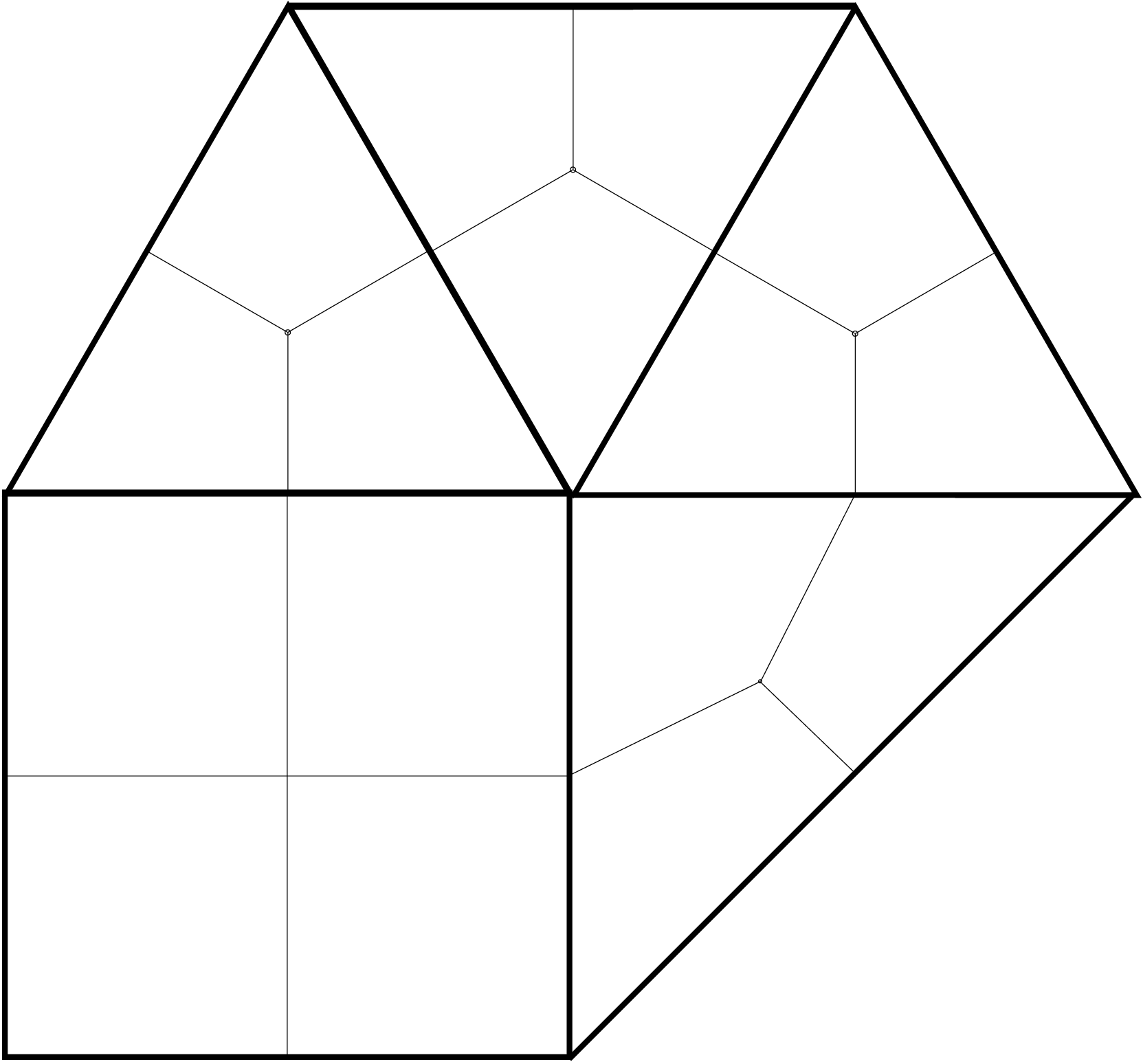}
\caption{}  \label{fig_sec3}
 \end{minipage}
\end{center}
\end{figure}

 Finally, observe  that we have
\begin{equation} \label{eq_disj_Ver0}
\face_0(q\cK) = \face_0(K_1) \sqcup \face_0(K_1|K_2) \sqcup \face_0(K_2)
\end{equation}
with
\begin{equation} \label{eq_def_K1|K2}
\face_0(K_1|K_2):= \{ \bar{e} \mid e \in \face_1(K_1)\} = \{ \bar{e} \mid e \in \face_1(K_2)\}
\end{equation}
where $\sqcup$ denotes disjoint union and $\bar{e}$ the barycenter of the edge $e$.

\begin{convention} \rm \label{conv_identify}
Let $V$ be a fixed finite-dimensional real vector space.
\begin{enumerate}
\item We set
$C_1(K):= C_1(K_1) \oplus C_1(K_2)$ \ and \ $C^1(K,V):= C^1(K_1,V) \oplus C^1(K_2,V) $.

\item Let $\psi: C_1(K) \to C_1(q\cK)$ be the (injective) linear map
given by
$$\psi(e)=e_1 + e_2 \quad \text{ for all } \quad e \in \face_1(K_1) \cup  \face_1(K_2)$$
where $e_1 = e_1(e), e_2 = e_2(e) \in  \face_1(q\cK)$ are the two edges of $q\cK$
``contained'' in $e$.
Using the identifications $C^1(q\cK,V) \cong C_1(q\cK) \otimes_{\bR} V$ and $C^1(K,V) \cong C_1(K) \otimes_{\bR} V$ we  naturally obtain the linear map
$$\psi^V := \psi \otimes \id_V: C^1(K,V) \to C^1(q\cK,V)$$

In the following we will identify $C_1(K)$ with the subspace $\psi(C_1(K))$ of $C_1(q\cK)$
and $C^1(K,V)$ with the subspace  $\psi^V(C^1(K,V))$ of $C^1(q\cK,V)$.
\end{enumerate}

\end{convention}

\subsubsection{The cell complexes $\cK \times \bZ_N$, $\cK' \times \bZ_N$, and $q\cK \times \bZ_N$}
\label{subsec4.0.4d}

 By $\cK \times \bZ_N$, $\cK' \times \bZ_N$, and $q\cK \times \bZ_N$
 we will denote the obvious product (polyhedral  cell) complexes.
 We omit the formal definitions. \par

Let $l = (l^{(k)})_{k \le n}$  be a simplicial loop in $q\cK \times \bZ_N$.
   By   $l_{\Sigma} = (l^{(k)}_{\Sigma})_{k \le n}$ and $l_{S^1} = (l^{(k)}_{S^1})_{k \le n}$
   we will denote\footnote{recall that $\Sigma$ and $S^1$ are the topological spaces
   that underly $q\cK$ and $\bZ_N$.
   We will later consider $l_{\Sigma}$ as a continuous curve in $\Sigma$.
   This is why we use the notation $l_{\Sigma}$
   rather than $l_{q\cK}$}
    the ``projected'' simplicial loops in $q\cK$  and $\bZ_N$.
(The definition of these loops  is obvious in the  ``vertex
point of view'' of a simplicial curve, cf. Sec.  \ref{subsec4.0.3}).\par
By $l^{red}_{\Sigma}$ we will denote the ``reduced''  simplicial loop in $q\cK$
which is obtained from $l_{\Sigma}$ by removing all the empty edges\footnote{this notation will be useful
for formulating condition (NCP) in Sec. \ref{sec6} below}.\par

It will be useful to introduce also for a simplicial ribbon  $R = (F_i)_{i \le n}$ in $q\cK \times \bZ_N$
the notion of a  ``reduced projected'' simplicial ribbon\footnote{we will use this notation  in Sec. \ref{subsec4.9} and in condition (NCP)' in Sec. \ref{sec6} below} $R_{\Sigma}$ in $q\cK$.
In order to define this notion
observe first that
for each $F_i$ appearing in $R$
there are two possibilities:\footnote{here we consider each $F_i$ as a subset of $M = \Sigma \times S^1$
and $\pi_{\Sigma}: M = \Sigma \times S^1 \to \Sigma$
is the canonical projection}
either $F_i$ is ``parallel'' or ``vertical'' w.r.t  $\Sigma$.
More precisely: either  $F^i_{\Sigma} := \pi_{\Sigma}(F_i)$ will be a 2-face of $q\cK$
or a 1-face of $q\cK$.
Clearly, the subsequence
of  the sequence $(F^i_{\Sigma})_{i \le n}$ of elements of $\face_2(q\cK) \cup \face_1(q\cK)$
which we obtain by omitting those  $F^i_{\Sigma}$ that are 1-faces
will be a simplicial ribbon in $q\cK$ which we will denote by $R_{\Sigma}$.

\subsection{Discrete Hodge star operators}
\label{subsec4.0.5}

Let $\cK$ and $\cK'$ be as in Sec. \ref{subsec4.0.4} above.
Moreover, let us assume that $\cK$ is a smooth polyhedral cell complex.
 For each $p \in \{0,1,2\}$ we define
  the operator  $\star_{\cK}:C_p(\cK) \to  C_{2-p}(\cK')$
 as the unique linear isomorphism such that for every
 $\alpha \in \face_p(\cK) \subset C_p(\cK)$ (cf. Convention \ref{conv_alpha=deltaalpha} above)
we have\footnote{recall that according to Convention \ref{conv_alpha=deltaalpha} above
by $\alpha$ and $\Check{\alpha}$
 we actually mean $\delta_{\alpha}$ and $\delta_{\Check{\alpha}}$}
\begin{equation} \star_{\cK} \alpha =
\begin{cases}
\ \ \Check{\alpha} & \text{ if  $or(\Check{\alpha})$ is the orientation induced by
$or(\alpha)$ and $\nu_{\Sigma}$} \\
- \Check{\alpha} & \text{otherwise} \\
\end{cases}
\end{equation}
where $\Check{\alpha} \in \face_{2-p}(\cK')$ is the face dual to $\alpha$,
$or(\alpha)$ and $or(\Check{\alpha})$ are the orientations of $\alpha$ and $\Check{\alpha}$,
and $\nu_{\Sigma}$ is the orientation on $\Sigma$.
 The operator $\star_{\cK'}:C_p(\cK') \to C_{2-p}(\cK)$ will be defined in
  a completely analogous way.

 \begin{convention} \rm \label{conv_star}
 In the following we will often simply write
 $\star$ instead of $\star_{\cK}$ or  $\star_{\cK'}$ .
   \end{convention}

 \begin{figure}[h]
\begin{center}
\begin{minipage}[b]{7 cm}
\begin{example} \label{ex_dualfaces} In the situation in Fig. \ref{dual_faces}
$\cC$ (resp. $\cC'$) is a ``hexagonal'' (resp. ``triangular'')
 (polyhedral) cell decomposition.
  The  2-face $F$ is dual to the 0-face $x'$,
the 0-face $x$ is dual to the 2-face $F'$,
and the two 1-faces $e$ and $e'$ are dual to each other.
Moreover, if $\cC$ is smooth and if  the  orientations of the  cells of $\cC'$
 are the ones induced by $\nu_{\Sigma}$ and the cell orientations of $\cC$
 we  have
$$\star x = F', \quad \star x' = F, \quad \star  e = e',$$
$$ \star e' = - e, \quad  \star  F = x', \quad  \star  F' = x $$
  \end{example}
\end{minipage}
\begin{minipage}[b]{1.0 cm}
\ \
\end{minipage}
\begin{minipage}[t]{7 cm}
 \includegraphics[height=7cm,width=7cm]{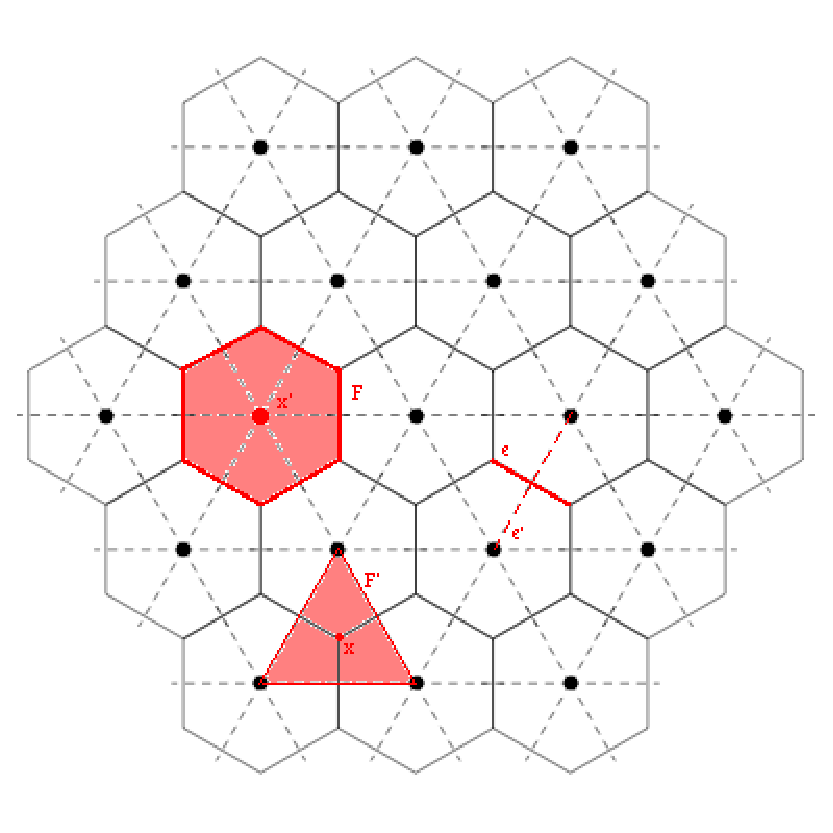}
\caption{}
\label{dual_faces}
 \end{minipage}
\end{center}
\end{figure}

Let  $V$ be a  real finite-dimensional vector space.
For each  $p \in \{0,1,2\}$ we define the
linear operators
$$\star^V_{\cK}:C^p(\cK,V) \to  C^{2-p}(\cK',V) \quad \text{ and } \quad
 \star^V_{\cK'}:C^p(\cK',V) \to  C^{2-p}(\cK,V) \quad \text{ by }$$
 $$ \star^V_{\cK} = \star_{\cK} \otimes \id_{V} \quad \text{ and } \quad \star^V_{\cK'} = \star_{\cK'} \otimes \id_{V}$$
where  we have used the obvious identifications $C^q(\cK,V) \cong C_q(\cK) \otimes V$ and
 $C^{q}(\cK',V) \cong C_{q}(\cK') \otimes V$ for $q \in \{0,1,2\}$.

 \begin{convention} \rm \label{conv_star2}
 In the following we will  simply write
 $\star_{\cK}$ and $\star_{\cK'}$ instead of $\star^V_{\cK}$ and $\star^V_{\cK'}$.
   \end{convention}

   From the definitions it follows that
\begin{equation} \label{eq4.2} \star_{\cK}^{-1} = (-1)^{p(2-p)} \star_{\cK'} =
\begin{cases} - \star_{\cK'} & \text{ if $p=1$  }\\
\ \ \star_{\cK'} & \text{ if $p \in \{0,2\}$  }
\end{cases}
\end{equation}
Moreover, when $V$ is equipped with a scalar product and
the spaces $C^p(\cK,V)$, $C^{2-p}(\cK',V)$, $C^p(\cK',V)$, and $C^{2-p}(\cK,V)$
with the induced scalar products then
$\star_{\cK}$ and  $\star_{\cK'}$ will be isometries.

\section{A simplicial realization of $\WLO(L)$}
\label{sec4}

The approach  for making rigorous sense of Eq. \eqref{eq2.48}
which we will introduce in the present section
was partly inspired by  Adams' approach in \cite{Ad0,Ad1} for discretizing   Abelian Chern-Simons
theory or, more precisely, Abelian $BF_3$-theory.
In fact, a crucial step in \cite{Ad0,Ad1} was
the transition to the ``BF-theory point of view'',  which involves, among other things,
  ``group doubling'' (or, equivalently, ``field doubling''),
cf. Sec. \ref{subsec3.1} above,   Sec. 7 in \cite{Ha7b}, and Appendix C in \cite{Ha7b}.\par

Adams's results  seem to suggest that for the
discretization of   non-Abelian CS theory  a similar strategy
will have to be used. It turns out, however, that for non-Abelian CS theory
the advantages of the ``$BF_3$-theory point of view'' are not as obvious
as in the Abelian case (cf. Remark 7.2 in Sec. 7 of \cite{Ha7b}) and that
 for the derivation of our main result, i.e. Theorem \ref{main_theorem} below, it will be sufficient
 to work with the original ``CS theory point of view''.
Accordingly, we will  postpone the transition to  $BF_3$-theory  to
Sec. 7 in \cite{Ha7b}.

\subsection{Definition of the spaces $\cB(q\cK)$, $\cA_{\Sigma}(q\cK)$,  $\cA_{\Sigma}(K)$, $\cA^{\orth}(q\cK)$, and $\cA^{\orth}(K)$}
\label{subsec4.1}

Let us first introduce suitable finite-dimensional analogues of the
spaces $\cB$,  $\cA_{\Sigma}$ and $\cA^{\orth}$.

 \smallskip

 For the transition from the continuum
to the discrete we now make the following replacements:

\medskip

\begin{tabular}{l c l}
$S^1$ & $\longrightarrow$ &  $\bZ_N$ \\
$\Sigma$ & $\longrightarrow$ &   $q\cK$\\
  $\cB = C^{\infty}(\Sigma,\ct) = \Omega^0(\Sigma,\ct)$   & $\longrightarrow$ & $\cB(q\cK):= C^0(q\cK,\ct)$\\
 $\cA_{\Sigma} = \Omega^1(\Sigma,\cG)$ &  $\longrightarrow$ &   $\cA_{\Sigma}(q\cK):= C^1(q\cK,\cG)$\\
$\cA^{\orth} \cong C^{\infty}(S^1, \cA_{\Sigma})$ & $\longrightarrow$ &
 $\cA^{\orth}(q\cK) := \Map(\bZ_N,\cA_{\Sigma}(q\cK))$\\
\end{tabular}

\medskip

\noindent
  Clearly, the scalar product $\langle \cdot, \cdot \rangle_{\cG}$ on $\cG$  induces
 scalar products $\ll \cdot, \cdot\gg_{\cB(q\cK)}$ and $\ll \cdot, \cdot\gg_{\cA_{\Sigma}(q\cK)}$
 on $\cB(q\cK)$ and $\cA_{\Sigma}(q\cK)$  in the standard way.
 We introduce a  scalar product
  $\ll \cdot , \cdot \gg_{\cA^{\orth}(q\cK)}$      on $\cA^{\orth}(q\cK) = \Map(\bZ_N,\cA_{\Sigma}(q\cK))$ by
   \begin{equation} \label{eq_norm_scalarprod}
\ll A^{\orth}_1 , A^{\orth}_2 \gg_{\cA^{\orth}(q\cK)}  =   \tfrac{1}{N} \sum_{t \in \bZ_N} \ll
A^{\orth}_1(t) , A^{\orth}_2(t) \gg_{\cA_{\Sigma}(q\cK)}
\end{equation} for all $A^{\orth}_1 , A^{\orth}_2 \in \cA^{\orth}(q\cK)$.

\begin{convention} \rm \label{conv4}
We identify $\cA_{\Sigma}(q\cK)$ with the subspace
$\{ A^{\orth} \in \Map(\bZ_N,\cA_{\Sigma}(q\cK)) \mid A^{\orth}  \text{ is constant}\}$  of $\cA^{\orth}(q\cK)$
 in the obvious way.
\end{convention}

For technical reasons\footnote{the use of these subspaces will allow us in  Sec. \ref{subsec4.2} below
to obtain a  ``good'' simplicial realization of the Hodge star operator $\star$ appearing in Eq. \eqref{eq2.42} above} we will not only work with the full spaces $\cA_{\Sigma}(q\cK)$ and $\cA^{\orth}(q\cK)$
but also with their subspaces  (cf. Convention \ref{conv_identify} above)
 $\cA_{\Sigma}(K)$ and $\cA^{\orth}(K)$ given by
\begin{align}
\cA_{\Sigma}(K) & :=C^1(K_1,\cG)  \oplus  C^1(K_2,\cG) \subset \cA_{\Sigma}(q\cK)\\
\cA^{\orth}(K) & := \Map(\bZ_N,\cA_{\Sigma}(K)) \subset \cA^{\orth}(q\cK)
\end{align}

\subsection{Discrete analogue of the operator $\tfrac{\partial}{\partial t} + \ad(B):\cA^{\orth} \to \cA^{\orth}$}
\label{subsec4.1b}

Let us now consider the issue of discretizing the operator\footnote{ Recall  that
$\tfrac{\partial}{\partial t}$ is the vector  field
on $S^1$ induced by the map $i_{S^1}: \bR \ni s \mapsto e^{2\pi i s} \in S^1$}
 $\tfrac{\partial}{\partial t} + \ad(B)$,
appearing in Eq. \eqref{eq2.42} above.  We will sometimes write $\partial_t$ instead of $\tfrac{\partial}{\partial t}$
 in the following.

\subsubsection{The operators $\hat{L}^{(N)}(b)$,  $\Check{L}^{(N)}(b)$, and $\bar{L}^{(N)}(b)$}

As a preparation let us consider first, for fixed $b \in \ct$,
 the continuum operator
 $$L(b):= \partial_t + \ad(b): C^{\infty}(S^1,\cG) \to  C^{\infty}(S^1,\cG)$$
 We want to find  a natural discrete analogue $L^{(N)}(b): \Map(\bZ_N,\cG) \to \Map(\bZ_N,\cG)$.
 Three natural choices for $L^{(N)}(b)$ are
$$\hat{\partial}_t^{(N)} + \ad(b), \quad \quad \Check{\partial}_t^{(N)} + \ad(b), \quad \quad \bar{\partial}_t^{(N)} + \ad(b)$$
with
\begin{equation} \label{eq_naive_dif.op}
\hat{\partial}_t^{(N)}  : =  N( \tau_1  -\tau_0), \quad \quad
\Check{\partial}_t^{(N)}  := N( \tau_0 - \tau_{-1}), \quad \quad
\bar{\partial}_t^{(N)}  := \tfrac{N}{2} ( \tau_1 - \tau_{-1})
\end{equation}
Here  $1$ and $-1$ are the obvious elements of $\bZ_N = \bZ/ N \bZ$
and $\tau_x:\Map(\bZ_N,\cG) \to  \Map(\bZ_N,\cG)$, for  $x \in \bZ_N$,
is  the translation operator
given by $(\tau_x f)(t) = f(t +x)$ for all $f \in \Map(\bZ_N,\cG)$ and $t \in \bZ_N$.\par

In fact, there are other  very natural discrete analogues of
$L(b)= \partial_t + \ad(b)$,
namely
\begin{subequations}  \label{eq_def_LOp}
\begin{align}
\hat{L}^{(N)}(b) & := N( \tau_1 e^{\ad(b)/N} -1)\\
\Check{L}^{(N)}(b) & := N(1 - \tau_{-1} e^{-\ad(b)/N}) \\
\label{eq_def_LOp_c}
\bar{L}^{(N)}(b) & := \tfrac{N}{2}( \tau_1 e^{\ad(b)/N} - \tau_{-1} e^{-\ad(b)/N}) \quad \text{ if $N$ is even }
\end{align}
\end{subequations}
Here we have written $1$ instead of  $\tau_0$ and $e^{s \ad(b)}$ instead of $\exp(s \ad(b))$ for $s \in \bR$.

\smallskip

The three operators \eqref{eq_def_LOp}  have a crucial advantage
over the operators \eqref{eq_naive_dif.op}: we do not
have to perform a ``continuum limit in the $S^1$-direction''\footnote{cf. Sec. 5.11 in the original
 (= June 2012) version of the present paper, see arXiv:1206.0439v1} in order to obtain the correct values
for $\WLO_{rig}(L)$ defined below.\par

We will now demonstrate that the three ``twisted'' difference operators \eqref{eq_def_LOp}
are indeed very natural\footnote{in fact,
twisted difference operators   play a crucial role in the
definition of the concept of the  Reidemeister torsion.
Since we are looking for a simplicial version of the approach in \cite{Ha4,Ha6}
which involved the Ray-Singer torsion  it should not be a surprise that
the operators \eqref{eq_def_LOp} will be useful for us}.
 We will restrict our attention to  the first of the three operators above, i.e. $\hat{L}^{(N)}(b)$.
Similar considerations can be made for the other two operators
$\Check{L}^{(N)}(b)$ and $\bar{L}^{(N)}(b)$.

\smallskip

Recall that  $i_{S^1}: \bR \ni s \mapsto e^{2\pi i s} \in S^1$.
We will often simply write $i(s)$ instead of $i_{S^1}(s)$, $s \in \bR$.
Recall also that
 we have identified $\bZ_N$ with the subgroup
 $ \{ e^{\frac{2 \pi i}{N} k} \mid 1 \le k \le N\}$ of
 the Lie group\footnote{We will  write the group law of $S^1$ additively} $S^1$.
Note that under this identification
$1 \in \bZ_N$ is identified with $ e^{2 \pi i \frac{1}{N}} = i_{S^1}(1/N) \in S^1$.

\subsubsection*{$ \hat{L}^{(N)}(b)$ is a natural discrete analogue of $L(b)$: first demonstration}

Let $(T_s)_{s \in \bR}$ be the    1-parameter group of  orthogonal operators on the real Hilbert space $L^2_{\cG}(S^1,dt)$ which is  generated by (the closure of) $L(b)$.
We have the following explicit formulas:
\begin{equation} \label{eq_cont_semigr} T_s = \tau_{i(s)} e^{s \ad(b)}, \quad s \in \bR
\end{equation}
\begin{equation} \label{eq_cont_gen} L(b) = \lim_{s \to 0} \tfrac{T_s - T_0}{s}
\quad \text{ on $C^{\infty}(S^1,\cG)$ }
 \end{equation}
where $\tau_t$ is the translation operator  $L^2_{\cG}(S^1,dt) \to L^2_{\cG}(S^1,dt)$
given by $(\tau_t f)(t') = f(t + t')$.

\medskip

As a  discrete analogue  of $(T_s)_{s \in \bR}$
we now take the family $(T^{(N)}_s)_{s \in \frac{1}{N} \bZ}$ given by
$$T^{(N)}_s = \tau_{i(s)} e^{s \ad(b)}, \quad s \in \tfrac{1}{N} \bZ$$
and a natural discrete analogue  of  the RHS  of Eq. \eqref{eq_cont_gen} is then
$$\tfrac{T^{(N)}_{1/N} - T^{(N)}_{0}}{1/N}=
 N(\tau_{i(\frac{1}{N})} e^{\frac{1}{N} \ad(b)} - 1) = \hat{L}^{(N)}(b)$$

\subsubsection*{$ \hat{L}^{(N)}(b)$ is a natural discrete analogue of $L(b)$: second demonstration}

Observe that $\hat{L}^{(N)}(b)$ coincides with $\hat{\partial}^{(N)}_t$
on $\Map(\bZ_N,\ct)$, which is a very natural operator.
For our purposes it is therefore enough to demonstrate that the operator
$$S^{(N)}:= \hat{L}^{(N)}(b)_{|\Map(\bZ_N,\ck)}$$
 is a natural discretization of the  continuum operator
  $$S:= L(b)_{|C^{\infty}(S^1,\ck)}$$
In the special case where $b \in \ct_{reg}$ (which is the only case relevant for us)
$S$ is invertible and it is easy to verify that
$S^{-1}: C^{\infty}(S^1,\ck) \to C^{\infty}(S^1,\ck)$ is given explicitly by
\begin{equation}
(S^{-1}f)(t)  = \bigl((e^{\ad(b)})_{| \ck} - 1_{\ck} \bigr)^{-1} \cdot \int_0^1 e^{s \ad(b)} f(t+ i(s)) ds,
 \quad \quad t \in S^1
\end{equation}
for all $f \in C^{\infty}(S^1,\ck)$.
This suggests the following discrete analogue $S^{-1}_{(N)} : \Map(\bZ_N,\ck) \to \Map(\bZ_N,\ck)$
of $S^{-1}$:
\begin{equation}
(S^{-1}_{(N)} f)(t)  = \bigl((e^{\ad(b)})_{| \ck} - 1_{\ck}  \bigr)^{-1} \cdot \tfrac{1}{N} \sum_{k=0}^{N-1}
\bigl[ e^{s \ad(b)} f(t+ i(s))\bigr]_{| s=k/N}, \quad \quad t \in \bZ_N
\end{equation}
for all $f \in \Map(\bZ_N,\ck)$.
 Clearly, if $S^{-1}_{(N)}$ is invertible then
$(S^{-1}_{(N)})^{-1}$ can be considered as a discrete analogue of $S$.
Now a short computation shows that
$S^{-1}_{(N)} \cdot S^{(N)}  = \id_{\Map(\bZ_N,\ck)}$
so  $S^{-1}_{(N)}$ is indeed invertible and we
have  $  (S^{-1}_{(N)})^{-1}  = S^{(N)}  = \hat{L}^{(N)}(b)_{|\Map(\bZ_N,\ck)}$.

\subsubsection{The operator $L^{(N)}(B)$}

For each $B \in \cB(q\cK)$ we will denote by $L^{(N)}(B)$
 the operator $\cA^{\orth}(K) \to \cA^{\orth}(K)$ which, under the identification
 \begin{equation} \label{eq_cA^{orth}(K)ident}
\cA^{\orth}(K) \cong  \Map(\bZ_N,C^1(K_1,\cG)) \oplus   \Map(\bZ_N,C^1(K_2,\cG))
\end{equation}
 is   given by (cf. Remark \ref{rm_for_appE} below)
 \begin{equation} \label{def_LN} L^{(N)}(B) = \left( \begin{matrix}
 \hat{L}^{(N)}(B) && 0 \\
0 && \Check{L}^{(N)}(B)
\end{matrix} \right)
\end{equation}
Here $\hat{L}^{(N)}(B):  \Map(\bZ_N, C^1(K_1,\cG))  \to \Map(\bZ_N, C^1(K_1,\cG)) $
and $\Check{L}^{(N)}(B):  \Map(\bZ_N, C^1(K_2,\cG))  \to \Map(\bZ_N, C^1(K_2,\cG)) $
are given by
\begin{align*}
 (\hat{L}^{(N)}(B)  A^{\orth}_1)(t)(e) & =  \hat{L}^{(N)}(B(\bar{e})) \cdot A^{\orth}(t)(e)  \quad
 \forall e \in \face_1(K_1), t \in  \bZ_N,  A^{\orth}_1 \in \Map(\bZ_N, C^1(K_1,\cG)\\
 (\Check{L}^{(N)}(B)  A^{\orth}_2)(t)(e) & =  \Check{L}^{(N)}(B(\bar{e})) \cdot A^{\orth}(t)(e)  \quad
 \forall e \in \face_1(K_2), t \in  \bZ_N,  A^{\orth}_2 \in \Map(\bZ_N, C^1(K_2,\cG)
 \end{align*}
where  $A^{\orth}_j \in \Map(\bZ_N,C^1(K_j,\cG))$, $j \in \{1,2\}$,
are the  components of
$A^{\orth} \in \cA^{\orth}(K)$ w.r.t the decomposition \eqref{eq_cA^{orth}(K)ident} above.

\begin{remark} \rm
Alternatively, $\hat{L}^{(N)}(B)$ and $\Check{L}^{(N)}(B)$ can be characterized
by \begin{subequations}
\begin{align} \label{eq_LN_ident1}
\hat{L}^{(N)}(B)) & \cong \oplus_{\bar{e} \in  \face_0(K_1 | K_2)}
 \hat{L}^{(N)}(B(\bar{e})) \\
 \label{eq_LN_ident2} \Check{L}^{(N)}(B)) & \cong \oplus_{\bar{e} \in  \face_0(K_1 | K_2)}
 \Check{L}^{(N)}(B(\bar{e}))
\end{align}
\end{subequations}
where $\face_0(K_1 | K_2)$ is as in Eq. \eqref{eq_def_K1|K2} above.
In Eqs.  \eqref{eq_LN_ident1} and \eqref{eq_LN_ident2} we used the obvious identification
$$ \Map(\bZ_N, C^1(K_j,\cG)) \cong \oplus_{e \in  \face_1(K_j)}
\Map(\bZ_N,\cG) \cong \oplus_{\bar{e} \in  \face_0(K_1 | K_2)}
\Map(\bZ_N,\cG)$$
\end{remark}

\begin{remark} \label{rm_for_appE} \rm
  Regarding the definition of the operator    $L^{(N)}(B)$ above
one might wonder why instead of the operators $\hat{L}^{(N)}(b)$ and $\Check{L}^{(N)}(b)$
we did not use the  more ``symmetric'' operators  $\bar{L}^{(N)}(b)$ appearing
in Eq. \eqref{eq_def_LOp_c} above.
 In fact, we will see in  Appendix D in \cite{Ha7b} that -- after making the aforementioned
 transition to  $BF_3$-theory (cf. the beginning of Sec. \ref{sec4} above) --
  one actually can work with the   operators $\bar{L}^{(N)}(b)$ in a natural way.
 \end{remark}

\subsection{Definition of $S^{disc}_{CS}(A^{\orth},B)$}
\label{subsec4.2}

 Let us now  introduce a
simplicial analogue for
\begin{equation} \label{eq4.7}
S_{CS}(A^{\orth},B) =  \pi k \bigl[ \ll A^{\orth},
 \star  \bigl(\tfrac{\partial}{\partial t} + \ad(B) \bigr) A^{\orth} \gg_{\cA^{\orth}}
+ 2 \ll \star A^{\orth},  dB \gg_{\cA^{\orth}} \bigr],
\end{equation}
cf.  Eq. \eqref{eq2.42} above. In order to do so  we have to look
for suitable  simplicial analogues of the mappings $ \star:
\cA_{\Sigma} \to \cA_{\Sigma}$,
 $\star: C^{\infty}(S^1, \cA_{\Sigma}) \to C^{\infty}(S^1, \cA_{\Sigma})$,
  $ d: C^{\infty}(\Sigma,\cG) \to  \cA_{\Sigma}$,
and the scalar product  $\ll \cdot , \cdot \gg_{\cA^{\orth}}$ on
  $\cA^{\orth} \cong C^{\infty}(S^1, \cA_{\Sigma})$
appearing in Eq. \eqref{eq4.7}.\par

For the transition from the continuum setting to the simplicial setting
let us now make the following replacements:

\medskip

\begin{tabular}{l c l}
$\ll \cdot , \cdot \gg_{\cA^{\orth}}$ & $\longrightarrow$ & $\ll \cdot , \cdot \gg_{\cA^{\orth}(q\cK)}$\\

\ &  \ \\

$d$ & $\longrightarrow$ & $d_{q\cK}$ \\

\ &  \ \\

$\star$ & $\longrightarrow$ &  $\biggl( \begin{matrix} 0 && \star_{K_2} \\
\star_{K_1} && 0 \end{matrix} \biggl) =: \star_K$ \\

\end{tabular}

\medskip

Here  $d_{q\cK}: C^0(q\cK,\ct) \to C^1(q\cK,\ct)$ is as  in Sec. \ref{sec4.0}.
The matrix operator notation for  the operator $\star_K$ refers  to the
identification \eqref{eq_cA^{orth}(K)ident}
and the  operators     $\star_{K_1}:\Map(\bZ_N,C^1(K_1,\cG)) \to \Map(\bZ_N,C^1(K_2,\cG)) $, \
 $\star_{K_2}  :\Map(\bZ_N,C^1(K_2,\cG)) \to
 \Map(\bZ_N,C^1(K_1,\cG))$
 appearing on the secondary diagonal  are the linear isomorphisms
defined in the obvious way\footnote{i.e.
by $(\star_{K_j} A^{\orth}_j)(t) = \star_{K_j} (A^{\orth}_j(t))$
for each $ A^{\orth}_j \in \Map(\bZ_N,C^1(K_j,\cG))$, $t \in \bZ_N$, and $j=1,2$
where  $\star_{K_j}: C^1(K_j,\cG) \to C^{1}(K_{3-j},\cG)$  is given as in  Sec. \ref{sec4.0} above}.\par

Using the replacements listed above and the operator $L^{(N)}(B)$ introduced in Sec. \ref{subsec4.1b} above
we now arrive at the following simplicial analogue for
$S_{CS}(A^{\orth},B)$:
\begin{equation} \label{eq4.10}
S^{disc}_{CS}(A^{\orth},B) :=   \pi  k \biggl[ \ll A^{\orth},
\star_K  L^{(N)}(B)
   A^{\orth} \gg_{\cA^{\orth}(q\cK)}
 + 2 \ll  \star_K A^{\orth},  d_{q\cK}  B \gg_{\cA^{\orth}(q\cK)}  \biggr]
\end{equation}
 for all $A^{\orth} \in \cA^{\orth}(K) \subset \cA^{\orth}(q\cK) $ and $B \in \cB(q\cK)$.
   (Observe that $d_{q\cK}  B \in \cA_{\Sigma}(q\cK) \subset \cA^{\orth}(q\cK)$  according to Convention \ref{conv4} in Sec. \ref{subsec4.1} above.)

\begin{proposition}
The  operator  $\star_K L^{(N)}(B): \cA^{\orth}(K) \to \cA^{\orth}(K)$
is symmetric  w.r.t to the scalar product $\ll  \cdot, \cdot  \gg_{\cA^{\orth}(q\cK)}$.
\end{proposition}

\begin{proof} From the definition of $\star_K$, Eq. \eqref{eq4.2}
and the last paragraph of Sec. \ref{subsec4.0.5}
it follows that $\star_K$  is anti-symmetric  w.r.t to the scalar product
$\ll  \cdot, \cdot  \gg_{\cA^{\orth}(q\cK)}$.
On the other hand a short computation shows that the adjoint of $L^{(N)}(B)$
  w.r.t to the scalar product $\ll  \cdot, \cdot  \gg_{\cA^{\orth}(q\cK)}$
 is the operator
$$  \left( \begin{matrix}
 - \Check{L}^{(N)}(B) && 0 \\
0 && - \hat{L}^{(N)}(B)
\end{matrix} \right)$$
 From these two observations the assertion easily follows.
\end{proof}

\subsection{Definition of $\Hol^{disc}_{l}(A^{\orth},  B)$ and $\Hol^{disc}_{R}(A^{\orth},  B)$}
\label{subsec4.3}

We will now introduce two ``simplicial versions''
for the expression $ \Hol_{l}(A^{\orth},   B) = \Hol_{l}(A^{\orth}  + B dt)$
appearing in Eq. \eqref{eq2.49c} above.
The first version is $\Hol^{disc}_{l}(A^{\orth},  B)$ where $l$ is a simplicial loop
 and the second version is   $\Hol^{disc}_{R}(A^{\orth},  B)$
where $R$ is a closed simplicial ribbon.

\subsubsection{Simplicial loop case}
\label{subsubsec4.3.1}

We start with the
observation that for $A^{\orth} \in \cA^{\orth} \cong
C^{\infty}(S^1,\cA_{\Sigma})$ and $B \in \cB =
C^{\infty}(\Sigma,\ct)$ we have
\begin{equation} \label{eq4.16} A^{\orth}(l'(t))  = A^{\orth}(l_{S^1}(t))(l'_{\Sigma}(t)), \quad \quad
(Bdt) (l'(t))  = B(l_{\Sigma}(t)) \cdot dt(l'_{S^1}(t))
\end{equation}
 for $t \in [0,1]$.
 Here we have set $l_{\Sigma}:=
\pi_{\Sigma} \circ l$ and $l_{S^1}:= \pi_{S^1} \circ l$
where $\pi_{\Sigma} : \Sigma \times S^1 \to \Sigma$,  $\pi_{S^1} : \Sigma \times S^1 \to S^1$
are the canonical projections.
From  Eqs. \eqref{eq2.5}, \eqref{eq2.49c}, and \eqref{eq4.16} we obtain
\begin{equation} \label{eq4.17}
 \Hol_{l}(A^{\orth},   B)  = \lim_{n \to \infty} \prod_{k = 1}^n
 \exp\biggl(\bigl[\tfrac{1}{n} \bigl(
  A^{\orth}(l_{S^1}(t))(l'_{\Sigma}(t)) +
 B(l_{\Sigma}(t)) \cdot dt(l'_{S^1}(t)) \bigr) \bigr]_{| t = \tfrac{k}{n}}\biggr)
\end{equation}

Let us now discretize the RHS  of this equation. Let $l =
(l^{(k)})_{k \le n}$  be a simplicial loop in
 $q\cK \times \bZ_N$ and
let  $l_{\Sigma} = (l^{(k)}_{\Sigma})_{k \le n}$ and $l_{S^1} = (l^{(k)}_{S^1})_{k \le n}$
denote the projected simplicial loops in $q\cK$  and $\bZ_N$, cf. Sec. \ref{subsec4.0.4d}
 above.\par

In contrast to the situation in the continuum setting where the
parameter $n \in \bN$ was sent to $\infty$ we will leave $n$ fixed.
 It is now natural to make the replacements

\medskip

\begin{tabular}{l c l}

$l_{\Sigma}(\tfrac{k}{n})$   & $\longrightarrow$ & $\start l^{(k)}_{\Sigma}$\\
$l_{S^1}(\tfrac{k}{n})$ & $\longrightarrow$ & $\start l^{(k)}_{S^1}$ \medskip \\

$\tfrac{1}{n} l'_{\Sigma}(\tfrac{k}{n})$ & $\longrightarrow$ & $l^{(k)}_{\Sigma}$\\
$\tfrac{1}{n} l'_{S^1}(\tfrac{k}{n})$ & $\longrightarrow$ & $l^{(k)}_{S^1}$\\
$ dt$ & $\longrightarrow$ & $dt^{(N)}$
\end{tabular}

\medskip

\noindent with $dt^{(N)} \in C^1(\bZ_{N},\bR)$
  given by
 $dt^{(N)}(e)= \tfrac{1}{N}$ for all $e \in \face_1(\bZ_N)$.
 We will make the identification $C^1(\bZ_{N},\bR) \cong \Hom_{\bR}(C_1(\bZ_{N}),\bR)$
 and  $C^1(q\cK,\cG) \cong \Hom_{\bR}(C_1(q\cK),\cG)$. \par

 Applying the replacements above to the
RHS  of Eq. \eqref{eq4.17}  we   arrive at the ansatz
\begin{equation} \label{eq4.18}
\Hol^{disc}_{l}(A^{\orth},   B) :=    \prod_{k=1}^n \exp\biggl(
A^{\orth}(\start l^{(k)}_{S^1})(l^{(k)}_{\Sigma})  +
  B(\start l^{(k)}_{\Sigma}) \cdot dt^{(N)}(l^{(k)}_{S^1}) \biggr)
\end{equation}
for $A^{\orth} \in \cA^{\orth}(K) \subset \cA^{\orth}(q\cK)$ and $B \in \cB(q\cK)$.

\subsubsection{Simplicial ribbon case}
\label{subsubsec4.3.2}

 Instead of working with a  simplicial loops $l$
in $q\cK \times \bZ_N$ let us now work with closed simplicial
ribbons $R$ in $q\cK \times \bZ_N$. \par

According to Remark \ref{rm_appK_0}  above
 every (closed) simplicial ribbon $R  = (F_k)_{k \le n}$, $n \in \bN$,
  in $q\cK \times \bZ_N$ induces a pair  $(l,l')$ of
  simplicial loops $l = (l^{(k)})_{k \le n}$ and   $l' = (l^{'(k)})_{k \le n}$
      in $q\cK \times \bZ_N$ in the obvious way.
Let $l_{\Sigma}$, $l'_{\Sigma}$, $l_{S^1}$, $l'_{S^1}$  denote
the corresponding ``projected'' simplicial loops in $q\cK$ and $\bZ_N$,
 cf. Sec. \ref{subsec4.0.4d} above.

\smallskip

 We can then  introduce the following ribbon
analogue of Eq. \eqref{eq4.18} above
\begin{multline} \label{eq4.21}
\Hol^{disc}_{R}(A^{\orth},   B) :=
  \prod_{k=1}^n \exp\biggl( \tfrac{1}{2}\bigl(A^{\orth}(\start l^{(k)}_{S^1})\bigr)(l^{(k)}_{\Sigma}) + \tfrac{1}{2} \bigl(A^{\orth}(\start l^{'(k)}_{S^1})\bigr)(l^{'(k)}_{\Sigma}) \\
  +  \tfrac{1}{2} B(\start l^{(k)}_{\Sigma})
 \cdot dt^{(N)}(l^{(k)}_{S^1}) + \tfrac{1}{2} B(\start l^{'(k)}_{\Sigma})  \cdot dt^{(N)}(l^{'(k)}_{S^1}) \biggr)
\end{multline}

\subsection{Definition of $\Det^{disc}_{FP}(B)$}
\label{subsec4.4}

We will need a discrete analogue $\Det^{disc}_{FP}(B)$ of the factor $\Det_{FP}(B) =   \det\bigl(1_{\ck}-\exp(\ad(B))_{|\ck}\bigr)$ appearing in Eq. \eqref{eq2.49} above. We make the
ansatz
\begin{equation} \label{eq4.23}
\Det^{disc}_{FP}(B) :=
 \prod_{x \in \face_0(q\cK)}   \det\bigl(1_{\ck}-\exp(\ad(B(x)))_{| \ck}\bigr)
 \end{equation}
for every $B \in \cB(q\cK)$ where $1_{\ck}$ is the identity operator on $\ck$.

\subsection{Discrete version of $1_{C^{\infty}(\Sigma,\ct_{reg})}(B)$}
\label{subsec4.6}

Let us introduce a discrete analogue of the indicator function
$1_{C^{\infty}(\Sigma,\ct_{reg})}(B)$.
The obvious candidate is
$\prod_{x \in \face_0(q\cK)}
1_{\ct_{reg}}(B(x))$.
However, with this choice we would get some problems
later due to the fact that $1_{\ct_{reg}}: \ct \to \{0,1\} \subset [0,1]$
is non-continuous.
For this reason we will regularize the function
$1_{\ct_{reg}}$. We fix a family
$(1^{(s)}_{\ct_{reg}})_{s > 0}$,
of elements of $C^{\infty}_{\bR}(\ct)$,
 with the following properties:
\begin{itemize}
\item $\Image(1^{(s)}_{\ct_{reg}}) \subset [0,1]$ \quad and \quad
 $\supp(1^{(s)}_{\ct_{reg}}) \subset  \ct_{reg}$
\item $1^{(s)}_{\ct_{reg}} \to 1_{\ct_{reg}}$ pointwise as $s \to 0$
\item Each $1^{(s)}_{\ct_{reg}}$ is  invariant under the operation of the affine Weyl group $\cW_{\aff}$
 on $\ct$, cf. part \ref{appB} of the Appendix below.
\end{itemize}

\noindent
For fixed $s > 0$ we  take as the discrete analogue of $1_{C^{\infty}(\Sigma,\ct_{reg})}(B)$
 the expression
\begin{equation} \label{eq4.28} \prod_{x} 1^{(s)}_{\ct_{reg}}(B(x)):=
\prod_{x \in \face_0(q\cK)}
1^{(s)}_{\ct_{reg}}(B(x))
\end{equation}
Later we will let $s \to  0$.

\subsection{Discrete version of the  decomposition
$\cA^{\orth} =  \Check{\cA}^{\orth} \oplus \cA^{\orth}_c$}
\label{subsec4.7}

We introduce a decomposition of $\cA^{\orth}(K)$, which is
analogous to the decomposition $\cA^{\orth} =  \Check{\cA}^{\orth}
\oplus \cA^{\orth}_c$ in Sec.
 \ref{subsubsec2.3.3}. In order to do so we
introduce the notation
\begin{equation} \label{eq4.29}
  \cA_{\Sigma,V}(K)  := C^1(K_1,V) \oplus C^1(K_2,V)
\end{equation}
for every real vector space $V$. We then have the following analogue
of the decomposition $\cA^{\orth} =  \Check{\cA}^{\orth} \oplus
\cA^{\orth}_c$ in Sec. \ref{subsubsec2.3.3} above namely,
\begin{equation} \label{eq4.30} \cA^{\orth}(K) =  \Check{\cA}^{\orth}(K) \oplus \cA^{\orth}_c(K)
\end{equation}
where
\begin{subequations} \label{eq4.31}
\begin{align}
 \cA^{\orth}_c(K) & := \{ A^{\orth} \in  \cA^{\orth}(K) \mid A^{\orth}
 \text{ is constant and } \cA_{\Sigma,\ct}(K)\text{-valued}\} \cong \cA_{\Sigma,\ct}(K) \\
 \Check{\cA}^{\orth}(K) & :=
 \{ A^{\orth} \in  \cA^{\orth}(K) \mid
\sum\nolimits_{t \in \bZ_N}  A^{\orth}(t) \in  \cA_{\Sigma,\ck}(K) \}
 \end{align}
 \end{subequations}
Observe that $\cA_{\Sigma}(K) = \cA_{\Sigma,\cG}(K) \cong
 \cA_{\Sigma,\ct}(K)  \oplus \cA_{\Sigma,\ck}(K)$ and
 that for every $\Check{A}^{\orth} \in  \Check{\cA}^{\orth}(K)$,
  $A^{\orth}_c \in \cA^{\orth}_c(K)$, and $B \in \cB(q\cK)$ we have
\begin{equation} \label{eq4.32}
S^{disc}_{CS}(\Check{A}^{\orth} + A^{\orth}_c ,B) =
S^{disc}_{CS}(\Check{A}^{\orth},B) + S^{disc}_{CS}(A^{\orth}_c,B)
\end{equation}
with the two expressions on the RHS  given explicitly by  Eqs.
\eqref{eq4.33} and \eqref{eq4.35} below.

\smallskip

For referring to the elements of the space $\Check{\cA}^{\orth}(K)$
we will use the variable
$\Check{A}^{\orth}$
and for the elements of the space $\cA_{\Sigma,\ct}(K) \cong \cA^{\orth}_c(K)$
the variable $A^{\orth}_c$.

\subsection{Discrete versions of the two Gauss-type measures in Eq. \eqref{eq2.48}}
  \label{subsec4.8}

\begin{definition} \label{def3.1}
  An ``oscillatory
 Gauss-type measure'' on  a Euclidean vector space $(V, \langle \cdot, \cdot \rangle)$
 is a  complex Borel measure $d\mu$ on $V$
 of the form
 \begin{equation} \label{eq3.1}
 d\mu(x) = \tfrac{1}{Z} e^{ - \tfrac{i}{2} \langle x - m, S (x-m) \rangle} dx
\end{equation}
with $Z \in \bC \backslash \{0\}$,
   $m \in V$, and where  $S$ is a  symmetric endomorphism of $V$  and
 $dx$  the normalized\footnote{i.e. unit hyper-cubes have volume $1$ w.r.t. $dx$}
 Lebesgue measure on $V$.
 Note that $Z$, $m$ and $S$ are uniquely determined by $d\mu$
 so we can use the notation $Z_{\mu}$, $m_{\mu}$ and $S_{\mu}$
 in order to denote these objects.
\begin{enumerate}
\item We call $d\mu$ ``centered''iff $m_{\mu}=0$.

\item  We call $d\mu$ ``degenerate'' iff $S_{\mu}$ is not invertible
 \end{enumerate}
 \end{definition}

Clearly,  we have
\begin{align}  \label{eq4.33} S^{disc}_{CS}(\Check{A}^{\orth},B)
 & =  \pi k  \ll \Check{A}^{\orth}, \star_K  L^{(N)}(B) \cdot
 \Check{A}^{\orth} \gg_{\cA^{\orth}(q\cK)}
\end{align}
for all $\Check{A}^{\orth} \in \Check{\cA}^{\orth}(K) \subset \cA^{\orth}(q\cK)$ and  $B \in \cB(q\cK)$.
Moreover, if $B \in \cB(q\cK)$ fulfills $\prod_{x} 1_{\ct_{reg}}(B(x)) \neq 0$
then $L^{(N)}(B): \Check{\cA}^{\orth}(K) \to \Check{\cA}^{\orth}(K)$
will be injective, cf. Proposition 5.1 in \cite{Ha7b}.
In this case  the  (rigorous) complex measure
 \begin{equation} \label{eq4.34}
 \exp(iS^{disc}_{CS}(\Check{A}^{\orth},B))  D\Check{A}^{\orth}
 \end{equation}
is a non-degenerate  centered
oscillatory Gauss type measure on $\Check{\cA}^{\orth}(K)$.
Here we have equipped $\Check{\cA}^{\orth}(K)$
with the restriction $\ll \cdot,\cdot\gg_{\Check{\cA}^{\orth}(K)}$
of the scalar product $\ll \cdot,\cdot\gg_{{\cA}^{\orth}(q\cK)}$
 onto $\Check{\cA}^{\orth}(K)$ and $D\Check{A}^{\orth}$
 denotes the normalized Lebesgue measure on $\Check{\cA}^{\orth}(K)$
 w.r.t. $\ll \cdot,\cdot\gg_{\Check{\cA}^{\orth}(K)}$.

 \medskip

 Moreover, since
\begin{equation} \label{eq4.35} S^{disc}_{CS}(A^{\orth}_c ,B)
=   2 \pi k  \ll  \star_{K} A^{\orth}_c,  d_{q\cK} B \gg_{\cA^{\orth}(q\cK)}
\end{equation}
it follows that the complex measure
 \begin{equation} \label{eq4.36}
 \exp(i  S^{disc}_{CS}(A^{\orth}_c,B))    (DA^{\orth}_c \otimes  DB)
 \end{equation}
is a (degenerate) centered  oscillatory Gauss type measure
 on $\cA^{\orth}_c(K)  \oplus \cB(q\cK)$. Here we have equipped $\cA^{\orth}_c(K)  \oplus \cB(q\cK)$
with the scalar product  $ \ll \cdot, \cdot \gg_{\cA^{\orth}_c(K)  \oplus
\cB(q\cK)} :=  \ll \cdot, \cdot \gg_{\cA^{\orth}_c(K)} \oplus
 \ll \cdot, \cdot \gg_{\cB(q\cK)}$
 (where $\ll \cdot,\cdot\gg_{{\cA}^{\orth}_c(K)}$
denotes the restriction of the scalar product
$\ll \cdot,\cdot\gg_{{\cA}^{\orth}(q\cK)}$ onto the space $\cA^{\orth}_c(K)$)
 and $DA^{\orth}_c$ and $DB$
 denote the obvious normalized Lebesgue measures.

\smallskip

In the next subsection we will need the following definition

\begin{definition} \label{def3.2} Let  $d\mu$  be an oscillatory
 Gauss-type measure on a  Euclidean vector space $(V, \langle \cdot, \cdot \rangle)$.
 A (Borel) measurable function
  $f: V \to \bC$ will be called improperly integrable w.r.t. $d\mu$
  iff\footnote{Observe that
$\int_{\ker(S_{\mu})}  e^{- \eps \|x\|^2} dx =
(\tfrac{\eps}{\pi})^{-n/2}$. In particular, the factor
$(\tfrac{\eps}{\pi})^{n/2} $ in Eq. \eqref{eq3.2} above  ensures
that also for degenerate oscillatory
 Gauss-type measure the improper integrals  $\int\nolimits_{\sim} 1 \ d\mu$  exists}
 \begin{equation}\label{eq3.2} \int\nolimits_{\sim} f d\mu := \int\nolimits_{\sim} f(x)
   d\mu(x): =
    \lim_{\eps \to 0} (\tfrac{\eps}{\pi})^{n/2} \int f(x) e^{- \eps |x|^2} d\mu(x)
  \end{equation}
  exists. Here  we have set  $n:=\dim(\ker(S_{\mu}))$.
   Note that if $d\mu$ is non-degenerate we have $n=0$ so the factor $(\tfrac{\eps}{\pi})^{n/2}$
is then trivial.
 \end{definition}

\subsection{Definition of  $\WLO^{disc}_{rig}(L)$ and $\WLO_{rig}(L)$}
 \label{subsec4.9}

 For the rest of this
paper let us assume that\footnote{cf. Remark \ref{rm_appK_1} in Sec. \ref{subsec4.0.3b} above}
 $L= (R_1, R_2, \ldots, R_m)$, $m \in \bN$,
is a fixed simplicial ribbon  link   in $q\cK \times \bZ_{N}$ with ``colors''
$(\rho_1,\rho_2,\ldots,\rho_m)$.

\smallskip

Using the definitions of the previous
subsections  we then arrive at the following
 simplicial   analogue $\WLO^{disc}_{rig}(L)$
of the heuristic expression $\WLO(L)$ in Eq. \eqref{eq2.48}

\begin{multline} \label{eq_def_WLOdisc}
\WLO^{disc}_{rig}(L)  :=  \lim_{s \to 0}
  \sum_{y \in I}\int\nolimits_{\sim} \biggl\{ \bigl( \prod_{x} 1^{(s)}_{\ct_{reg}}(B(x))
  \bigr) \Det^{disc}_{FP}(B)\\
\times \biggl[
\int\nolimits_{\sim} \biggl( \prod_{i=1}^m  \Tr_{\rho_i}\bigl( \Hol^{disc}_{R_i}(\Check{A}^{\orth} +  A^{\orth}_c, B)\bigr) \biggr)  \exp(iS^{disc}_{CS}(\Check{A}^{\orth}, B))
D\Check{A}^{\orth} \biggr] \\
 \times        \exp\bigl(- 2 \pi i k  \langle y, B(\sigma_0) \rangle \bigr) \biggr\}
  \exp(i  S^{disc}_{CS}(A^{\orth}_c,B))    (DA^{\orth}_c \otimes  DB)
\end{multline}
where $\sigma_0$ is an arbitrary fixed point of $\face_0(q\cK)$
which does not lie in   $\bigcup_{i \le m} \Image(R^i_{\Sigma})$.
Here  $R^i_{\Sigma} := (R_i)_{\Sigma}$ is defined as in Sec. \ref{subsec4.0.4d} above\footnote{Recall that according to part i) of Remark \ref{rm_appK_0} above
  we can consider $R^i_{\Sigma}$ as a map $S^1 \times [0,1] \to \Sigma$
in a natural way}.

\smallskip

Finally, we set\footnote{at this stage we do not yet claim that $\WLO^{disc}_{rig}(L)$ and $\WLO_{rig}(L)$
are actually well-defined}
\begin{equation} \label{eq4.42}
\WLO_{rig}(L) :=    \frac{\WLO^{disc}_{rig}(L)}{\WLO^{disc}_{rig}(\emptyset)}
\end{equation}
where   $\emptyset$ is the ``empty'' link\footnote{so
$\WLO^{disc}_{rig}(\emptyset)$ is a simplicial analogue of the partition function
$Z(\Sigma \times S^1)$}.

 \begin{remark}  \label{rm3.2} \rm
   We could equally well state our main result below in terms of  $\WLO^{disc}_{rig}(L)$ rather than
  $\WLO_{rig}(L)$.  For stylistic reasons we prefer $\WLO_{rig}(L)$.
   \end{remark}

\subsection{Two modifications}
 \label{subsec4.10}

One might expect that -- at least for  simplicial ribbon links $L$ which are equivalent
to the  framed links of the simple type mentioned in Sec. \ref{subsubsec2.3.3} above --
the expression  $\WLO_{rig}(L)$ is well-defined and that we have
\begin{equation} \label{eq4.42b}
\WLO_{rig}(L) = \frac{|L|}{|\emptyset|}
\end{equation}
where  $|\cdot|$ is is the shadow invariant associated to $\cG$ and $k$, cf. \eqref{eq2.50} above and
 cf. Remark \ref{rm3.4} below.\par

It turns out, however, that in order to obtain this result
 we have to modify our original approach.
In order to do so we will now
make two modifications (Mod1) and (Mod2).
More precisely, we will redefine the notation
$\WLO^{disc}_{rig}(L)$ according to the modifications (Mod1) and (Mod2)
which we will now describe.
  $\WLO_{rig}(L)$ will again be given by Eq. \eqref{eq4.42} (with the redefined version of $\WLO^{disc}_{rig}(L)$
  appearing on the RHS).

\subsubsection*{Modification (Mod1)}

 Let us  reconsider the question
of what a suitable  discrete analogue
$\Det^{disc}_{FP}(B)$ of the continuum expression
 $\Det_{FP}(B) = \det\bigl(1_{{\ck}}-\exp(\ad(B))_{| {\ck}}\bigr)$
 should be.
Above we made the ansatz
\begin{equation} \label{eq_ Det_disc_FP0} \Det^{disc}_{FP}(B) =
 \prod_{x \in \face_0(q\cK)}  \det\bigl(1_{{\ck}}-\exp(\ad(B(x)))_{| {\ck}}\bigr)
\end{equation}
We will now modify Eq. \eqref{eq_ Det_disc_FP0} and make instead the ansatz
\begin{equation} \label{eq_ Det_disc_FP1}
\Det^{disc}_{FP}(B) :=
 \prod_{x \in \face_0(q\cK)}  \det\nolimits^{1/2}\bigl(1_{{\ck}}-\exp(\ad(B(x)))_{| {\ck}}\bigr)
\end{equation}
where $\det^{1/2}\bigl(1_{{\ck}}-\exp(\ad(\cdot))_{| {\ck}}\bigr): \ct \to \bR$
is any one of the two\footnote{Since $\det\bigl(1_{{\ck}}-\exp(\ad(b))_{| {\ck}}\bigr)   =  \prod_{{\alpha} \in {\cR}} (1 - e^{2 \pi i \langle \alpha, b \rangle})     = \prod_{{\alpha} \in {\cR}_+} \bigl( 4 \sin^2( \pi    \langle \alpha, b \rangle ) \bigr)$ for all $b \in \ct$
where $\cR$ is the set of real roots associated to $(\cG,\ct)$ and $\cR_+ \subset \cR$
is the subset of positive real roots (w.r.t. a fixed Weyl chamber)
 we conclude that either  $\forall b \in \ct: \det\nolimits^{1/2}\bigl(1_{{\ck}}-\exp(\ad({b}))_{|{\ck}}\bigr) =  \prod_{{\alpha} \in {\cR_+}} \bigl( 2 \sin( \pi    \langle \alpha, b \rangle ) \bigr)$
  or $\forall b \in \ct: \det\nolimits^{1/2}\bigl(1_{{\ck}}-\exp(\ad({b}))_{|{\ck}}\bigr) =  - \prod_{{\alpha} \in {\cR_+}} \bigl( 2 \sin( \pi    \langle \alpha, b \rangle ) \bigr)$} {\em smooth} functions $f:\ct \to \bR$ fulfilling $f(b)^2 = \det\bigl(1_{{\ck}}-\exp(\ad(b))_{| {\ck}}\bigr)$ for all $b \in \ct$.

\begin{remark}  \rm
At the moment we do not have a totally clear understanding of the
 fact that we have to  include the exponent $1/2$ in Eq. \eqref{eq_ Det_disc_FP1}
if we want to obtain the correct values for the WLOs.
This point seems to get clearer
after making  the transition to the $BF_3$-theoretic setting, see  Sec. \ref{sec7} below
 and Sec. 7 in \cite{Ha7b}. \par
Alternatively, one can simply bypass this point by rewriting the heuristic Eq. \eqref{eq2.48} in a suitable way
before discretizing it, cf. Sec. 2.5 and Sec. 3.6 in \cite{Ha9}.
  \end{remark}

\subsubsection*{Modification (Mod2)}

In the following we will replace\footnote{so, in particular, $DB$ will denote the normalized Lebesgue
measure on $\cB_0(q\cK)$ where we have equipped $\cB_0(q\cK)$ with the scalar product induced
by the one on $\cB(q\cK)$} the space $\cB(q\cK) = C^0(q\cK,\ct)$  appearing on the RHS of Eq. \eqref{eq_def_WLOdisc}
by a certain subspace $\cB_0(q\cK)$ (chosen as naturally
as possible) with the property that
\begin{equation} \label{eq_obs1}
\ker(\pi \circ (d_{q\cK})_{| \cB_0(q\cK)}) = \cB_{c}(q\cK)
\end{equation}
holds where $\pi: C^1(q\cK,\ct) \to C^1(K,\ct)$ is the orthogonal projection w.r.t.
$\ll \cdot, \cdot \gg_{\cA_{\Sigma}(q\cK)}$ and where we have set
\begin{equation}
\cB_{c}(q\cK):= \{ B \in C^0(q\cK,\ct) \mid B \text{ constant}\}
\end{equation}
We remark that Eq. \eqref{eq_obs1} will play a crucial role in the proof of Theorem \ref{main_theorem}.
(We also remark that we cannot choose simply $\cB_0(q\cK) =  \cB(q\cK)$
because $\ker(\pi \circ d_{q\cK}) \neq \cB_{c}(q\cK)$).

\smallskip

The following choice is probably the best  when working with  simplicial ribbons in $\cK \times \bZ_N$ instead of
simplicial ribbons in $q\cK \times \bZ_N$ (cf. Remark \ref{rm_full_ribbons} below,
 Sec. \ref{sec7} in \cite{Ha7b}, and \cite{Ha9}).

\begin{choice} \label{example3} \rm
$\cB_0(q\cK):= \psi(\cB(\cK))$ where $\cB(\cK):= C^0(\cK,\ct)$ and where $\psi: \cB(\cK) \to \cB(q\cK)$
is the linear injection which associates to each $B \in  \cB(\cK)$ the extension $\bar{B} \in \cB(q\cK)$
given by\footnote{In other words,  $\cB_0(q\cK):= \{ B  \in \cB(q\cK) \mid
B(\bar{F}) = \mean_{x \in \face_0(\cK) \cap F} B(x) \ \forall F \in
 \bigcup_{p=0}^2 \face_p(\cK) \}$
 where $\bar{F}$ is the barycenter of $F$ (cf. Convention \ref{conv_identif} in part \ref{appF} of the Appendix below)}
$$\bar{B}(x) = \mean_{y \in C(x)} B(y) \quad \quad \text{for all $x \in \face_0(q\cK)$} $$
with  ``mean'' referring to the arithmetic mean.
Above $C(x)$ denotes the set of all $y \in \face_0(\cK)$ which
 lie in the closure of the unique open cell of $\cK$  containing $x$.
\end{choice}

In the present paper and the main part of  \cite{Ha7b}
we will work with simplicial ribbons in   $q\cK \times \bZ_N$. In this case Choice \ref{example3} will not work
and we will   make  the following choice (which is a bit technical but justified by Remark \ref{rm_loc_constant} above):

\begin{choice} \label{example1} \rm

$\cB_0(q\cK) :=  \cB^{loc}_{\sigma_0}(q\cK) \cap \cB_{\aff}(q\cK)$
 with
\begin{align*} \cB^{loc}_{\sigma_0}(q\cK) & := \{ B \in \cB(q\cK) \mid B \text{ is constant on $U(\sigma_0) \cap \face_0(q\cK)$}  \}, \quad \\
 \cB_{\aff}(q\cK) & := \{ B \in \cB(q\cK) \mid B \text{ is affine on each $F \in \face_2(q\cK)$} \}
 \end{align*}
Here
$  U(\sigma_0) \subset \Sigma$ is the union of those few $F \in \face_2(q\cK)$ which contain the point $\sigma_0$
and by ``$B$ is affine on $F$'' we mean that
\begin{equation} \label{eq_for_aff_B} B(p_1)+B(p_4) = B(p_2)+B(p_3)
\end{equation}
 holds where $p_1, p_2, p_3,p_4$ are the four vertices
of $F$ and numbered in such a way that $p_1$ is diagonal to $p_4$ and therefore
$p_2$ is diagonal to $p_3$.\par

Clearly, $ \cB^{loc}_{\sigma_0}(q\cK)$ is a simplicial analogue of the space $ \cB^{loc}_{\sigma_0}$
appearing in Remark \ref{rm_loc_constant} above.

\end{choice}

\section{The main result}
\label{sec6}

\subsection{A special class of simplicial links}
\label{subsec6.1}

As a preparation for Sec. \ref{subsec6.2} below let us
consider briefly a simple class
of  simplicial links $L= (l_1, l_2, \ldots, l_m)$ in $q\cK \times \bZ_N$.
This class consists of  those simplicial links
 $L$ which fulfill  the following two conditions

\begin{description}
\item[(NCP)]  The  link $L$ has no crossing points,
 i.e. the loops $l^1_{\Sigma},l^2_{\Sigma},\ldots,l^m_{\Sigma}$ in $\Sigma$
 neither intersect each other nor themselves\footnote{i.e.  if $i \neq j$ then  $l^i_{\Sigma}$ and $l^j_{\Sigma}$ have disjoint images and each $l^i_{\Sigma}$, $i \le m$  is an embedding  $S^1 \to \Sigma$}.

\item[(NH)]
Each $l^i_{\Sigma}$ is null-homologous.
\end{description}
Here by $l^i_{\Sigma}$ we denote the (reduced) $\Sigma$-projection\footnote{more precisely: $l^i_{\Sigma}$ is the simplicial loop in $q\cK$ given by $l^i_{\Sigma}:=(l_i)^{red}_{\Sigma}$ in  the notation of Sec. \ref{subsec4.0.4d}}
 of $l_i$, cf. Sec. \ref{subsec4.0.4d} above. Moreover,
  we consider each simplicial loop $l^i_{\Sigma}$  as a continuous map
$S^1 \to \Sigma$ in a natural way (cf.  Sec. \ref{subsec4.0.3} above).
In view of Remark \ref{rm3.3} below let us also introduce the notation $l^i_{S^1} := (l_i)_{S^1}$.

\medskip

Figures \ref{fig2b}--\ref{fig1} below
show examples for links fulfilling  (NCP)
\begin{figure}[h]
  \centering
  \begin{minipage}[b]{6 cm}
  \includegraphics[height=2.5cm,width=5 cm, angle=0]{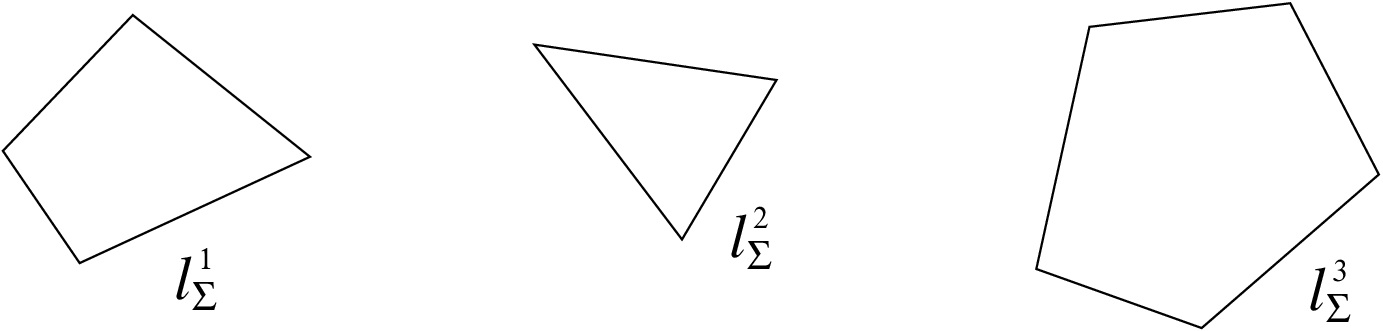}
    \caption{}     \label{fig2b}
   \end{minipage}
  \begin{minipage}[b]{3 cm}
   \ \
    \end{minipage}
   \begin{minipage}[b]{6 cm}
   \includegraphics[height=3cm,width=5 cm, angle=0]{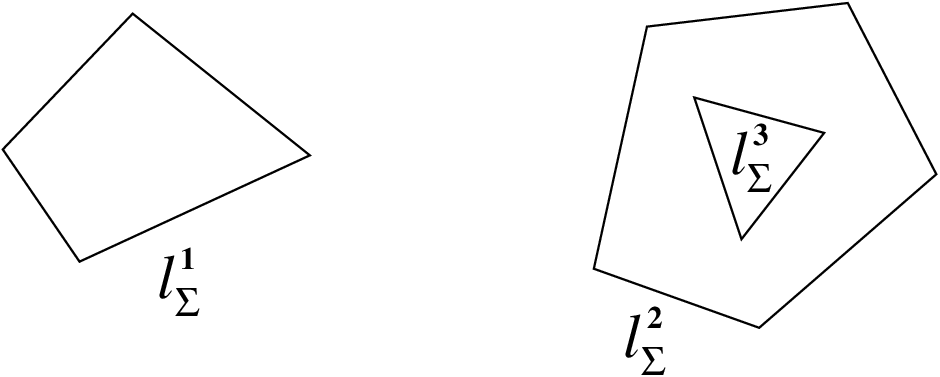}
    \caption{}     \label{fig2a}
  \end{minipage}
\end{figure}

 \begin{figure}[h]
\begin{center}
  \includegraphics[height=6cm,width=9 cm]{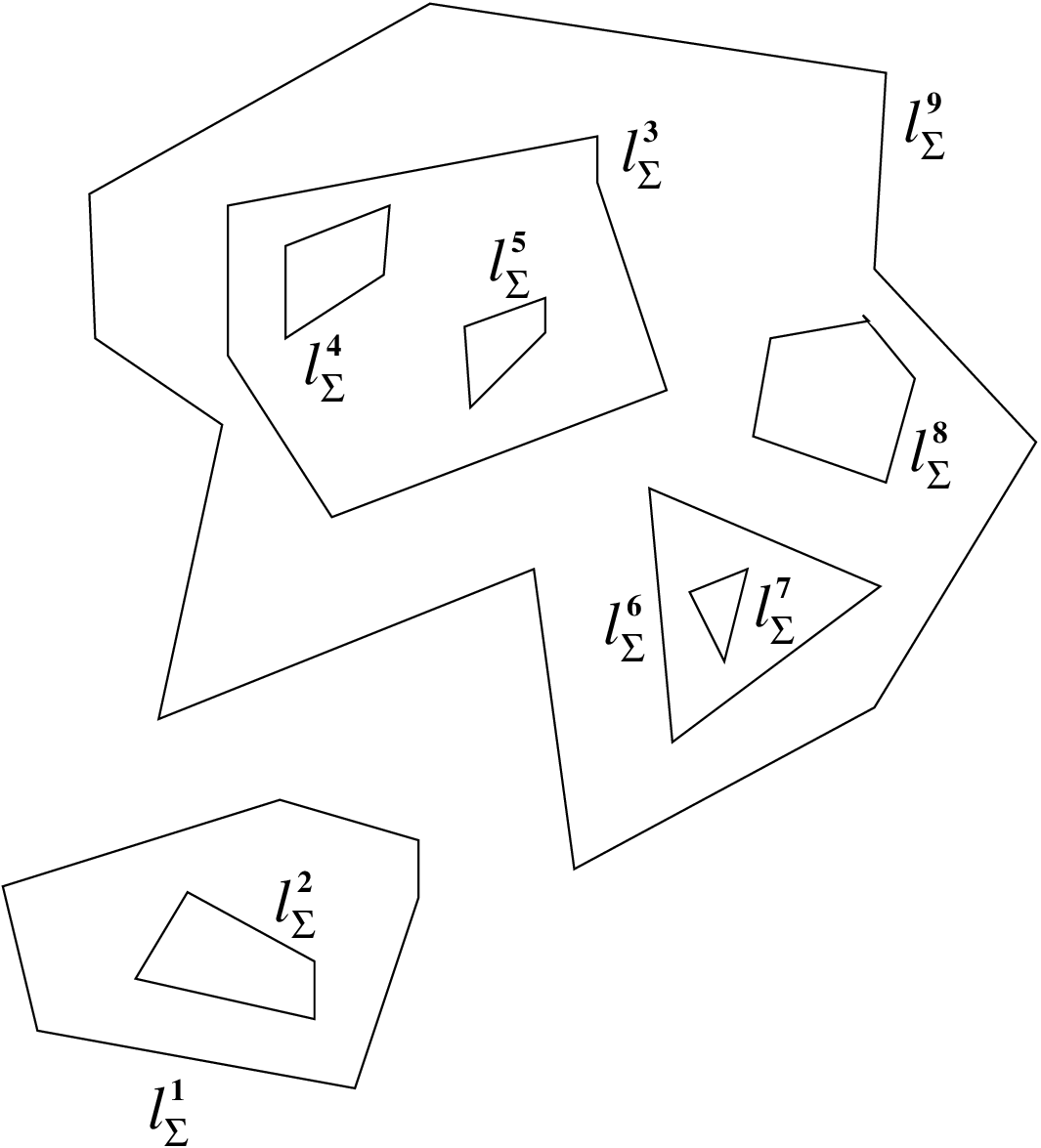}
\caption{}  \label{fig1}
\end{center}
\end{figure}

\begin{remark}  \label{rm3.3} \rm
 \begin{enumerate}
\item     Observe  that  (NCP) and (NH)  place no restrictions
  on the $S^1$-projections  $l^i_{S^1}$ of the loops $l_i$,  $i \le m$,
    and so in general the link $L$ will not be equivalent\footnote{for example, this applies to the links in
       Fig.  \ref{fig2a} and  Fig.  \ref{fig1} if, e.g., $l^i_{S^1}= \id_{S^1}$
       holds for $i \le 3$ and $i \le 9$, respectively; by contrast,  the link in Fig.  \ref{fig2b} will be equivalent         to a vertical link if $l^i_{S^1}= \id_{S^1}$ for $i \le 3$} to a link with the property
    that there is a sequence $(D_i)_{i \le m}$
    of pairwise disjoint disks $D_i \subset \Sigma$
    such that for each $i \le m$ the arc of  $l^i_{\Sigma}$ is contained in $D_i$.
                In particular, a link fulfilling conditions  (NCP) and (NH)
                 will in general not be equivalent to
     a link consisting of only ``vertical'' loops,
      i.e. loops whose $\Sigma$-projections
     are ``points'', see Fig. \ref{vertical_link} below for an example (cf. also  Remark \ref{rm2.5} above).
\begin{figure}[h]
\begin{center}
   \includegraphics[height=1.6cm,width=3in, angle=0]{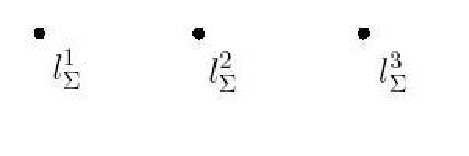}
    \caption{A vertical link consisting of three loops}   \label{vertical_link}
\end{center}
 \end{figure}

\item As a preparation for part ii) of Remark \ref{rm_more_general_links}  below we mention
that it would also be interesting to study the weaker version of condition (NCP)
which we obtain when instead of demanding that $l^i_{\Sigma}$   is an embedding  $S^1 \to \Sigma$
  itself we only demand that the image of  $l^i_{\Sigma}$ lies on  the image $C$ of an embedding $S^1 \to \Sigma$. Clearly, the image of $l_i$ will then lie on the torus $ C \times S^1 \cong  S^1 \times S^1$
 so  $l_i$ will be a special type of torus knot in $\Sigma \times S^1$.
   Let $p \in \bZ$ and $q \in \bZ$  be the two winding numbers of $l_i$ (considered as a map $S^1 \to C \times S^1$)  around the first and the second component of $C \times S^1$.
   The original version of (NCP) implies that $p = \pm 1$ (while $q$ can be an arbitrary integer).
  The weakened version of (NCP) just described implies nothing about $p$ so
  if we use the weakened version   both $p$ and $q$ can be arbitrary integers.
\end{enumerate}
 \end{remark}

In the following we will not work with simplicial loops and simplicial links
but instead with (closed) simplicial ribbons and with simplicial ribbon links.
 This should not be a surprise. It is very well known
 in the physics and mathematics literature on quantum 3-manifold  invariants
 that one must indeed work with framed links (or,  equivalently, ribbon links,
   cf. the beginning of Sec. \ref{subsec4.0.3b} above) if one wants to get meaningful results.\par

In fact, within the framework we use in the present paper we can give a very concrete argument\footnote{as a side remark we mention that  Remark \ref{rm_appB'_3} in part \ref{appB'_3} of the appendix  below gives another argument in favor of the use of simplicial ribbons instead of simplicial loops} which shows that things would go wrong if we used just
simplicial loops (without involving a framing or a simplicial ribbon in an explicit way):
 the crossing points of simplicial loops would  go ``undetected''\footnote{more precisely, if the simplicial
 loops $l^i_{\Sigma}$ are realized as simplicial loops in $K_1$ (or $K_2$)
 then all crossing points will go undetected in the sense that
 as in Sec. 5.1 in \cite{Ha7b} all relevant covariance expressions will vanish.
 If the simplicial
 loops $l^i_{\Sigma}$ are realized as simplicial loops in $q\cK$
 we would have the rather artificial phenomenon that some crossing points will go undetected
 while others will not}.
   Consequently, when evaluating $\WLO^{disc}_{rig}(L)$
there would be no chance of obtaining a factor like $|L|^{\varphi}_4$ in Eq. \eqref{eq_shadoinv} below.
 So when working with simplicial loops instead of simplicial ribbons
  there is no hope for finding a
 generalization\footnote{In fact, as we will explain in  Sec. 6 in \cite{Ha7b}, it is probably not possible
 to find a such  generalization of Theorem \ref{main_theorem} anyway, unless we perform additional modifications  like e.g. switching to the $BF_3$-theoretic setting.
 But also then we will have to work with simplicial ribbons
 (or with framed simplicial  loops)} of Theorem \ref{main_theorem} below
 which includes the case of   general links.

\subsection{A special class of simplicial ribbon links}
\label{subsec6.2}

From now on we will assume that the simplicial ribbon link $L= (R_1, R_2, \ldots, R_m)$
 in $q\cK \times \bZ_N$
fixed in Sec. \ref{subsec4.9} above
fulfills the following two conditions,
which are the ribbon analogues
 of the two conditions (NCP) and (NH) appearing in Sec. \ref{subsec6.1} above.

   \begin{description}
   \item[(NCP)'] The maps $R^i_{\Sigma}$, $i \le m$, neither intersect each other
    nor themselves\footnote{i.e. if $i \neq j$ then  $R^i_{\Sigma}$ and $R^j_{\Sigma}$ have disjoint images
    and i.e. each $R^i_{\Sigma}$, $i \le m$, considered
   as a continuous map $[0,1] \times S^1 \to \Sigma$, is an embedding.}.
    \item[(NH)'] Each of the maps $R^i_{\Sigma}$, $i \le m$ is null-homotopic.
  \end{description}
  where $R^i_{\Sigma}$ is defined as in Sec. \ref{subsec4.9} above.

\begin{remark} \label{rm_more_general_links} \rm
\begin{enumerate}
\item  In view of the discussion in Remark \ref{rm2.5} in Sec. \ref{subsubsec2.3.4} above
we remark that Theorem \ref{main_theorem} below (and its proof) can
easily be generalized to the situation of
 (colored) simplicial ribbon links
$L= (R_1, R_2, \ldots,  R_{m+r})$, $m, r \in \bN$ with the following properties:
\begin{itemize}

\item  the sub ribbon link $(R_1, R_2, \ldots, R_{m})$
fulfills conditions (NCP)' and (NH)'.

\item  each    $R_{m+j}$,  $j \le r$,  is a ``vertical''  simplicial ribbon
 in the sense that\footnote{this is equivalent to saying that
 each face $F_i$ appearing in  $R_{m+j}$ is  vertical w.r.t  $\Sigma$ in the sense of Sec. \ref{subsec4.0.4d} above} the set $e_{m+j}:=  \Image(\pi_{\Sigma} \circ R_{m+j})$ is an edge in $q\cK$.
 (Here we consider $R_{m+j}$ as a continuous map $[0,1] \times  S^1
 \to \Sigma \times S^1$).

\item we have $e_{m+j} \cap \Image(R^{i}_{\Sigma}) = \emptyset$ for all $j \le r$ and $i \le m$.
\end{itemize}
For this more general type of simplicial ribbon links Eq. \eqref{eq_maintheorem} below will still hold
if we generalize the definition of $|L|$ in a suitable way\footnote{more precisely, using
 the notation of Remark \ref{rm2.5} above and Remark \ref{rm3.4} below we must set
$$|L| := \sum_{\vf\in col(L)}
|L|_1^{\vf}\,|L|_2^\vf\,|L|_3^\vf \bigl( \prod_{i=m+1}^{m+r}  \tfrac{S_{ \vf_i \mu_i}}{S_{\vf_i 0}} \bigr)$$
 where $\mu_i$ is the  highest weight of the representation $\rho_i$ and  where for $i \in \{m+1, \ldots m+r\}$ we have set
  $\vf_i:=\vf(Y(i))$ where $Y(i)$ is the unique element
   of $F(L) = \{Y_0, Y_1, \ldots, Y_m\}$ containing $e_{i}$}.

\item  Theorem \ref{main_theorem} below (and its proof) can
also be generalized in another direction.
In fact, we can weaken condition (NCP)' in the following way (in the spirit of  part ii) of Remark \ref{rm3.3} above): Instead of demanding that each   $R^i_{\Sigma}$, $i \le m$, considered
   as a continuous map $[0,1] \times S^1 \to \Sigma$, is an embedding itself we only demand that
   the image of  $R^i_{\Sigma}$ lies on the image of such an embedding $[0,1] \times S^1 \to \Sigma$.
  Even when using this weaker version of condition (NCP)' we would still be able to evaluate
  $\WLO_{rig}(L)$ explicitly using a straightforward modification of the approach used in Sec. 5 in \cite{Ha7b}
  for proving Theorem \ref{main_theorem}.
\end{enumerate}
\end{remark}

\begin{remark} \label{rm_full_ribbons}
\rm   Instead of working with simplicial ribbons (and ribbon links) in $q\cK \times \bZ_N$
   one can also work with simplicial ribbons in
     $\cK \times \bZ_N$. In order to do so we only have to  modify the constructions  in Sec. \ref{subsubsec4.3.2}
    above in a more or less straightforward way, cf. Remark 3.2 in \cite{Ha7b}.\par

    The use of simplicial ribbons in $\cK \times \bZ_N$
    has several advantages but also two disadvantages over the use of
    simplicial ribbons in $q\cK \times \bZ_N$.
    The disadvantages are,  firstly, that the analogue of expression on the RHS of Eq. \eqref{eq4.21} in \ref{subsubsec4.3.2}
        for simplicial ribbons in $\cK \times \bZ_N$ is more complicated than the original expression
    on the RHS of Eq. \eqref{eq4.21}.     Secondly,   working with simplicial ribbons in $\cK \times \bZ_N$
    seems to require an additional regularization procedure,    cf.  Remark 5.4 in \cite{Ha7b}.\par

    The two main advantages of working with  simplicial ribbons in $\cK \times \bZ_N$
    are that in this case we can work with Choice \ref{example3} in Sec. \ref{subsec4.10}  above,
    which is clearly more natural than the use of Choice \ref{example1}.
         More importantly, the use of simplicial ribbons in $\cK \times \bZ_N$ is much better suited
     to deal with general simplicial ribbon links $L$, ie simplicial ribbon links
     not fulfilling a condition of the type (NCP)' above, cf. Remark 5.4 in \cite{Ha7b}. \par

    However, since the rest of the present paper and  the main part of \cite{Ha7b}
    will only be concerned with the study of special simplicial ribbon links (cf. condition (NCP)') the last
    advantage will not be relevant until Sec. 7 in \cite{Ha7b}. This is why
    we decided to work with simplicial ribbons in $q\cK \times \bZ_N$
    until the end of Sec. 5 in \cite{Ha7b}.

\end{remark}

\subsection{The main result}
\label{subsec6.3}

Recall that in Sec. \ref{subsec4.9} above we fixed a simplicial ribbon link $L= (R_1, R_2, \ldots, R_m)$ in $q\cK \times \bZ_N$ with colors $(\rho_1, \rho_2, \ldots, \rho_m)$. For $i \le m$ let $\lambda_i \in \Lambda_+$
denote the highest weight of $\rho_i$. Here  $\Lambda_+$ is the set of dominant real weights of
$\cG$ (w.r.t. $\ct$ and a fixed Weyl chamber).

\begin{theorem} \label{main_theorem}
Assume that $L= (R_1, R_2, \ldots, R_m)$  fulfills conditions
(NCP)' and (NH)' above.
 Assume also that\footnote{the situation  $0 < k < \cg$ is not interesting since in this case
the set $\Lambda^k_+$ appearing below is empty, cf. Remark B.1  in   Appendix B in \cite{Ha7b}.
 Accordingly,  $|L| = |\emptyset| = 0$. It turns out that we then also have $\WLO^{disc}_{rig}(L) = \WLO^{disc}_{rig}(\emptyset) = 0$,
 cf. Remark \ref{rm_WLOdisc} below } $k \ge \cg$ where $\cg$ is the dual Coxeter number of $\cG$
 and that $\lambda_i \in \Lambda_+^k\subset \Lambda_+ $, $i \le m$, where $\Lambda_+^k $ is as in  Remark \ref{rm3.4} below. Then  $\WLO_{rig}(L)$  is well-defined and we have
 \begin{equation} \label{eq_maintheorem}
 \WLO_{rig}(L) = \frac{|L|}{|\emptyset|}
 \end{equation}
 where $\emptyset$ is the ``empty link'' and
  where $|\cdot|$  is the  shadow invariant associated
 to  $\cG$ and $k$, cf. Remark \ref{rm3.4} below and
   Appendix B in \cite{Ha7b} for  the full definitions.
 \end{theorem}

Theorem \ref{main_theorem} will be proven in \cite{Ha7b}.

\begin{remark}  \label{rm_app1} \rm We emphasize  that $|\cdot|$ here is really the shadow invariant associated
 to  $\cG$ and $k$ and {\it not} to $\cG$ and $k + \cg$ where $\cg$ is the dual Coxeter number of $\cG$.
 In other words: we do not have a ``shift in $k$'', cf.
 Remark \ref{rm_shift_in_k}  in Sec. \ref{subsec3.1}  above.
 \end{remark}

\begin{remark} \rm \label{rm3.4}
As mentioned in Remark \ref{rm_more_general_links} above
Theorem \ref{main_theorem} can  be generalized in at least two ways.
In particular, by weakening condition (NCP)' in a suitable way one can deal with the case
where some of the simplicial ribbons $R_i$, $i \le m$, contained in $L$
are   ribbon analogues of framed torus knots, see \cite{Ha9}.
We emphasize here that
simplicial ribbon links $L$ fulfilling  the {\em original} version of condition (NCP)'   (and condition  (NH)')
are more interesting than they might appear at first sight.
In fact,  the expression for the shadow invariant $|L|$
   for such ribbon links is quite complicated.
   More precisely, using the notation $\Lambda_+^k$ for the set of ``dominant real weights of $\cG$
   (w.r.t. to $\ct$ and a  fixed Weyl chamber) which are integrable at level $k - \cg$'' (see
     Appendix A in \cite{Ha7b} and cf. also Remark \ref{rm2.5} above)
   and denoting by $F(L)$ the set  of connected components $\{Y_0,Y_1, \ldots, Y_m\}$ of\footnote{equivalently,
   we could work with the connected  components of $\Sigma \backslash
   \bigl( \bigcup_{ i \le m} \Image(R^i_{\Sigma}) \bigr)$} $\Sigma \backslash
   \bigl( \bigcup_{ i \le m} \arc(l^i_{\Sigma}) \bigr)$  we have\footnote{cf. footnote \ref{ft_turaev_norm} below}
      \begin{equation} \label{eq_shadoinv_simpl}
|L|= \sum_{\vf\in col(L)}
|L|_1^{\vf}\,|L|_2^\vf\,|L|_3^\vf
\end{equation}
where $col(L)$  is the set of all maps
$F(L) \to \Lambda_+^k$ ( ``area colorings'')
and where
  \begin{subequations} \label{eqA.5}
\begin{align} |L|_1^\vf&=\prod_{Y \in F(L)} \dim(\vf(Y))^{\chi(Y)}\\
|L|_2^\vf&= \prod_{Y \in F(L)}   \exp(\tfrac{\pi i}{{k}} \langle \vf(Y),\vf(Y) +2\rho\rangle)^{\gleam(Y)}\\
\label{eq_XL3}  |L|_3^\vf&= \prod_{e \in E(L)} N_{\gamma(e)
\varphi(Y^+_e)}^{\varphi(Y^-_e)}
  \end{align}
 \end{subequations}
 Here $\chi(Y)$ is the Euler characteristic of the region $Y \in F(L)$,
   $\gleam(Y) \in  \bZ$ is the so-called ``gleam'' of $Y$,
   and  $\rho$ is again the Weyl vector (cf. Remark \ref{rm2.5}  above).
   Moreover, we have set\footnote{$\dim(\lambda)$ is called the ``quantum dimension'' of $\lambda \in \Lambda^k_+$ in \cite{turaev}}
  $\dim(\lambda):=  S_{\lambda 0}/S_{00}$    for $\lambda \in \Lambda^k_+$ where $(S_{\mu \nu})_{\mu, \nu \in \Lambda^k_+}$ is the $S$-matrix associated to  $U_q(\cG_{\bC})$ with $q:= \exp( \tfrac{2 \pi i}{k})$ (cf. Remark \ref{rm2.5} above).
  Finally,  $N_{ \mu \nu}^{\lambda} \in \bN_0$, $\lambda, \mu, \nu \in \Lambda^k_+$
   are the corresponding ```fusion coefficients''\footnote{We mention that the fusion coefficients  are closely related to the Verlinde numbers.
   In fact we have $ N_{ \mu \nu}^{\lambda} = N_{\lambda^* \mu \nu}$
 where $\lambda^*$ is essentially the weight conjugated to $\lambda$
 (up to a shift in $\rho$).
 Let us emphasize that this time, however, the fusion numbers $N_{ \mu \nu}^{\lambda}$
 enter the computation not via an expression involving the $S$-matrix like in   Remark \ref{rm2.5} above
 but by a sum over weight multiplicities (the ``quantum Racah formula''), see Eq. B.7 Appendix B in \cite{Ha7b}
  and \cite{HaHa}.}.
   We will give full details (including  an explanation of $\gleam(Y)$ and the
notation  $E(L)$, $\gamma(e)$, $Y^{\pm}_e$ used above) in Appendix B in \cite{Ha7b}.

\smallskip

   For general  ribbon links $L$ in $\Sigma \times S^1$   the explicit formula
   is\footnote{\label{ft_turaev_norm} This definition of the shadow invariant (in the special case $M =  \Sigma \times S^1$)
   differs from Turaev's definition in Sec. 1.2 in Chap. X in \cite{turaev} by a normalization factor depending on $\cG$ and $k$.
We remark that the factor $|L|_1^{\vf}$ above corresponds to the factor $|X|_2^{\vf}$ in Sec. 1.2 in Chap. X in \cite{turaev},
  the factor $|L|_2^{\vf}$ corresponds to the factor $|X|_3^{\vf}$, the
  factor $|L|_3^{\vf}$ above corresponds to the factor $|X|_4^{\vf}$, and the  factor
  $|L|_4^{\vf}$ above corresponds to the product $|X|_1^{\vf} |X|_5^{\vf;contr}$
  where $|X|_5^{\vf;contr}$ is the number obtained by
   contracting the term $ \cdot |X|_5^{\vf} $ with the tuple $(\rho_1, \rho_2, \ldots, \rho_m)$ of colors of $L$ }
 \begin{equation} \label{eq_shadoinv}
      |L|= \sum_{\vf\in col(L)}
|L|_1^{\vf}\,|L|_2^\vf\,|L|_3^\vf\,|L|_4^\vf
\end{equation}
Here the factors $|L|_1^{\vf}$ and $|L|_2^\vf$ are as above
and  $|L|_3^\vf$ is given as in Eq. \eqref{eq_XL3} but with the product $\prod_{e \in E(L)} $
replaced by the product $\prod_{e \in E_*(L)} $ where $E_*(L)$ is a certain subset\footnote{cf. the
definition of the so-called ``circle-1-strata'' in \cite{turaev};
for the special links $L$ appearing above we have
$E_*(L) = E(L)$} of $E(L)$.
  The factor $|L|_4^{\vf}$ in Eq. \eqref{eq_shadoinv} is the most interesting one.
  Its definition contains a factor of the form $ \prod_{x \in V(L)} T(x,\vf)$
 where $V(L)$ is the set of crossing points of the loops $\{ l^i_{\Sigma} \mid i \le m\}$ in $\Sigma$
 and where the factors $T(x,\vf)$ involve the so-called
  ``quantum 6j-symbols''  associated to $U_q(\cG_{\bC})$ with $q:= \exp( \tfrac{2 \pi i}{k})$.
         Observe that   three of the four factors $|L|^{\varphi}_1$,  $|L|^{\varphi}_2$,
     $|L|^{\varphi}_3$, $|L|^{\varphi}_4$
  appearing in the formula for the shadow invariant  of a general ribbon link $L$
    also appear in Eq. \eqref{eq_shadoinv_simpl}. \par

  The fact that we can obtain the RHS  of Eq. \eqref{eq_shadoinv_simpl}
   directly from the CS path integral is (hopefully) interesting by itself but, of course,
  we   are mainly interested in the computation
    of $\WLO_{rig}(L)$ for general ribbon links.
  We will come back to this point in  Sec. 6 in \cite{Ha7b} where we will study the case
  of general ribbon links in more detail.

\end{remark}

 \begin{remark} \label{rm_WLOdisc}  \rm The explicit expression for
   $\WLO^{disc}_{rig}(L)$   is
   \begin{equation} \label{eq_rm3.4}
    \WLO^{disc}_{rig}(L) = c_1   k^{c_2} \bigl(\prod_{\alpha \in \cR_+} \sin(\tfrac{\pi}{k} \langle \rho,\alpha \rangle)\bigr)^{\chi(\Sigma)} \ |L|
    \end{equation}
    where $\cR_+$ is the set of positive real roots of $\cG$ (w.r.t. to $\ct$ and a  fixed Weyl chamber)
    and  where $c_1, c_2 \in \bC$ only depend on $G$, $\cK$,  $N$, and $\sigma_0$
     but not on $k$.   (We omit the precise formulas for $c_1$, $c_2$).

\end{remark}

\section{Outlook}

\label{sec7}

Theorem \ref{main_theorem} will be proven in \cite{Ha7b}. This is the first of the two main issues
in \cite{Ha7b}. The second important issue in \cite{Ha7b} is the transition
to the $BF_3$-theoretic setting we referred to at the beginning of Sec. \ref{sec4} above.
The $BF_3$-theoretic setting is more complicated than the original CS-theoretic setting
and in the case of non-Abelian structure groups $G$
the advantages of the $BF_3$-theoretic setting are not as obvious as in the Abelian case (cf. Remark 7.2 in
\cite{Ha7b}). However, a closer look shows that also for non-Abelian structure groups $G$
the  $BF_3$-theoretic setting has several important advantages:

\begin{enumerate}

\item As we will explain in Sec. 6 in \cite{Ha7b}
 it is probably not possible   to generalize  Theorem \ref{main_theorem}
    to the case of general ribbon links unless we modify our approach in a suitable way.
 The transition to the $BF_3$-theoretic setting
 could provide one way (and, possibly, the most natural way)
  to generalize Theorem \ref{main_theorem} successfully.

 \item The $BF_3$-theoretic setting leads to certain stylistic improvements, cf. Remark \ref{rm_for_appE} above
 and Appendix D in \cite{Ha7b}.

 \item The $BF_3$-theoretic setting leads to a better understanding
     of the $1/2$-exponents appearing   in (Mod1) in Sec. \ref{subsec4.10} above, cf.
     Appendix D in \cite{Ha7b}.

\end{enumerate}

\noindent
Other problems/questions which we plan to study in the near future are (cf. \cite{Ha9,Ha8}):

\begin{itemize}

\item[(P1)]  In part ii) of Remark \ref{rm_more_general_links} above  we mentioned that also for simplicial ribbon links $L$    fulfilling\footnote{recall that in this case
 some of the simplicial ribbons $R_i$, $i \le m$, contained in $L$
are allowed to be   ribbon analogues of non-trivial framed torus knots}
 only a  weaker version of condition (NCP)' we can  evaluate
   $\WLO_{rig}(L)$ explicitly using a suitable modification of the approach used in Sec. 5 in \cite{Ha7b}
   for proving Theorem \ref{main_theorem}.
   Do the values obtained for   $\WLO_{rig}(L)$ for these $L$ coincide with the values expected in the literature
   for all (simply-connected compact) structure groups $G$?

\item[(P2)] Is it possible to generalize the approach of the present paper to the situation where $\Sigma$
     is not a closed surface but a compact surface with boundary?
     If yes, then is  it also possible\footnote{we remark that if both (P1) and (P2) can be resolved successfully then one would immediately obtain a rigorous definition \& computation
      of the WLOs of arbitrary colored torus knots in $M=S^3$} to implement rigorously Witten's surgery arguments,
     at least in some special cases like the  surgery  $S^2 \times S^1 \to S^3$ performed on a ``vertical''\footnote{cf. Remark \ref{rm3.3} above} loop in  $S^2 \times S^1$?

\item[(P3)] In \cite{BlTh4} torus gauge fixing is a applied to the evaluation of the heuristic  path integral
   for the CS partition function $Z(M)=\WLO(\emptyset)$ where $M$ is a non-trivial $S^1$-bundle.
   It would be interesting to study whether the methods in   \cite{BlTh4} can be generalized
   so that   $\WLO(L)$ can be evaluated at a heuristic level also if $L \neq \emptyset$.
   If this is the case then one should study if it is possible to find a rigorous simplicial realization of the path integral expressions for $Z(M)$ and $\WLO(L)$.

\item[(P4)] If  (P3) can be resolved successfully  then
    it should be  very interesting to compare the new  results and methods   with those of the approach in  \cite{BeaWi,Bea}    where non-Abelian localization is applied to the CS path integral.

\end{itemize}

\bigskip

\noindent
 {\em Acknowledgements:}  I want to thank the anonymous referee of
my paper \cite{Ha4} whose comments motivated me to look for an
alternative approach for making sense of the RHS  of Eq.
\eqref{eq2.41_pre}, which is less technical than the continuum approach
in \cite{Ha3b,Ha4,Ha6}.  This eventually led  to  the present paper and its sequel \cite{Ha7b}.
 Moreover, I  would  like to thank Laurent Freidel for
 pointing out to me the widespread confusion about the ``shift in $k$''-issue,
 cf. Remark \ref{rm_shift_in_k} in Sec. \ref{subsec3.1} above. \par

I am also grateful to Jean-Claude Zambrini for several comments
which led to improvements in the presentation of the present paper. \par

Finally, it is a great pleasure for me to  thank  Benjamin Himpel
for  many useful and important comments and suggestions,  which
not only had a major impact on the presentation and
overall structure of  the present paper
but also inspired me to reconsider the issue
of  discretizing the operator $\partial_t + \ad(B)$
appearing in Eq.  \eqref{eq4.7} above. (This eventually led me
 to the operators \eqref{eq_def_LOp} in Sec. \ref{subsec4.1b} above).

 \renewcommand{\thesection}{\Alph{section}}
\setcounter{section}{0}

\section{Appendix: Lie theoretic notation I}
\label{appB}

For the convenience of the reader we summarize here the Lie theoretic notation  used
in the main part\footnote{this excludes Remark \ref{rm2.5} and Remark \ref{rm3.4} above where some additional notation
is used that will be explained only in Appendix A of \cite{Ha7b}} of the present paper.

\subsection{List of notation, part I}

\begin{itemize}
\item $G$: the simply-connected compact Lie group fixed in Sec. \ref{subsec2.1}
\item $\cG$: the Lie algebra of $G$

\item $T$: the maximal torus of $G$ fixed in Sec. \ref{subsec2.2}
\item $\ct$: the Lie algebra of $T$

\item  $\exp$: the exponential map $\cG \to G$ of $G$.

\item $I \subset \ct$: the kernel of $\exp_{|\ct}:\ct \to T$.

\item $G_{reg}:= \{g \in G \mid g \text{ is regular} \}$.
Recall from \cite{Br_tD} that an element $g $ of $G$ is called ``regular'' iff it is contained in exactly one maximal torus of $G$.

\item $T_{reg}:= T \cap G_{reg}$
\item  $\cG_{reg}:= \exp^{-1}(G_{reg})$
\item  $\ct_{reg}:= \exp^{-1}(T_{reg}) = \exp_{| \ct}^{-1}(T_{reg})$

\item  $\langle \cdot,\cdot \rangle$: the unique $\Ad$-invariant scalar product
  on $\cG$ such that $\langle \Check{\alpha},\Check{\alpha}
\rangle = 2$ holds for every short real coroot $\Check{\alpha}$
of the pair $(\cG,\ct)$, cf.  Appendix A in \cite{Ha7b}.

\item $\ck$: the $\langle \cdot,\cdot \rangle$-orthogonal
complement of $\ct$ in $\cG$

\item $\pi_{\ct}$: the $\langle \cdot,\cdot \rangle$-orthogonal projection $\cG \to
\ct$

\item $\pi_{\ck}$:  the $\langle \cdot,\cdot \rangle$-orthogonal projection $\cG \to
\ck$

\item $P$: the Weyl alcove fixed in Sec. \ref{subsubsec2.2.4}. Recall from \cite{Br_tD} that a Weyl alcove
       is, by definition, a connected component of $\ct_{reg}$.

\item $\cW_{\aff} \subset \Aff(\ct)$:  the affine Weyl group associated to $(\cG,\ct)$,
cf.  Appendix A in \cite{Ha7b}.
We remark that $\cW_{\aff}$ operates freely and transitively on the set of Weyl alcoves.

\end{itemize}

\medskip

\noindent Using  $\langle \cdot,\cdot \rangle$ we can make the identification $ \ct \cong \ct^*$.
 Sometimes we write
$\langle \cdot,\cdot \rangle_{\cG}$  instead of $\langle \cdot,\cdot \rangle$.

\subsection{Example: the special case $G=SU(2)$}

 In the present paper and in \cite{Ha7b}
we  work with general  simply-connected compact Lie groups
$G$ since this does not involve much more work
than we would have to invest if we restricted ourselves
to a special case like, e.g., $G=SU(2)$.
On the other hand treating the general situation
makes it necessary to use several abstract concepts from Lie theory.
The reader who prefers a more elementary treatment should feel free to
 concentrate on the special case $G=SU(2)$ for which we have the following explicit
 formulas:

  \medskip

In the special case $G = SU(2)$ we can choose the maximal torus $T$ as

\smallskip

$T = \{\exp( \theta \tau) \mid  \theta \in \bR\}$ where $\tau :=\left( \begin{matrix} i  && 0 \\ 0 && -i  \end{matrix} \right)$

\smallskip

\noindent Then we have (with $\cong$ meaning ``homeomorphic''
and with the convention  $S  \tau  := \{ t \tau \mid t \in S\}$ for any $S \subset \bR$):
\begin{align*}
 G  & =SU(2) = \{ A \in \Mat(2,\bC) \mid A A^* = 1, \det(A)=1\}  &&\cong S^3 \\
 \cG  & = su(2) = \{ A \in \Mat(2,\bC) \mid A + A^* = 0, \Tr(A)=0\} &&\cong \bR^3\\
 T & = \{\exp( \theta \tau \mid  \theta \in \bR\} =
\left\{ \left( \begin{matrix} e^{i \theta} && 0 \\ 0 && e^{-i \theta} \end{matrix} \right) \mid \theta \in \bR \right\}  && \cong S^1 \\
 \ct &  = \bR \cdot \tau = \{\theta \tau \mid \theta \in \bR \} && \cong \bR\\
 G/T  & = \{gT \mid g \in G \} &&  \cong S^2 \\
 G_{reg} & = SU(2) \backslash \{-1,1\}  && \cong S^2 \times (0,1) \\
\cG_{reg} & = \cG \backslash \bigcup_{n \in \bN_0} \{b \in \cG \mid |b| =n \}  &&
\cong S^2 \times ( \bR_+ \backslash \bN) \\
T_{reg} & = T \backslash  \{-1,1\} && \cong S^1 \backslash \{-1,1\} \\
\ct_{\reg} & = \ct \backslash \{ n \pi \tau \mid n \in \bZ \} && \cong \bR \backslash \bZ
\end{align*}

Moreover, we have

\begin{itemize}

\item $\exp: \cG \to G$ is the restriction  of the  exponential map of $\Mat(2,\bC)$ onto $\cG$, i.e.
$$\exp(A) = \sum_{n=0}^{\infty} \tfrac{A^n}{n!} \in G \subset \Mat(2,\bC)  \quad  \text{ for } \quad A \in \cG \subset \Mat(2,\bC)$$

\item $I =  \bZ \cdot 2 \pi \tau$

\item the scalar product $\langle \cdot, \cdot \rangle$ is given by
$ \langle A, B \rangle  = -\tfrac{1}{4\pi^2} \Tr_{\Mat(2,\bC)}(A \cdot B)$ for all $A, B \in su(2)$.
The norm  $|\cdot|$ appearing in the formula for $\cG_{reg}$
denotes the norm associated to this  scalar product.

\item $\ck = \left\{ \left( \begin{matrix} 0 && -z \\ \bar{z} && 0 \end{matrix} \right)
\mid z \in \bC \right\}$

\item the set of Weyl alcoves is $\{P_n \mid n \in \bZ \}$ where $P_n:= (n \pi, (n+1) \pi) \ \tau$.
Accordingly,  for $P$ we could take, e.g., the set   $P := P_0 = (0, \pi)  \tau$.

\item $\cW_{\aff}$: the subgroup of the affine group  $\Aff(\ct)$ which is generated by
   the reflection $\ct \ni b \mapsto -b \in \ct$ and  the translation
    $\ct \ni b \mapsto b + 2\pi \tau \in \ct$.

\end{itemize}
Let us also mention that
$$\det(1_{\ck} - \exp(\ad(x \tau))_{| \ck}) = 4 \sin^2(x), \quad \quad x \in \bR$$

\section{Appendix: Some technical details for Sec.  \ref{sec2}}
 \label{appB'}

\subsection{Some additional details for Sec. \ref{subsec2.2}}
 \label{appB'_1}

\noindent  i) We will now motivate the choice of the space $\overline{V}$ and the map $\Pi_{\overline{V}}$
 appearing in  Sec. \ref{subsubsec2.2.4}.

\smallskip

 Observe first that $\G = \tilde{\G} \rtimes \G_{\Sigma}$
  where $\tilde{\G}:=  \{ \Omega \in \G \mid \forall \sigma \in \Sigma: \Omega((\sigma,1)) = 1  \}$
  and  where $\rtimes$ denotes the semi-direct product
 ($\tilde{\G}$ being the normal subgroup).
 Taking this into account we easily see that there is a natural
  right-operation of $\G_{\Sigma}$ on $A_{reg}/ \tilde{\G}$
 and that
  \begin{equation}\label{eq2.24_plus1} A_{reg}/ \G \cong( A_{reg}/ \tilde{\G} ) /
\G_{\Sigma}
\end{equation}
Secondly,  if $S$ is a connected component of $\cG_{reg}= \exp^{-1}(G_{reg})$
then  $\exp: S \to G_{reg}$ is a
diffeomorphism\footnote{this follow because $\exp:\cG_{reg} \to G_{reg}$
is a smooth covering and  $G_{reg}$ is simply-connected (that $G_{reg}$ is simply-connected
follows from our assumption that $G$ is simply-connected)}.
  It is not difficult to see that this implies that also
  $q:A^{\orth} \times C^{\infty}(\Sigma,S) \ni (A^{\orth},B) \mapsto
  (A^{\orth} + B dt) \cdot  \tilde{\G} \in  A_{reg}/ \tilde{\G} $ is a bijection, cf. Proposition 3.1 in
\cite{Ha3c}.
Thirdly,  if $P$ is a Weyl alcove contained in $S$
 then the map  $\theta: P \times G/T \ni (b,\bar{g}) \mapsto \bar{g} b
\bar{g}^{-1} \in S$ is a well-defined diffeomorphism\footnote{cf.  Example \ref{rm_SU2_Top} below
 for the special case $G=SU(2)$}.
Thus we have the identification
 \begin{equation}\label{eq2.24_plus2}
 \cA_{reg}/ \tilde{\G} \cong \cA^{\orth} \times
C^{\infty}(\Sigma,S)
 \cong \cA^{\orth} \times C^{\infty}(\Sigma,P)
 \times C^{\infty}(\Sigma,G/T)
 \end{equation}
Note that under this identification the  $\G_{\Sigma}$-operation on
$\cA_{reg}/ \tilde{\G}$ mentioned above induces the  $\G_{\Sigma}$-operation on
$\cA^{\orth} \times C^{\infty}(\Sigma,P)
 \times C^{\infty}(\Sigma,G/T)$ which is given by
 $(A^{\orth},B, \bar{g}) \cdot \Omega =  (A^{\orth}\cdot \Omega,B,
 \Omega^{-1}\bar{g})$ for each $\Omega \in \G_{\Sigma}$.\par

 Fourthly, observe that the map $p: \cA^{\orth} \times C^{\infty}(\Sigma,P)
 \times C^{\infty}(\Sigma,G/T)/\G_{\Sigma}
 \to \bigl( \cA^{\orth} \times  C^{\infty}(\Sigma,P)
 \times C^{\infty}(\Sigma,G/T) \bigr) / \G_{\Sigma}$
which maps each $(A^{\orth},B, \cl)$ to the $\G_{\Sigma}$-orbit of
$(A^{\orth},B, \bar{g}_{\cl})$  is a surjection.

\smallskip

We have just seen how the set $ C^{\infty}(\Sigma,G/T)/\G_{\Sigma} = [\Sigma,G/T]$ arises  naturally.
 Moreover, by replacing the space  $C^{\infty}(\Sigma,P)$,
 which looks a bit technical, by the space $\cB = C^{\infty}(\Sigma,\ct)$
 we arrive at the space $\overline{V}$.
 It is easy to check that, under the two identifications
\eqref{eq2.24_plus1} and \eqref{eq2.24_plus2} $p$ coincides with  the restriction of
 $\Pi_{\overline{V}}$ to $\cA^{\orth} \times
 C^{\infty}(\Sigma,P)  \times [\Sigma,G/T]$, ie Eq. \eqref{eq2.21} is fulfilled if we replace
 $\Pi_{\overline{V}}$ by $p$. By using that equation to extend $p$
  to the space  $\overline{V}$ we arrive at the map $\Pi_{\overline{V}}$. \par
Finally, it is  clear that
 \begin{equation} \label{eq2.22_app}  \cA_{reg}/ \G = \Pi_{\overline{V}}(\cA^{\orth} \times
 C^{\infty}(\Sigma,P)  \times [\Sigma,G/T]) \subset
 \Image(\Pi_{\overline{V}}),
 \end{equation}
 which implies the inclusion \eqref{eq2.22} in Sec. \ref{subsubsec2.2.4}.

 \medskip

ii) In the derivation of  \eqref{eq2.22} in part i) we did not make use
 of the assumption that $\Sigma$ is compact,
 so in fact relation \eqref{eq2.22} also holds for noncompact $\Sigma$.
 But since\footnote{this follows, e.g., by
 combining the observation in  Footnote \ref{ft11} with the first two observations in Remark \ref{rm2.3}}
the set  $[\Sigma,G/T]$ then just consists of the single point
$[1_T]$ where $1_T$ is the constant function on $\Sigma$ taking only the value $T \in G/T$,
 the relation \eqref{eq2.22}  reduces to relation
\eqref{eq2.13} above.

\medskip

iii) By contrast, if  $\Sigma$ is compact then  relation
\eqref{eq2.13} does {\em not} hold. In order to see this observe
that in this case $[\Sigma,G/T]$  has infinitely many elements, cf.
Remark \ref{rm2.4}. The map $p$ in part i)
 is not injective but it does have the weaker property that
$\cl_1 \neq \cl_2$ implies
 $p(A^{\orth}_1, B_1, \cl_1) \neq  p(A^{\orth}_2, B_2, \cl_2)$.
 Thus for compact $\Sigma$ the set
 $$ \Pi_{\overline{V}}(\cA^{\orth} \times
 C^{\infty}(\Sigma,P)  \times \{ [ 1_T]\}) =
\pi_{\G}(\cA^{\orth} \oplus  C^{\infty}(\Sigma,P) dt) = \pi_{\G}(\cA^{qax}(T)) \cap \cA_{reg}/ \G$$
  will be a proper subset of
 $ \Pi_{\overline{V}}(\cA^{\orth} \times
 C^{\infty}(\Sigma,P)  \times [\Sigma,G/T]) = \cA_{reg}/ \G$, cf.  part i).
 Clearly, this implies  $\cA_{reg} / \G \not\subset \pi_{\G}(\cA^{qax}(T))$.

\begin{example} \label{rm_SU2_Top} \rm
It is probably instructive to verify
some of the claims made above (and some of the claims made in Sec. \ref{subsubsec2.2.4})
directly in the special case where $\Sigma = S^2$ and where $G=SU(2)$.
 \begin{enumerate}

\item Let $S$ be a  connected component of $\cG_{reg}$
 and $P$  a connected component of $\ct_{reg}$
 (ie  a Weyl alcove of $(\cG,\ct)$).
 Above we claimed that $\exp: S \to G_{reg}$ is a diffeomorphism.
 Moreover, we claimed that if $P$ is contained in $S$
 then $\theta: P \times G/T \ni (b,\bar{g}) \mapsto \bar{g} b
\bar{g}^{-1} \in S$ is a (well-defined)  diffeomorphism.
In particular, this means that the three spaces  $G_{reg}$, $ S$, and $ P \times G/T$
are  homeomorphic to each  other.
In the special case $G= SU(2)$ we can verify the  homeomorphy of these three spaces
directly by using the  concrete formulas for $G_{reg}$, $G/T$, $\cG_{reg}$, and $\ct_{reg}$
in part \ref{appB} of the Appendix.

\item   In the special case $G=SU(2)$ we have $G\cong S^3$ and $G/T \cong S^2$
and  the fiber bundle   $\pi_{G/T}:G \to G/T$ turns out to be isomorphic to the  Hopf fibration.
If $\Sigma = S^2$ it follows from the well-known result that $\pi_2(S^2) \cong \bZ$ and $\pi_2(S^3) =0$
 that not every map $\bar{g}:\Sigma = S^2 \to S^2$ admits a lift w.r.t. $\pi_{G/T}$ to a map
 $\Omega:\Sigma = S^2 \to S^3$.
 On the other hand,   the restriction of $\bar{g}$ to the (contractible) subset
 $\Sigma \backslash \{\sigma_0\} =
 S^2 \backslash \{\sigma_0\} \cong \bR^2$ always admits such a lift. (This illustrates
  Remark \ref{rm2.3} in Sec. \ref{subsubsec2.2.4} in the present special case.)

\item Remark \ref{rm2.4} in Sec. \ref{subsubsec2.2.4} implies
that once we have fixed an orientation on $\Sigma$
there is a natural bijection from $[\Sigma,G/T]$ to $I = \ker(\exp_{|\ct})$.
In the special case $\Sigma = S^2$ this is quite plausible
since\footnote{recall footnote \ref{ft11}} $[\Sigma,G/T] = [S^2,G/T] = \pi_2(G/T)$ (as sets)
and   $\pi_2(G/T) \cong \pi_2(S^2) \cong \bZ \cong I$ (as groups).
\end{enumerate}
\end{example}

\subsection{Justification for Remark \ref{rm_loc_constant} in Sec. \ref{subsubsec2.3.1}}
 \label{appB'_3}

Let $B \in \cB$ and set
\begin{equation} \label{eq_appB'_3a0}
T(B):=   \int_{\cA^{\orth}} \prod_i \Tr_{\rho_i}\bigl( \Hol_{l_i}(A^{\orth} + B dt)\bigr)    \exp(i  S_{CS}( A^{\orth} +  B dt))  DA^{\orth}
\end{equation}
We will now show on a heuristic level that if $T(B) \neq 0$
then $B$ must be constant on every connected component
of $\Sigma \backslash (\bigcup_{j=1}^m \arc(l^j_{\Sigma}))$.
In particular, $B$ must be locally constant around $\sigma_0$ (cf. condition \eqref{eq_ass_sigma0} in Sec. \ref{subsubsec2.3.1} above).

\smallskip

 Let $U$ be any open subset of $\Sigma$ which fulfills the following condition
\begin{equation} \label{appB'_2_Udef}
 U \subset \Sigma \backslash (\bigcup_{j=1}^m \arc(l^j_{\Sigma}))
 \end{equation}
 Let  $a_c$ be an arbitrary element of $\cA_{\Sigma,\ct} \cong \cA_c^{\orth} \subset \cA^{\orth}$
 fulfilling $\supp(a_c) \subset U$.
 Then we have
\begin{align} \label{eq_appB'_3_rm}
 T(B) & = \int_{\cA^{\orth}} \prod_i \Tr_{\rho_i}( \Hol_{l_i}(A^{\orth} + Bdt))    \exp(i  S_{CS}( A^{\orth}  + B dt))  DA^{\orth} \nonumber \\
  & \overset{(+)}{=} \int_{\cA^{\orth}} \prod_i \Tr_{\rho_i}( \Hol_{l_i}(A^{\orth} + a_c + Bdt  ))    \exp(i  S_{CS}( A^{\orth}  + Bdt ))  DA^{\orth} \nonumber \\
   & \overset{(++)}{=} \int_{\cA^{\orth}} \prod_i \Tr_{\rho_i}( \Hol_{l_i}(A^{\orth}  + Bdt ))    \exp(i  S_{CS}( A^{\orth}  - a_c + Bdt ))  DA^{\orth} \nonumber \\
     & = \biggl[ \int_{\cA^{\orth}} \prod_i \Tr_{\rho_i}( \Hol_{l_i}(A^{\orth}  + Bdt  ))    \exp(i  S_{CS}( A^{\orth} +  Bdt ))  DA^{\orth} \biggr] \exp(i S_{CS}(- a_c + B dt)) \nonumber \\
   & = T(B)  \exp\biggl(- i  2 \pi k \int_{\Sigma} \Tr\bigl(a_c \wedge  dB\bigr) \biggr)
\end{align}
In step  $(+)$ we exploited the support properties of $a_c$
and in step  $(++)$  we performed the change of variable $ A^{\orth} + a_c \to  A^{\orth}$.\par

If $T(B) \neq 0$ then we  obtain
$\exp\bigl(- i  2 \pi k \int_{\Sigma} \Tr\bigl(a_c \wedge  dB\bigr) \bigr)=1$.
Since this holds for all $a_c \in \cA_{\Sigma,\ct}$ with $\supp(a_c) \subset U$
we can conclude that $dB \equiv 0$ on $U$.
 Since  $U$ was an arbitrary open subset of $\Sigma$ fulfilling  \eqref{appB'_2_Udef}
 we can conclude at a heuristic level that
if $T(B) \neq 0$ then $B$ must indeed be constant on every connected component
of $\Sigma \backslash (\bigcup_{j=1}^m \arc(l^j_{\Sigma}))$.

\begin{remark} \label{rm_appB'_3} \rm
According to what we just said only  ``step functions''
will contribute to the integral $\int \cdots DB$ appearing in Eq. \eqref{eq2.41simpl} in Sec. \ref{subsubsec2.3.1}. This is good news because this means that -- when evaluating the Wilson loop observables $\WLO(L)$ -- we can
expect  to obtain ``sum over area coloring''-expressions
even in the case of  general links $L$,
and not only for the special links $L$ appearing in Theorem \ref{main_theorem} above.\par

On the other hand there  is an obvious complication. Recall that $B$ was supposed to be an
element of  $\cB = C^{\infty}(\Sigma,\ct)$
and the only elements of $\cB$ that have the step function property just mentioned will
be the constant maps.
This strongly suggests\footnote{recall that at the end of Sec. \ref{subsec6.1} above we
give several additional arguments supporting this point of view}  that instead of working with links $L$ consisting of genuine loops $l_1, \ldots, l_m$ we should rather be working with links $L$ consisting
of  ``closed ribbons'' (cf.  Definition \ref{def_cont.ribbon} in Sec. \ref{subsec4.0.3b} above).
And this is of course exactly what we were doing in Sec. \ref{sec4} and Sec. \ref{sec6} of the present paper
where closed simplicial ribbons play a key role.\par

Let us mention that the ribbon analogues of the expressions
$\Tr_{\rho_i}\bigl(\Hol_{l_i}(A)\bigr)$ appearing in Eq. \eqref{eq2.4} above
are not  gauge-invariant functions (cf. also point (C2) in part \ref{appG} of the Appendix below).
So in order to derive the ribbon analogues of the formulas in Sec. \ref{subsec2.3} and
Appendix \ref{appB'_2} and  \ref{appB'_3} we will have to find suitable gauge-invariant versions
of theses expressions $\Tr_{\rho_i}\bigl(\Hol_{l_i}(A)\bigr)$.
A quick  way to do this is to simply
reverse engineer\footnote{Using  the notation of Sec. \ref{subsubsec2.2.2} above
and considering for simplicity the case of a proper gauge fixing (instead of an abstract gauge fixing)
this ``reverse engineering'' procedure amounts to the following:
If $V$ is a gauge fixing subspace of $\cA$ and $\chi_V: V \to \bC$ any function then
by setting $\chi:= \chi_V \circ \Pi_V^{-1} \circ \pi_{\G}$
we obtain a gauge invariant function $\chi$ on $\cA$ which extends $\chi_V$ }
 the desired gauge-invariant functions from
 the ribbon analogues of the restrictions of $\Tr_{\rho_i}\bigl(\Hol_{l_i}(A)\bigr)$
 onto $\cA^{qax}(T)$.
 \end{remark}

\subsection{Careful derivation of Eq. \eqref{eq2.41simpl} in Sec. \ref{subsubsec2.3.1}}
 \label{appB'_2}

Let $U \subset \Sigma$ be an open neighborhood of $\sigma_0$  which is sufficiently small
not to intersect any of the $\arc(l^j_{\Sigma})$, $j \le m$ (recall condition \eqref{eq_ass_sigma0} in Sec. \ref{subsubsec2.3.1} above).
 In other words we assume again that
condition \eqref{appB'_2_Udef} above is fulfilled.
Let $\eta:\Sigma \to [0,1]$ be smooth and fulfill
$$\eta \equiv 1 \quad \text{ on $V$} \quad \quad \text{ and } \quad \quad  \eta \equiv 0 \quad \text{on $\Sigma \backslash U$}$$
where $V$ is another  open neighborhood of $\sigma_0$  such that the closure $\bar{V}$ is contained in $U$.\par

From the definition $ \overline{\Tr_{\rho_i}( \Hol_{l_i}\bigr)}$ in Sec. \ref{subsubsec2.3.1} above
 and the properties of $\eta$ it follows that
for all $A^{\orth} \in \cA^{\orth}$, $B \in \cB$, and $\cl \in [\Sigma,G/T]$ we have
\begin{equation}  \overline{\Tr_{\rho_i}( \Hol_{l_i}\bigr)}\bigl(A^{\orth}+ A_{\sing}(\cl) + Bdt\bigr)  \\
= \overline{\Tr_{\rho_i}( \Hol_{l_i}\bigr)}\bigl(A^{\orth}+ (1-\eta) A_{\sing}(\cl) + Bdt\bigr)
\end{equation}
Observe also that we can consider
 $(1-\eta) A_{\sing}(\cl) \in \cA_{\Sigma \backslash \{\sigma_0\},\ct}$ in a natural way as an element of  $\cA_{\Sigma,\ct} \subset \cA^{\orth}$ (by trivially extending $(1-\eta) A_{\sing}(\cl)$ in the point $\sigma_0$).

\smallskip

Accordingly, we obtain for fixed $B \in \cB$ and $\cl \in [\Sigma,G/T]$
\begin{align}  \label{eq_appB'_2a}
& \int_{\cA^{\orth}} \prod_i  \overline{\Tr_{\rho_i}( \Hol_{l_i}\bigr)}\bigl(A^{\orth}+ A_{\sing}(\cl) + Bdt\bigr)    \exp(i  \overline{S^{qax}_{CS}}( A^{\orth} + A_{\sing}(\cl) + B dt))  DA^{\orth} \nonumber \\
& =  \int_{\cA^{\orth}} \prod_i  \overline{\Tr_{\rho_i}( \Hol_{l_i}\bigr)}\bigl(A^{\orth}+ (1-\eta) A_{\sing}(\cl) + Bdt\bigr)    \exp(i  \overline{S^{qax}_{CS}}( A^{\orth} + A_{\sing}(\cl) + B dt))  DA^{\orth}  \nonumber \\
& \overset{(*)}{=}  \int_{\cA^{\orth}} \prod_i  \overline{\Tr_{\rho_i}( \Hol_{l_i}\bigr)}\bigl(A^{\orth} + Bdt\bigr)    \exp(i  \overline{S^{qax}_{CS}}( A^{\orth} - (1-\eta) A_{\sing}(\cl) +  A_{\sing}(\cl) + B dt))  DA^{\orth}  \nonumber \\
& = \biggl[ \int_{\cA^{\orth}} \prod_i  \overline{\Tr_{\rho_i}( \Hol_{l_i}\bigr)}\bigl(A^{\orth} + Bdt\bigr)    \exp(i  \overline{S^{qax}_{CS}}( A^{\orth} +B dt))  DA^{\orth} \biggr] \times \exp(i  \overline{S^{qax}_{CS}}( \eta A_{\sing}(\cl) + B dt)) \nonumber \\
 & =  \biggl[ \int_{\cA^{\orth}} \prod_i  \Tr_{\rho_i}\bigl( \Hol_{l_i}(A^{\orth} + Bdt)\bigr)    \exp(i  S_{CS}( A^{\orth} +B dt))  DA^{\orth} \biggr]  \times \nonumber\\
  & \hspace{8cm} \times  \exp\biggl(i  2 \pi k \int_{\Sigma \backslash
\{\sigma_0\}} \! \! \Tr\bigl(d (\eta A_{\sing}(\cl)) \cdot B\bigr)\biggr)
\end{align}
where in step $(*)$ we applied
 the informal change of variable $A^{\orth} + (1-\eta) A_{\sing}(\cl) \to A^{\orth}$
(which is now justified since  $(1-\eta) A_{\sing}(\cl)$ is an element of  $\cA^{\orth}$)
and  in the last step we used the definition of $ \overline{S^{qax}_{CS}}$.

\smallskip

Clearly, it is enough\footnote{in view of Eq. \eqref{eq_appB'_2a} above and the definition of $T(B)$
it is clear that for those $B$ for which $T(B)$ does vanish
Eq. \eqref{eq_appB'_2d} below will be trivially fulfilled}
to restrict oneself  to those $B$ for which the integral $T(B)$ in \eqref{eq_appB'_3a0} above does not vanish. According to the heuristic argument in Sec. \ref{appB'_3} above
all such $B$ will be locally constant around $\sigma_0$.
Using this  we obtain
\begin{align} \label{eq_appB'_2b} \int_{\Sigma \backslash
\{\sigma_0\}} \Tr\bigl(d (\eta A_{\sing}(\cl)) \cdot B\bigr)
& =  \int_{\Sigma \backslash \{\sigma_0\}} \Tr\bigl(d (\eta A_{\sing}(\cl) \cdot B)\bigr)
   + \int_{\Sigma \backslash \{\sigma_0\}} \Tr (\eta A_{\sing}(\cl) \wedge dB)  \nonumber \\
 & =  \int_{\Sigma \backslash \{\sigma_0\}} \Tr\bigl(d ( A_{\sing}(\cl) \cdot \eta B)\bigr)
   + \int_{\Sigma \backslash \{\sigma_0\}} \Tr ( A_{\sing}(\cl) \wedge \eta dB) \nonumber \\
& \overset{(+)}{=}  \Tr\bigl( n(\cl) \cdot (\eta B)(\sigma_0))
   + \int_{\Sigma \backslash \{\sigma_0\}} \Tr ( A_{\sing}(\cl) \wedge \eta dB) \nonumber \\
& \overset{(++)}{=}  \Tr\bigl( n(\cl) \cdot B(\sigma_0))
   + \int_{\Sigma \backslash \{\sigma_0\}} \Tr ( A_{\sing}(\cl) \wedge \eta dB)
\end{align}
where in step $(+)$ we used the same
 argument as in step $(*)$ in Eq. \eqref{eq_identity_ncl} in Sec. \ref{subsubsec2.3.1} above
(taking into account that both $B$ and $\eta$ are locally constant around $\sigma_0$).
Step $(++)$ follows because $\eta(\sigma_0)=1$.
Observe that Eq. \eqref{eq_appB'_2a} holds for all $U$ and  $\eta$ fulfilling the assumptions made above.
Since $B$ is locally constant around $\sigma_0$
 the last integral in Eq. \eqref{eq_appB'_2b} vanishes if $U$ was chosen small enough.
Taking this into account we get from Eq. \eqref{eq_appB'_2a} and Eq. \eqref{eq_appB'_2b}
\begin{multline} \label{eq_appB'_2d}
 \int_{\cA^{\orth}} \prod_i  \overline{\Tr_{\rho_i}( \Hol_{l_i}\bigr)}\bigl(A^{\orth}+ A_{\sing}(\cl) + Bdt\bigr)    \exp(i  \overline{S^{qax}_{CS}}( A^{\orth} + A_{\sing}(\cl) + B dt))  DA^{\orth} \\
  =  \biggl[    \int_{\cA^{\orth}} \prod_i \Tr_{\rho_i}\bigl( \Hol_{l_i}(A^{\orth} + B dt)\bigr)    \exp(i  S_{CS}( A^{\orth} +  B dt))  DA^{\orth} \biggr] \exp\bigl( 2 \pi i k\Tr\bigl( n(\cl) \cdot B(\sigma_0))\bigr)
\end{multline}
Using Eq. \eqref{eq_appB'_2d} in  Eq. \eqref{eq2.41_pre}
 for each  $B \in \cB$ and $\cl \in [\Sigma,G/T]$
 we arrive at Eq. \eqref{eq2.41simpl}.

\section{Appendix: Polyhedral cell complexes}
\label{appF}

Talking informally, a
``polyhedral cell-complex''  is just a cell-complex
which is obtained by glueing together ``convex
polytopes''   in an analogous way as simplicial complexes arise from glueing
together simplices. For the convenience of the reader we will now state the formal definitions
and introduce some additional notation which we use in the main text of the present paper.

\begin{definition} \label{def_C1}
\begin{enumerate}
\item  Let $V$ be a finite-dimensional real vector space.
A  convex polytope\footnote{or, more precisely,
 a closed bounded convex polytope} in $V$
 is a non-empty bounded subset $P$ of $V$
 which is   of the form
 \begin{equation} \label{eq_convPoly} P = \bigcap_{a \in A} H_{a}
 \end{equation}
  where  $(H_{a})_{a \in A}$ is a finite family of
 closed halfspaces
 of $V$. Here with ``closed halfspace'' we mean  a  subset $H$ of $V$ of the form
 $H = \{x \in V \mid l(x) \ge 0 \}$ for some non-trivial
 linear form $l: V \to \bR$.

\item Let $V$ and $P$ be as above,  let  $(H_{a})_{a \in A}$ be a
family of halfspaces such that \eqref{eq_convPoly} holds,
and let $h_a$ be  the hyperplane bounding $H_a$, for  $a \in A$.  \par

  A ``face'' of $P$ is a non-empty proper subset $S$ of $P$ of the form
   $S =  P \cap \bigcap_{a \in A'}  h_a$    where $A' \subset A$.
  Observe that  each face $S$ of $P$
  is again a convex polytope in $V$.

\item An abstract convex polytope (or simply, a ``convex polytope'') is
a pair $(P,V)$ where $V$ is a finite-dimensional real vector space
and $P$ a convex polytope in $V$ (equipped with the topology inherited from the standard
topology of $V$). Instead of $(P,V)$ we will often write simply $P$.

\item The dimension $\dim(P)$ of an abstract convex polytope $P=(P,V)$
is the dimension of the linear span $V_P$ of the subset $P - x$ in $V$ where $x$ in an arbitrary  point in $P$.\par

An orientation on $P=(P,V)$ is a non-vanishing element of  $\Lambda^d V_P$ where $d:= \dim(V_P)$.

 \item We will equip every abstract convex polytope $P=(P,V)$ with the obvious topology.
 Clearly, after doing so each $P$ will  be a a closed $n$-cell, i.e. homoemorphic to
an $n$-dimensional closed ball, where $n := \dim(P)$.
  \end{enumerate}
   \end{definition}

\begin{definition} \label{def_C2}
A ``polyhedral cell decomposition'' $\cC$ of a topological space $X$
 is a family $\cC = ((P_a,\Phi_a))_{a \in A}$
 where each $P_a = (V_a,P_a)$ is an (abstract) convex polytope
 and $\Phi_a$ is an embedding $P_a \to X$
 such that the following conditions are fulfilled:
 \begin{enumerate}
\item Each point $x \in X$ lies in exactly one of the sets
  $\Phi_a(\overset{\circ}{P}_a)$, $a \in A$.\footnote{in the terminology of Definition \ref{def_C4} below,  condition i)  just says that  $X$ is the disjoint union of all the ``open cells'' of $\cC$}
  Here  $\overset{\circ}{P}_a$
  denotes the interior of $P_a$ in $V_a$.

\item   For each face $F$ of $P_a$, $a \in A$,
 the restriction $(\Phi_a)_{| F}$ of  $\Phi_a$ onto $F$
 can be ``identified'' with one of the $\Phi_b$, $b \in A$.
 More precisely, we have  $(\Phi_a)_{| F} \circ \psi_{| P_b} = \Phi_b$
  for a suitably chosen $b \in A$ and  a suitably chosen affine injective map $\psi:V_b \to V_a$
with the property $\psi(P_b) = F$.

\item A set $S \subset X$ is open in $X$ iff $\Phi_a^{-1}(S)$ is open
   in $P_a$ for all $a$.
 \end{enumerate}
 We call $\cC$ finite iff the set $A$ is finite.
 The dimension of $\cC$ is the supremum of the set $\{\dim(P_a) \mid a \in A \} \subset \bN_0 \cup \{\infty\}$.
  We call $\cC$ oriented iff every $P_a$, $a \in A$, is equipped with an orientation.
  In the special case where $X$ is a smooth manifold we call $\cC$ smooth iff each of the maps
  $\Phi_a$, $a \in A$, is smooth.
\end{definition}

\begin{example}
  Fig. \ref{Fig. 1a} shows a polyhedral cell decomposition of a closed square
and Fig. \ref{Fig. 1b} shows a polyhedral cell decomposition
of $S^2$.
\begin{figure}[h]
  \centering
  \begin{minipage}[b]{4 cm}
   \includegraphics[height=1.5in,width=1.5in, angle=0]{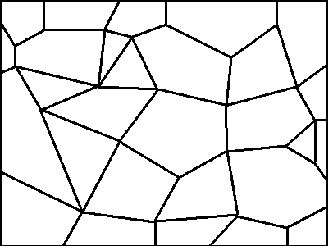}
    \caption{}
    \label{Fig. 1a}
  \end{minipage}
   \begin{minipage}[b]{3.5 cm}
   \ \
      \end{minipage}
  \begin{minipage}[b]{5.6 cm}
    \includegraphics[height=1.5in,width=2in, angle=0]{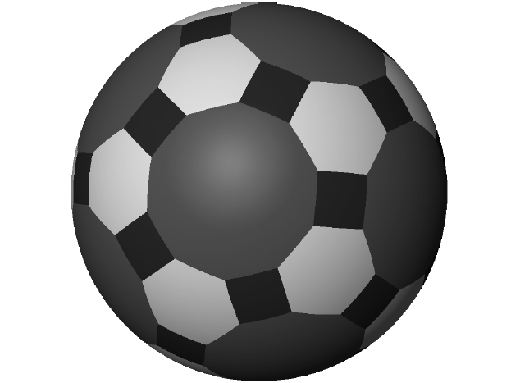}
    \caption{}
    \label{Fig. 1b}
  \end{minipage}
\end{figure}
\end{example}

\begin{definition} \label{def_C3}
A ``polyhedral cell complex''
is a pair $\cP = (X, \cC)$ where $X$ is a topological (Hausdorff)  space
and $\cC$ is a  ``polyhedral cell decomposition'' of $X$. \par
We call $\cP$ finite (resp. oriented)  iff $\cC$ is finite (resp. oriented).
The dimension of $\cP$ is the dimension of $\cC$.
In the special case where $X$ is a smooth manifold we call
$\cP$  smooth iff $\cC$ is smooth.
\end{definition}

\begin{remark} \label{rm_C3}  \rm
\begin{enumerate}

\item  Clearly, the notions   ``polyhedral cell decomposition'' and ``polyhedral cell complex''
   generalize the notions ``triangulation'' and ``simplicial complex''.

\item A polyhedral cell complex is a regular CW-complex where\par
  1. Each of the  closed cells involved   has some extra structure  and,\par
   2. The way in which these closed cells are glued together respects
    this extra structure (as explained in condition ii) in Definition \ref{def_C2}  above).

\end{enumerate}
\end{remark}

\begin{definition} \label{def_C4}  Let $\cC = ((P_a,\Phi_a))_{a \in A}$
 be a polyhedral cell decomposition of $X$ and $p \in \bN_0$.
\begin{enumerate}
\item We set $\face_p(\cC):= \{ \Phi_a({P}_a) \mid a \in A \text{ with } \dim(P_a) = p\}$.
The elements of $\face_p(\cC)$
are called the ``p-faces''  of $\cC$.

\item We set  $\cell_p(\cC):= \{ \Phi_a(\overset{\circ}{P}_a) \mid  a \in A \text{ with } \dim(P_a) = p\}$.
The elements of  $\cell(\cC) := \bigcup_p \cell_p(\cC)$ are called
the ``open cells'' of $\cC$.

\item  If $\cP = (X, \cC)$ is polyhedral cell complex
we write $\face_p(\cP)$ instead of $\face_p(\cC)$.
\end{enumerate}
\end{definition}

\begin{remark} \label{rm_C4}  \rm There is a
 natural 1-1-correspondence
  between the elements of $\cell_p(\cC)$ and of $\face_p(\cC)$:
    $ E \in \cell_p(\cC)$ corresponds to $F \in \face_p(\cC)$
  iff $F$ is the closure of $E$ in $X$
\end{remark}

Let us now fix a surface $\Sigma$ and
a polyhedral cell decomposition $\cC$ of $\Sigma$.

\begin{definition} \label{def_C5}
Let $\cC'$ be  another polyhedral cell decomposition  of $\Sigma$.
\begin{enumerate}
\item We say that $\cC'$ and $\cC$ are dual to each other iff the following three conditions are fulfilled
\begin{itemize}
\item There is a  bijection $\psi_0:  \cell_0(\cC) \to  \cell_2(\cC')$
such that $x \in \psi_0(x)$ for all $x \in  \cell_0(\cC)$.

\item There is a bijection $\psi_1:  \cell_1(\cC) \to  \cell_1(\cC')$
such that  each $e \in \cell_1(\cC)$
intersects $\psi_1(e) \in  \cell_1(\cC')$ in exactly one point
and $e$ intersects none of the other elements of $ \cell_1(\cC')$.

  \item There is a bijection $\psi_2:  \cell_2(\cC) \to  \cell_0(\cC')$
such that $x' \in \psi_2^{-1}(x')$ for all $x' \in  \cell_0(\cC')$.
\end{itemize}
Observe that each of the bijections $\psi_p:\cell_p(\cC) \to \cell_{2-p}(\cC')$
 is necessarily unique
   and  induces a bijection
 $\bar{\psi}_p:\face_p(\cC) \to \face_{2-p}(\cC')$ (via the 1-1-correspondence in Remark \ref{rm_C4})

\item If $\cC$ and $\cC'$ are dual to each other
 we set $\Check{F}:= \bar{\psi}_p(F) \in \face_{2-p}(\cC')$  for $F \in \face_p(\cC)$
 and $\Check{F}':= \bar{\psi}^{-1}_p(F') \in \face_{2-p}(\cC)$ for $F' \in \face_p(\cC')$.\par
   We call  call $\Check{F}$ (resp. $\Check{F}'$)  the face ``dual'' to $F$ (resp.  ``dual'' to $F'$),
 cf. Example \ref{ex_dualfaces} in Sec. \ref{subsec4.0.5} above.

\item  If $\cC$ and $\cC'$ are dual to each other
 then for each $e \in \face_1(\cC)$ (or $e \in \face_1(\cC')$)
  the unique intersection point
of $e$ and the dual face $\Check{e}$
      will be denoted by\footnote{the notation $\bar{e}$  is motivated
 by the fact that in the special case where $\cC'$ is the canonical dual of $\cC$ as defined below, $\bar{e}$  will just be the ``barycenter'' of  $e$}
 $\bar{e}$.
\end{enumerate}
\end{definition}

In the set of polyhedral cell decomposition which are dual to $\cC$
there is a distinguished element $\cC'$, which we will call  ``the canonical dual of $\cC$'', or simply,  ``the'' dual of $\cC$, cf.  Fig. \ref{dual_faces} in Sec. \ref{subsec4.0.5} above for an example.\par

The canonical dual $\cC'$ is constructed in a two step process
where one first ``breaks down'' $\cC$ into smaller pieces,
which are given by the so-called  ``barycentric subdivision'' of $\cC$,
and then, in a second step, reassembles these pieces in a suitable way.
The   explicit\footnote{strictly speaking we would also have to describe
 the maps $\Phi'_{a'}$, $a' \in A'$, explicitly.
 Since this is both technical and straightforward
 we omit this here} description of the dual $\cC'$  is as follows:
\begin{itemize}
\item  $\face_0(\cC') := \{ \bar{F} \mid F \in \face_2(\cC)\}$
       where $\bar{F}$ is the ``barycenter'' of $F$

\item $\face_1(\cC') := \{[x_+(e),\bar{e}] \cup [\bar{e},x_-(e)]    \mid e \in \face_1(\cC) \}$
where, for each $e \in \face_1(\cC)$,
$\bar{e}$ is the ``barycenter'' of $e$
and $x_{+}(e)$ and $x_{-}(e)$ are the barycenters of the two
2-faces $F_{+}(e), F_{-}(e) \in \face_2(\cC)$
which bound $e$.

\item The elements of $\face_2(\cC')$ are the closures
  of the connected components of $\Sigma \backslash (\face_0(\cC') \cup \face_1(\cC'))$
  \end{itemize}

Above we have used the following convention:
\begin{convention} \rm \label{conv_identif} Let $X$ and $\cC = ((P_a,\Phi_a))_{a \in A}$
be as above,  let  $F \in \face_p(\cK)$, $p \in \bN_0$
and let $b \in A$ be given by $F= \Phi_b(P_b)$.
We then usually  identify  $F$ with the convex polytope $P_b$
via the homeomorphism $\Phi_b$.
Clearly, after doing so,  the notions of the  ``barycenter'' $\bar{F}$ of $F$
 and the ``convex hull'' $[x,y]$ of two points
 $x, y, \in F$ are then well-defined.
 \end{convention}

\section{Appendix: The simplicial program revisited}
\label{appG}

In Sec. \ref{sec3} above we described the general ``simplicial program'' for   CS-theory/$BF_3$-theory.
The simplicial program should be considered as a long term project.
One possible strategy for attacking this project is to divide it into suitable sub projects.
Here are three natural sub projects:

\begin{description}

\item[Project 1] For non-Abelian $CS$-theory/$BF_3$-theory  on the special base manifold $M = \Sigma \times S^1$ find a simplicial realization of the path integral expressions   in the torus gauge
    for the WLOs associated to general links.

 \smallskip

Clearly, Project 1 is exactly what we are dealing with
in   the present paper and in  \cite{Ha7b}.
 Theorem \ref{main_theorem} above
 can be seen as a first step  towards the completion of Project 1.

 \smallskip

 {\em Comment:} At least in the case of $BF_3$-theory
  the chances for completing Project 1 successfully should be fairly good,
cf.  Sec. \ref{sec7} above and Sec. 6 and Sec. 7 in \cite{Ha7b}.

\item[Project 2] Generalize the results obtained in Project 1  to a larger class of  base manifolds $M$,
for example for manifolds $M$ which are (non-trivial) $S^1$-bundles or which can be obtained from $\Sigma \times
S^1$ by  surgery operations on suitable links.

\smallskip

{\em Comment:} At least in the case $M=S^3$ the chances are very good since in this case it should be possible to implement Witten's surgery argument rigorously. Also the chances for dealing successfully with base manifolds $M$ which  are (non-trivial) $S^1$-bundles should be quite good (cf. Problem (P3) in Sec. \ref{sec7} above).

\item[Project 3] Find a  simplicial definition
 of the WLOs associated to general links  for the {\it non-gauge fixed}
 non-Abelian $CS$-theory/$BF_3$-theory  on $M$ (where $M$ is as general as possible).

\smallskip

Observe that there is a natural ``discrete'' analogue of the torus gauge
fixing procedure. So if one can complete Project 1 successfully there might be a  quick
 way to complete also Project 3 in the special case $M = \Sigma \times S^1$.
 In order to do so we could  look at the simplicial definition of the WLOs used in Project 1
and then try to ``reverse engineer'' from it a non-gauge fixed ``version'',
i.e. a suitable simplicial expression which -- after applying ``discrete'' torus gauge fixing --
leads to the simplicial expression used in Project 1.

\smallskip

{\em Comment:} There are several serious ``obstructions'' or ``complications''
for Project 3 even in the special case $M = \Sigma \times S^1$,  see the list below. This is why I am rather sceptical regarding the  successful implementation of Project 3.

\end{description}

We end the present section with a list of the ``complications''
 regarding Project 3, which we mentioned above\footnote{for simplicity we
 have formulated these complications within the CS-theoretic setting relevant in the present paper.
 It is straightforward to rewrite (C1)--(C3) within the $BF_3$-theoretic setting of Sec. 7 in \cite{Ha7b}}:

\begin{description}

\item[Complication (C1)] Recall the definition of $\Det^{disc}_{FP}(B)$ in Eq. \eqref{eq_ Det_disc_FP1}

 We had
$$ \Det^{disc}_{FP}(B)  :=
   \prod_{x \in \face_0(q\cK)}  \det\nolimits^{1/2}\bigl(1_{{\ck}}-\exp(\ad(B(x)))_{| {\ck}}\bigr)$$
Discrete torus gauge fixing (cf.  Project 3 above) is related to the map
$q: G/T \times T \to G$ given by $ q(\bar{g},t) = \Ad(g) t$ for $t \in T$ and  $\bar{g} = gT  \in G/T$.
 Since $(\bar{g},t) \mapsto \det\bigl(1_{{\ck}}- \Ad(t)_{|\ck}\bigr)$  is the ``Jacobian''\footnote{more precisely, we have $q^*(\nu_G) = \det\bigl(1_{{\ck}}-\Ad(\pi_2(\cdot))_{| \ck}\bigr) (\pi_1^*(\nu_{G/T}) \wedge
\pi_2^*(\nu_T)))$
where $\nu_G$, $\nu_T$, and $\nu_{G/T}$ are the normalized left-invariant
 volume forms on $G$, $T$, and $G/T$ and $\pi_1:G/T \times T \to G/T$ and $\pi_2:G/T \times T \to T$
the canonical projections}
of the map $q$
it is easy to see  how factors of the form  $\det\bigl(1_{{\ck}}-\exp(\ad(B(x))_{|\ck})=
\det\bigl(1_{{\ck}}-\Ad(\exp(B(x)))_{|\ck}\bigr)$ in Eq. \eqref{eq_ Det_disc_FP1} above
can arise from  discrete torus gauge fixing.\par

However, it is not clear how the $1/2$-exponents in
 Eq. \eqref{eq_ Det_disc_FP1} above  can arise from such a procedure.

\item[Complication (C2)] In the standard formulation of lattice gauge theory (LGT)
the traces of the holonomies associated to simplicial loops
  are gauge-invariant  functions (ie invariant under the operation of the lattice gauge group).
  However, in the present paper
  we used closed simplicial ribbons for reasons explained  at the end of Sec. \ref{subsec6.1} above.
   In contrast to the traces of the holonomies associated to simplicial loops
  the traces of the holonomies associated to closed simplicial ribbons will no longer be
   be gauge-invariant functions.
     In other words: even if we can resolve Complication (C1) above and we can work within
     a non-gauged-fixed LGT setting there will not be a totally natural
  candidate for a ``non-gauge fixed version'' of the expression
    $\Tr_{\rho}\bigl( \Hol^{disc}_{R}(A^{\orth},   B)\bigr)$
   where $\Hol^{disc}_{R}(A^{\orth},   B)$    is as in  Eq. \eqref{eq4.21} above
        (and where $\rho$ is a fixed finite-dimensional complex representation of $G$). \par

  One possible approach for resolving (C2) could be  to use the LGT-analogue of the
    strategy sketched in Remark \ref{rm_appB'_3} at the end of Appendix \ref{appB'_3}
             for resolving a similar complication in the continuum setting.

\item[Complication (C3)] Recall that the sum $\sum_{y \in I} \cdots$ and the
factor $\exp\bigl( - 2\pi i k  \langle y, B(\sigma_0) \rangle \bigr)$ in the continuum
equation Eq. \eqref{eq2.48} above are the result of certain topological obstructions
which arose when applying (continuum)  torus gauge fixing, cf. Sec. \ref{subsubsec2.2.4} above.
On the other hand, for discrete torus gauge fixing
there are no topological obstructions.
Accordingly, it is not clear how -- by applying discrete torus gauge fixing
one can obtain the sum $\sum_{y \in I} \cdots$
and the factors $\exp\bigl( - 2\pi i k  \langle y, B(\sigma_0) \rangle \bigr)$
appearing in Eq. \eqref{eq2.48} above.
\end{description}

\end{document}